\documentclass[dvipdfmx,11pt]{amsart}
\usepackage{amssymb,amsmath,mathtools}
\usepackage{graphicx}
\usepackage{amssymb}
\usepackage{bm}
\usepackage{color} 
\usepackage{epsfig}
\usepackage{subfigure}
\usepackage{overpic}
\usepackage{comment}

\newcommand{\cblue}[1]{\textcolor{black}{#1}}
\newcommand{\ccyan}[1]{\textcolor{black}{#1}}
\newcommand{\cmag}[1]{{\color{black} #1}}

\newcommand{\ep}{\varepsilon}

\usepackage[margin=20mm]{geometry}

\numberwithin{equation}{section}
\numberwithin{figure}{section}

\theoremstyle{plain} 
\newtheorem{theorem}{Theorem}[section] 
 
\newtheorem{proposition}{Proposition}[section] 

\theoremstyle{definition} 
\newtheorem{remark}{Remark}[section]

\begin{document}

\title[A stochastic phase model for cardiac muscle cells]{A stochastic phase model with reflective boundary and induced beating for the cardiac muscle cells}

\author{Guanyu Zhou$^1$, Tatsuya Hayashi$^2$ and  Tetsuji Tokihiro$^3$}

\address{$^1$ Institute of Fundamental and Frontier Sciences, University of Electronic Science and Technology of China \\
$^2$ Graduate School of Information Science and Technology,
Hokkaido University \\
$^3$ The Graduate School of Mathematical Sciences, The University of Tokyo
}
\email{$^1$ wind\_geno@live.com, $^2$ thayashi@ist.hokudai.ac.jp, $^3$ toki@ms.u-tokyo.ac.jp}

\begin{abstract}
We consider the stochastic phase models for the community effect of cardiac muscle cells.
The model is the extension of the stochastic integrate-and-fire model in which we incorporate the irreversibility after beating, induced beating and refractory.
We focus on investigating the expectation and variance of (synchronized) beating interval.
In particular, for the single-isolated cell, we obtain the closed-form expectation and variance of the beating interval, 
and we discover that the coefficient of variance (CV) has upper limit $\sqrt{2/3}$.
For two-coupled cells, we derive the partial differential equations (PDEs) for the expected synchronized beating intervals and the distribution density of phase. 
Moreover, we also consider the conventional Kuramoto model for both two- and $N$-cells models,   
where we establish a new analysis using stochastic calculus to obtain the CV of the ``synchronized'' beating interval, 
and make some improvement to the literature work \cite{Kori}. 
\end{abstract}


\maketitle

\section{Introduction} \label{sec:int}
A cardiac muscle cell (cardiomyocyte) has a distinguishing property among biological cells; 
it generates spontaneous pulsation.
Heartbeat is a macroscopic phenomenon in which pulsations of cardiac muscle cells are tuned to a certain rate.
Since each cell has its own beating rhythm when isolated, there must be a certain mechanism for synchronizing pulsations of cardiac muscle cells.
Extensive works have been devoted to understanding this mechanism both experimentally and  theoretically 
\cite{Abramovich-Sivan, Dehaan_etal, Goshima-Tonomura, Guevara-Lewis, Harary-Farley,  Jogsma_etal, Lovell_etal, Merks-Koolwijk, Michael-Matyas-Jalife, Mitchell-Schaeffer, Petrov_etal, Torre, Yamauchi_etal}. 
Contraction of a cardiac muscle cell is caused by complex electrophysiological processes and detailed analyses
require elaborated mathematical models composed of a huge number of equations \cite{Hatano_etal}.
To understand the essence of synchronization, however, a small number of simultaneous ordinary equations of membrane currents and action potentials, 
such as the Hodgkin-Huxley equation \cite{Hodgkin} or its reduced  forms, the FitzHugh-Nagumo (FN) equation \cite{FitzHugh61, Nagumo62} and the Van der Pol equation (cf. \cite{Keener-Sneid, Murray}), 
are enough to capture the key phenomenon of the cell dynamics. 

\ccyan{The cardiac muscle cells in a tissue are individual entities with identical genetic informations; 
however, these difference of individual cells are ironed out when they becomes clusters or tissues, 
which is called the ``community effect'' of cells as induced uniformity \cite{KanekoT07, KojimaK06}.  
Besides of the individual information (for example, the dynamics of the membrane currents of individual cells), 
to achieve a comprehensive understanding of the cardiomyocytes' dynamics, 
the analysis of the epigenetic information (the community effect) is mandatory. 
Since it is difficult to control the conditions and qualities of cells, 
there exists limitations in the biological experiments to study the community effect. 
To overcome these problem, 
the mathematical modeling is one of the most powerful approaches.   
In the present paper, 
we intend to understand the community effect of cardiomyocytes by proposing and studying the mathematical models, 
which should incorporate the essential properties of the biological system,   
and can somehow reproduce the experimental results \cite{KanekoT07, KojimaK06}.} 

To investigate the community effect of cardiomyocytes, 
we modify the conventional Kuramoto model \cite{Kori, Kuramoto} by incorporating the conceptions of irreversibility of beating, the induced beating and refractory to capture the essential properties of cardiomyocytes' synchronization. 
Our model can be regarded as an modification of the stochastic phase model or the integrate-and-fire model \cite{Chang-Juang, Mirollo-Strogatz}, 
which has been widely used as a spiking neuron model \cite{Burkitt, Keener_etal, Peskin, Sacerdote}. 
We utilize the phase models for two reasons. 
First, from the biological experiments \cite{KanekoT07, KojimaK06}, only the data of beating intervals is available. 
However, the Hodgkin-Huxley, FitzHugh-Nagumo or Van der Pol  equations \cite{Hodgkin, FitzHugh61, Nagumo62, Keener-Sneid, Murray} model the dynamics of membrane currents or ion concentration.  
Without adequate information of potential and ion concentration, 
it is hard to determine the parameters of these equations appropriately. 
Moreover, the application of these models to each cell in $N$-cells network ($N\gg 1$) yields a large number of nonlinear equations,  
which is difficult to dealt with.  
Next, since we mainly focus on investigating the beating intervals, 
one can think of the cardiomyocyte with rhythmic beating as an oscillator.  
Then, the phase models \cite{Kori, Kuramoto, Winfree} are suitable mathematical tools to analyze the oscillation. 
In fact, the stochastic phase model is well applicable to model the distribution of the beating intervals (oscillation periods). 
For instance, in the case of single-isolated cell (single oscillator), 
one can decide the intrinsic frequency and noise strength of the phase model by the experimental data of beating intervals and the formulas \eqref{eq:1-c-t^1-E}, \eqref{eq:1-c-t^1-CV}.   
In addition, one can derive the phase equation from the FN model (see \cite{Kuramoto}). 

\cblue{The main contribution of this paper is summarized in two aspects. 
First, it is an original idea to incorporate the stochastic phase equation with reflective boundary, induced pulsation and refractory, to model the (synchronized) oscillation of cardiomyocytes. 
\cite{Hayashi} compares the simulation of the proposed models with the observation from biological experiments \cite{KanekoT07, KojimaK06}, 
which indicates the well applicability of our models.
The present paper, as a theoretical supplement to \cite{Hayashi}, is only devoted to the theoretical analysis. 
For single-isolated cardiomyocyte, we obtain the explicit relationship between the parameters (intrinsic frequency and noise strength) of the model and the statistic properties (expectation and variance) of the beating interval. 
For two-coupled cells, by the renewal theory and Fokker-Planck equation, 
we derive the PDEs associated with the expectation the synchronized beating intervals and the distribution density of phases. 
Although we cannot obtain the closed-form of the statistic properties, 
the PDEs with non-standard boundary conditions deserve the comprehensive theoretical/numerical analysis from the mathematical points of view.}

\cblue{Second, we also consider the conventional phase model, 
and make several improvement to the existing results \cite{Kori}. 
In particular, we present a rigorous calculation of the coefficient of variance (CV) for both two-  and $N$-cells models using the theories of It\^{o} integral,   
thanks to which, we provide the formulas to determine the proper reaction coefficients of the model for the case of two-coupled cells.} 

The rest of this paper is organized as follows. 
In Section~\ref{sec:2}, we introduce the biological background of our work, 
and explain the connection between the FN model and the phase eqaution. 
Section~\ref{sec:1-c} is devoted to the model with reflective boundary for the single-isolated cell. 
We study the stochastic phase models for two-coupled cells in Section~\ref{sec:2-c}.  
The N-cell network is dealt with in Section~\ref{sec:N-c}. 
The concluding remark is addressed in Section~\ref{sec:C-R}. 
%
%

\section{\ccyan{From the FitzHugh-Nagumo model to the phase model}}
\label{sec:2}
\ccyan{As a preliminary, we briefly introduce the The experimental approach to understand the epigenetic information of cardiomyocytes. 
And then, let us explain the connection between the the FN model and the phase model. 
\subsection{The experimental approach}
The on-chip cellomics technology has been applied to investigating the community effect of cardiomyocytes \cite{KanekoT07, KojimaK06}, 
which, simply speaking, includes three steps:   
(1) The cells are taken from a community/tissue using a nondestructive cell sorting procedure. 
(2) We put the cells in a microchamber (on chip) where we can design the cell network and control the medium environment. 
(3) We measure the beating intervals of each cell on chip by light signal (not the membrane currents).  
The procedure of the bio-experiment is described in Figure~\ref{fig:exp-2} (a) (see \cite{KanekoT07, KojimaK06} for details). }

\ccyan{To analyze the distribution of the beating intervals from experiments (see Figure~\ref{fig:exp-2} (right)), 
one can apply the FN equations with noise \eqref{eq:FN-noise} to model the dynamics of the membrane currents. 
For example, in Figure~\ref{fig:FN-2} (c)(d), we plot the trajectories of FN model with rhythmic action potential influenced by noise.
However, it is nontrivial to determine the suitable parameters for FN model such that the distribution of beating interval generated by simulation (see Figure~\ref{fig:FN-phi-noise} (b)) coincides with the experimental data (Figure~\ref{fig:exp-2} (b)) well. 
To tackle this problem, we regard the cardiomyocyte as a oscillator satisfying the phase model \eqref{eq:phi-0-noise} with  intrinsic frequency $\mu$ and noise strength $\sigma$. 
This two parameters $(\mu, \sigma)$ can be easily determined from the bio-experimental data of the beating interval. 
For above reason, we utilize the phase model instead of FN model. 
In fact, the phase model can be derived from the FN system. }
\begin{figure}
\begin{center}
  \subfigure{%
     \begin{overpic}[width=0.5\linewidth]{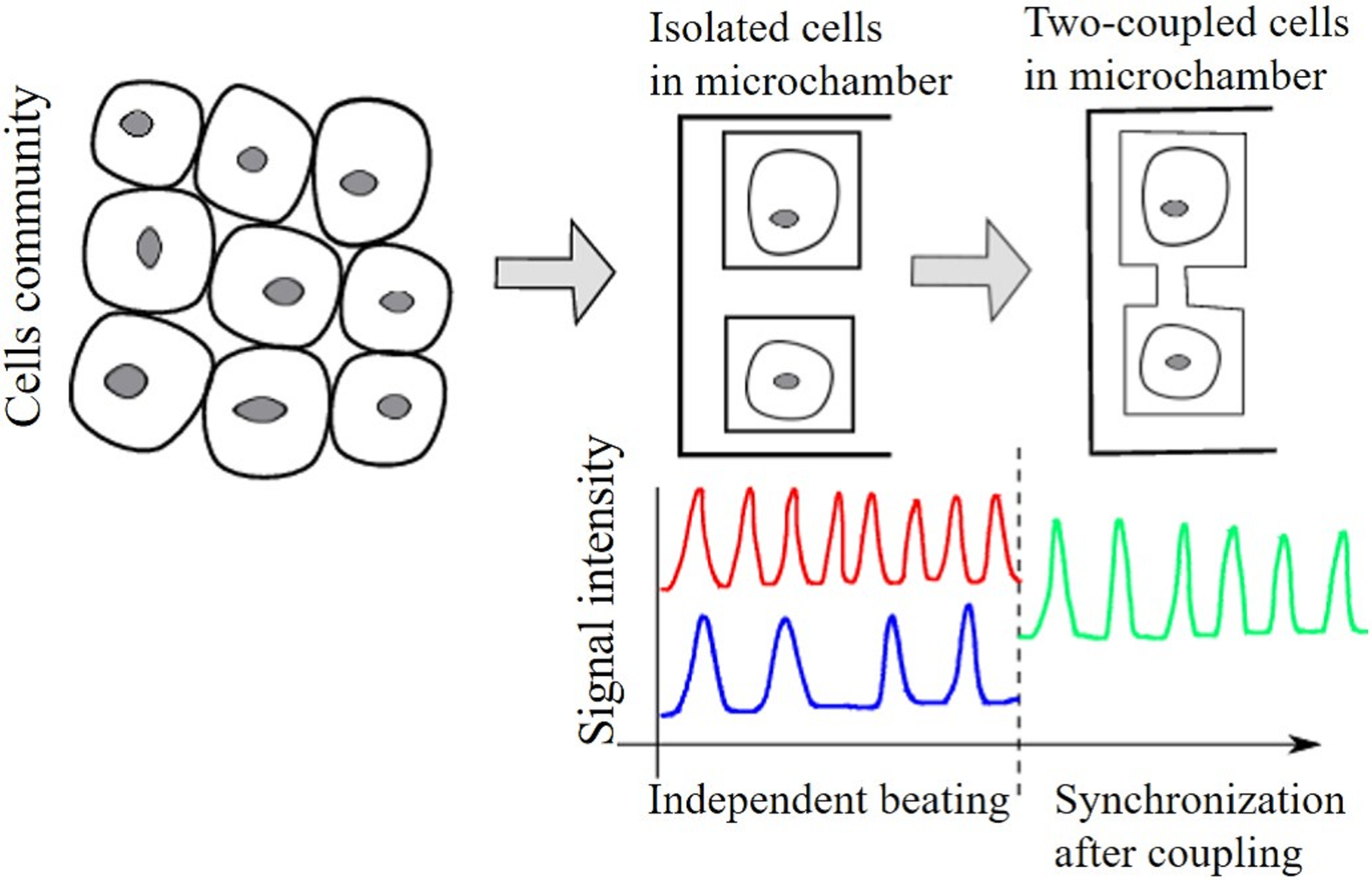}
     \put(0,65){\small \bf (a)} 
     \end{overpic}}%
   \hspace*{0.5cm}
  \subfigure{%
     \begin{overpic}[width=0.4\linewidth]{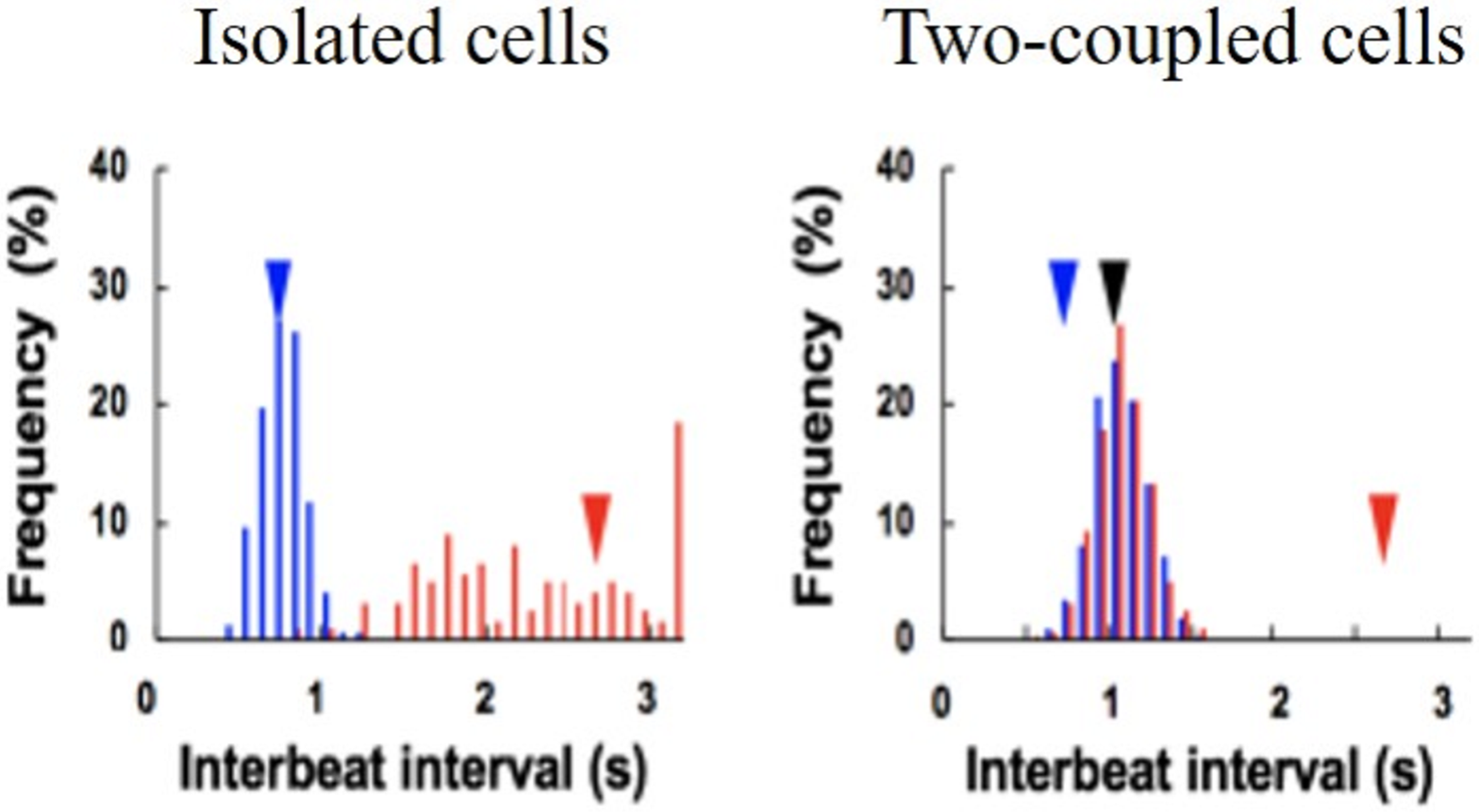}
     \put(0,50){\small \bf (b)}
     \end{overpic}}%
 \end{center}
 \caption{\ccyan{(a) The on-chip cellomics technology. (b) An example of the experimental data of the beating interval of two cardiomyocytes before and after coupling.}}
\label{fig:exp-2}
\end{figure}

\subsection{\ccyan{From the FitzHugh-Nagumo model to the phase model}}\label{sec:2-2}
\ccyan{The FN model has been widely applied to model the membrane current of the spiking neuron or cardiomyocyte, 
which can be regarded as a simplification of the famous Hodgkin-Huxley model. 
First, let us pay attention to the case of the single cell without noise effect,  
the FN model of which is given by:
\begin{subequations}\label{eq:FN}
\begin{align}
& \frac{du}{dt} = u(u-a)(1-u) - w, \\
& \frac{dw}{dt} = \tau(u - bw).
\end{align}
\end{subequations}
Here $u$ denotes the membrane current,
and $(a,b,\tau)$ the parameters ($\tau \ll 1$). 
When $w$ decreases below $0$, $u$ increase instantly, which corresponds to the pulsation of the membrane potential (beating). 
Therefore, one can regard that $w$ is associated to the refractory ($w$ also depends on $u$). 
For $a = -0.1$, $b = 0.5$, $\tau = 0.01$, 
we see that $u$ behaves like a $T$-periodic function (see Figure~\ref{fig:FN-u} (a) with $T \approx 108$). 
One can validate that, for sufficiently large time $t$, 
the trajectory $(u,w)$ tends to a \emph{limit cycle}, 
that is, $(u(t+T),w(t+T)) = (u(t),w(t))$ (see Figure~\ref{fig:FN-u} (b)). 
Hence, one can find a homeomorphism which maps the points $(u(t),w(t))$ on the limit cycle to the phase function $\phi(t)$ given by  
\begin{equation}\label{eq:phi-0}
d\phi(t) = \mu dt,   
\end{equation}
where $\mu = \frac{1}{T}$ denotes the intrinsic frequency. 
Here, $\phi$ is also $T$-periodic if we set $\phi = \phi + k \ (\forall k \in \mathbb{Z})$, 
or equivalently $\phi$ takes value in torus $[0,1)$ (see Figure~\ref{fig:FN-u} (d)), 
which means $\phi$ jumps to $0$ when approaching $1$ ($\phi(t)= 0$ when $\phi(t-):= \lim_{s \uparrow t} \phi(s) = 1$). 
See Figure~\ref{fig:FN-u} (c) for an example of $\phi$. }

\ccyan{Denote by $\bm{x}(t) := (u(t),w(t))$ the trajectories of system \eqref{eq:FN}, 
and by $\bm{\chi}(t):= (u(t),w(t))$ the trajectories of limit cycle, i.e., $\bm{\chi}(t+T) = \bm{\chi}(t)$. 
Then, since $\phi$ is $T$-periodic, 
we can regard $\phi(t)$ as a function of $\bm{\chi}$, i.e., $\phi(t) = \phi(\bm{\chi}(t))$. 
In fact one can extend such $\phi$ to all trajectories $\bm{x}$, namely $\phi(\bm{x}(t))$ (see \cite{Kuramoto} for the detailed argument).}

\ccyan{Setting $\bm{f}(\bm{x}) := [u(u-a)(1-u) - w, \tau(u - bw)]^\top$, 
we find that 
\begin{equation}\label{eq:phi-0-dx}
\frac{d\phi (\bm{x}(t))}{dt} = \frac{\partial \phi}{\partial \bm{x}} \cdot \frac{d \bm{x}}{dt} = \frac{\partial \phi}{\partial \bm{x}} \cdot \bm{f}(\bm{x}). 
\end{equation}
Putting together with \eqref{eq:phi-0},
\begin{equation}\label{eq:phi-0-dx-chi}
\mu = \left. \frac{d\phi (\bm{x}(t))}{dt} \right|_{\bm{x} = \bm{\chi}(\phi)} = \left. \frac{\partial \phi}{\partial \bm{x}} \right|_{\bm{x}=\bm{\chi}(\phi)} \cdot \bm{f}(\bm{\chi}(\phi)). 
\end{equation}
In brief, to model the dynamics of the membrane currents, 
one can apply the FN model involving with two variables $(u,w)$ and parameters $(a,b,\tau)$. 
Meanwhile, to describe the rhythmic oscillation, the phase model with intrinsic frequency $\mu$ (or the beating interval $T$) is sufficient. } 
\begin{figure}
\begin{center}
  \subfigure{%
     \begin{overpic}[width=0.24\linewidth]{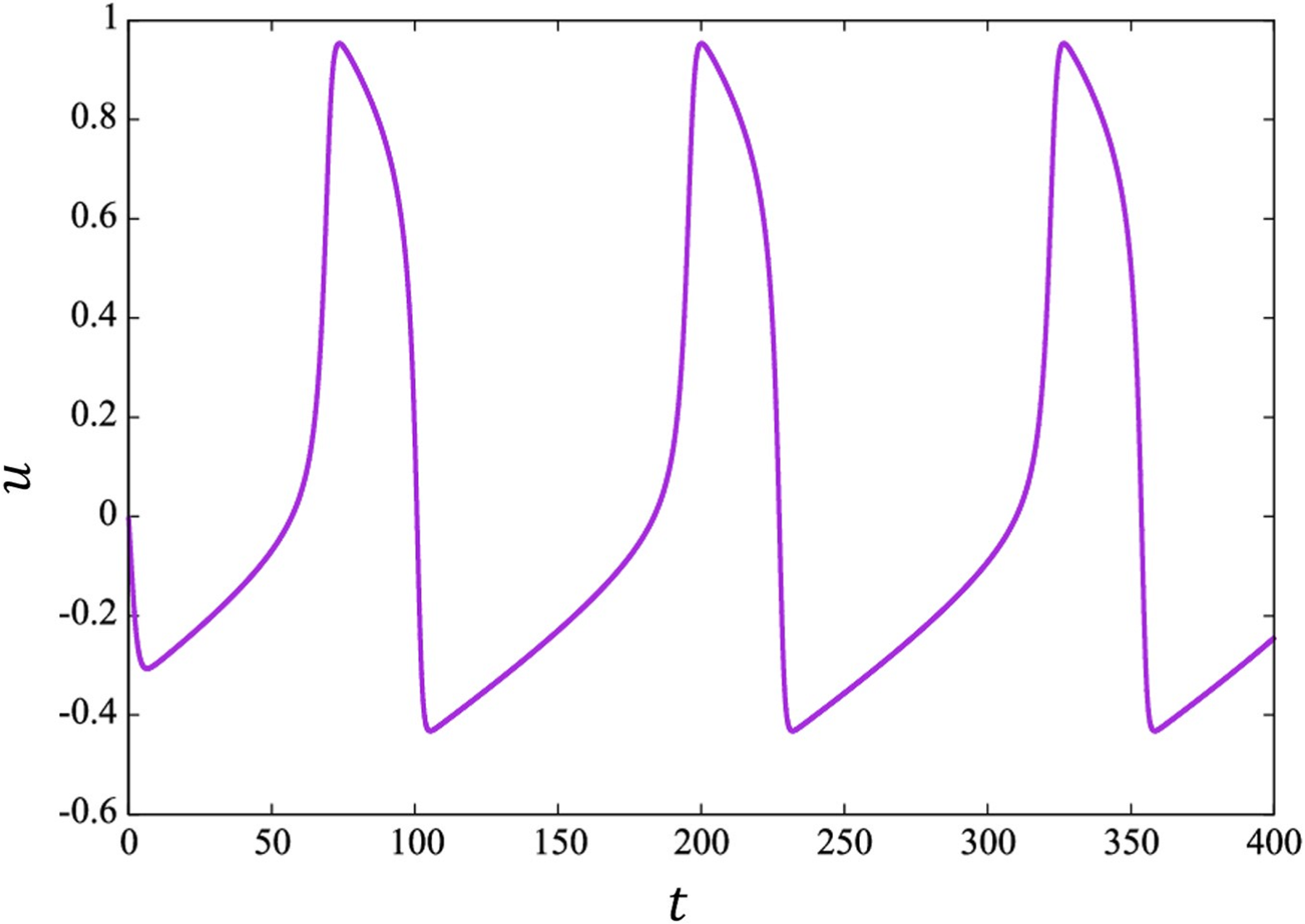}
     \put(-4,70){\tiny \bf (a)}
     \end{overpic}}%
   \hspace*{0.3cm}
  \subfigure{%
     \begin{overpic}[width=0.24\linewidth]{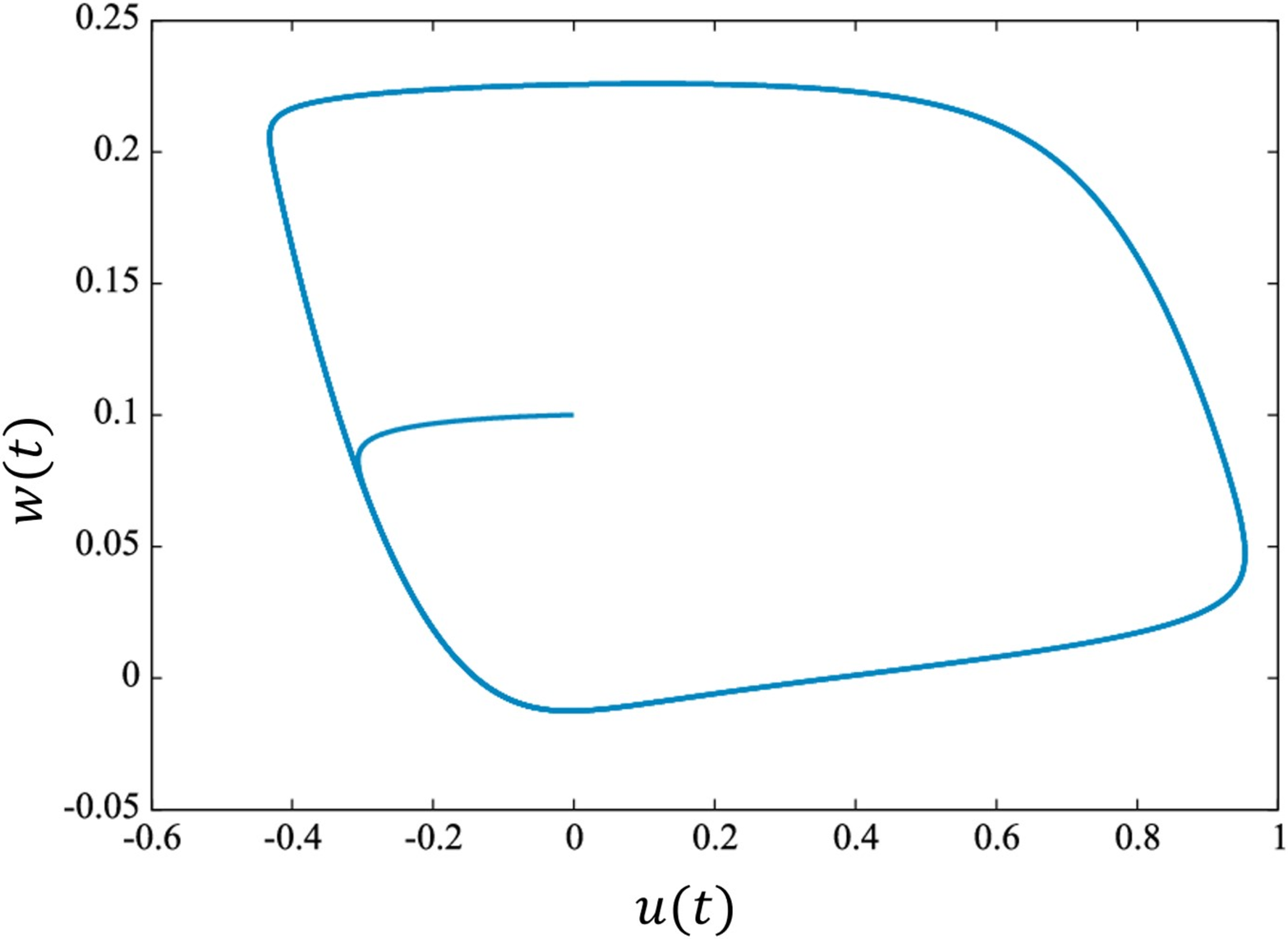}
     \put(-4,70){\tiny \bf (b)}
     \end{overpic}}%
   \hspace*{0.2cm}
  \subfigure{%
     \begin{overpic}[width=0.26\linewidth]{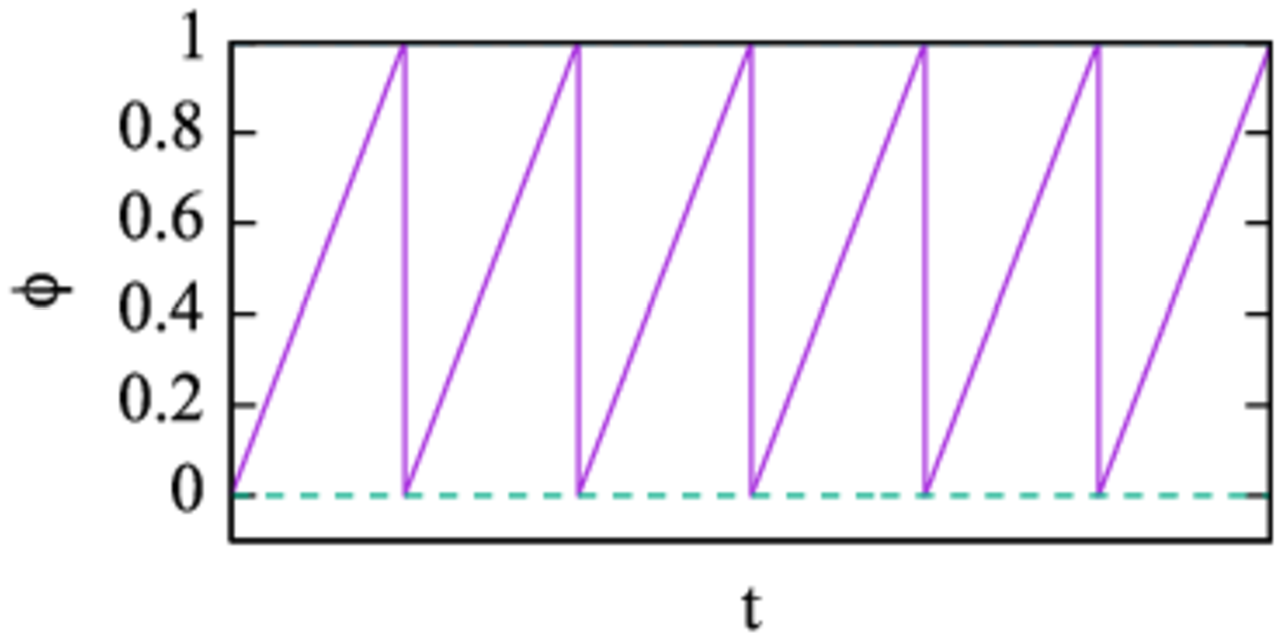}
     \put(10,63){\tiny \bf (c)}
     \end{overpic}}%
   \hspace*{0.3cm}
  \subfigure{%
     \begin{overpic}[width=0.13\linewidth]{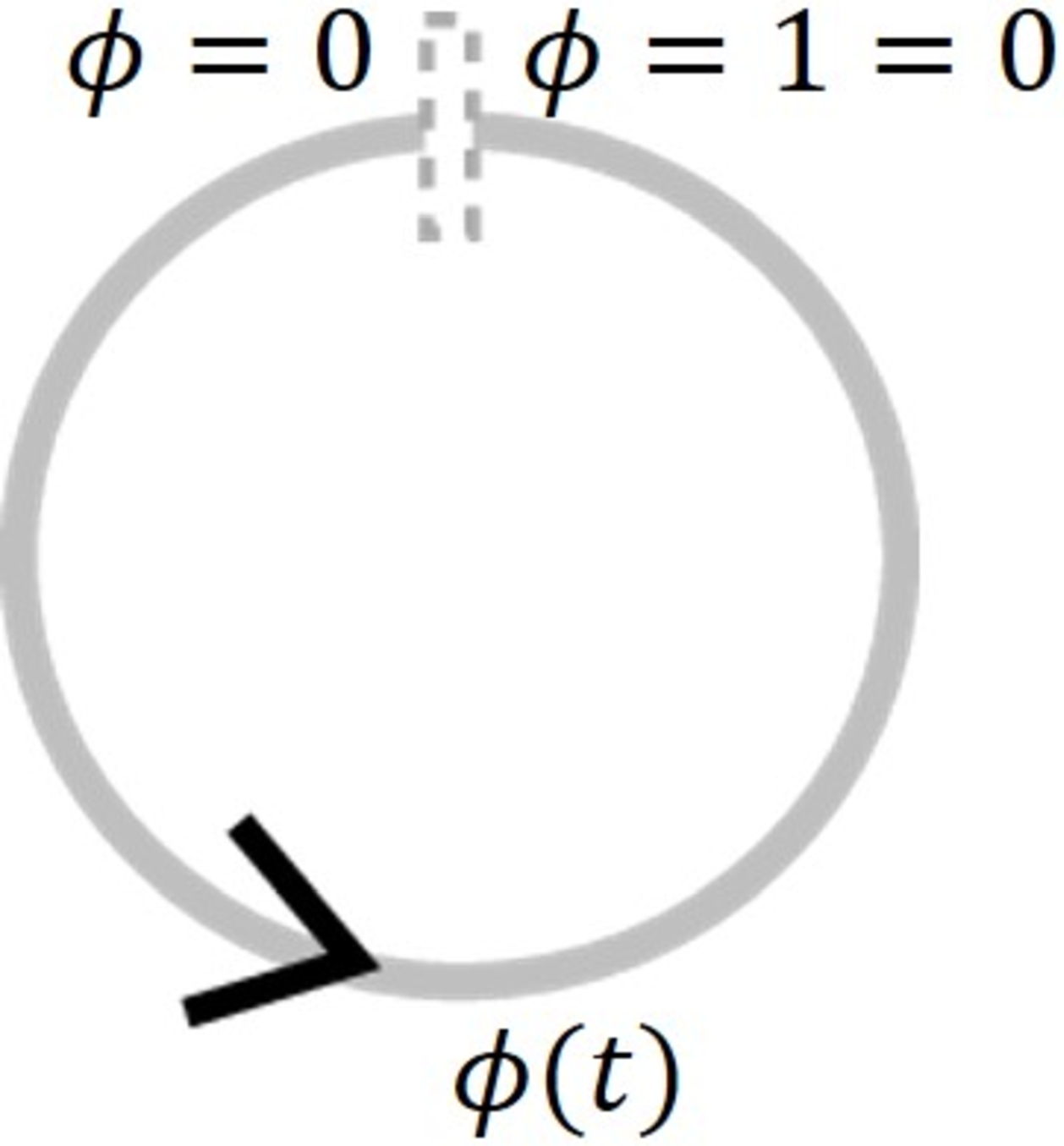}
     \put(2,115){\tiny \bf (d)}
     \end{overpic}}%
 \end{center}
 \caption{\ccyan{(a) The periodic solution $u(t)$ of the FitzHugh-Nagumo model \eqref{eq:FN}. (b) The periodic trajectories of $(u(t),w(t))$. (c)(d) The phase model \eqref{eq:phi-0} with $\phi$ in torus $[0,1)$. }}
\label{fig:FN-u}
\end{figure}

\ccyan{For two-coupled cells, let $u_1$ and $u_2$ represent the membrane currents for two cells respectively,  
which satisfies the coupled FN model:  
\begin{subequations}\label{eq:FN-2cell}
\begin{align}
& \frac{du_1}{dt} = u_1(u_1-a)(1-u_1) - w_1 + \kappa (u_2 - u_1), \quad \frac{dw_1}{dt} = \tau(u_1 - b w_1), \\
& \frac{du_2}{dt} = u_2(u_2-a)(1-u_2) - w_2 + \kappa (u_1 - u_2), \quad \frac{dw_2}{dt} = \tau(u_2 - b w_2). 
\end{align}
\end{subequations}
Here, $\kappa (u_i - u_j)$ describes the interaction between two cells. 
In Figure~\ref{fig:FN-2} (a)(b), we show an example of the synchronization of $(u_1, u_2)$, 
where the trajectories $(u_1,w_1)$ and  $(u_2,w_2)$ tend to the same limit cycle for sufficiently large time $t$. }

\ccyan{Assume that trajectories $\{(u_i,w_i)\}_{i=1,2}$ are synchronized and $T$-periodic for large $t$, 
that is, $\{(u_i,w_i)\}_{i=1,2}$ both tend to the limit cycle $\bm{\chi}$ with $\bm{\chi}(t+T) = \bm{\chi}(t)$ $(i=1,2)$. 
We set
\[
\bm{x}_i(t):=(u_i(t),w_i(t)), \quad \kappa (\bm{x}_j - \bm{x}_j) :=[u_j-u_i, 0]^\top (i,j=1,2, \ i \neq j),
\]
and rewrite the model \eqref{eq:FN-2cell} as follows  
\begin{subequations}\label{eq:FN-2cell-x}
\begin{align}
& \frac{d \bm{x}_1}{dt} = \bm{f}(\bm{x}_1) + \kappa (\bm{x}_2 - \bm{x}_1),  \\
& \frac{d \bm{x}_2}{dt} = \bm{f}(\bm{x}_2) + \kappa (\bm{x}_1 - \bm{x}_2). 
\end{align}
\end{subequations}
}

\ccyan{Denoting by $\phi_1(t)$ and $\phi_2(t)$ the phase functions for $\bm{x}_1$ and $\bm{x}_2$ respectively, 
we consider $\phi_i(t)$ as a function of the limit cycle $\bm{\chi}$, i.e., $\phi_i(t) = \phi_i(\bm{\chi}_i(t))$, satisfying  
\begin{equation}\label{eq:phi-0-2-cell}
\frac{d\phi_i}{dt} = \mu \quad \quad (\mu = \frac{1}{T}).
\end{equation}
Here, $\phi_i$ is $T$-periodic if we set $\phi = \phi + k \ (\forall k \in \mathbb{Z})$.  
Reversely, one can think of $\bm{\chi}$ as a function of $\phi_i$, saying $\bm{\chi}(\phi(t))$. 
Note that $\phi_i(\bm{\chi}(t))$ can be extended to all trajectories $\bm{x}_i$, namely $\phi_i(t) = \phi_i(\bm{x}_i(t))$. 
Analogously to \eqref{eq:phi-0-dx} and \eqref{eq:phi-0-dx-chi}, 
on the limit cycle $\bm{\chi}(\phi_i)$,  
\begin{equation}\label{eq:phi-0-dx-chi-Z}
\mu = \left. \frac{d\phi_i (\bm{x}_i(t))}{dt} \right|_{\bm{x}_i = \bm{\chi}(\phi_i)} = \underbrace{\left. \frac{\partial \phi_i}{\partial \bm{x}_i} \right|_{\bm{x}_i=\bm{\chi}(\phi_i)} }_{=: \bm{Z}(\phi_i)} \cdot [\bm{f}(\bm{\chi}(\phi_i)) + \kappa (\bm{\chi} (\phi_2)-\bm{\chi}(\phi_1)) ] = \bm{Z}(\phi_i) \cdot \bm{f}(\bm{\chi}(\phi_i)),  
\end{equation}
where $\kappa (\bm{\chi} (\phi_2)-\bm{\chi}(\phi_1)) = \kappa (\bm{\chi} (\mu t)-\bm{\chi}(\mu t)) = 0$ because the synchronization $\phi_1=\phi_2 = \mu t$ occurs at $\bm{\chi}$.  
And $\bm{Z}(\phi_i(\bm{\chi}(t)))$ is also $T$-periodic. }

\ccyan{Under the assumption that the difference between  $\bm{x}_i$ and $\bm{\chi}$ are small, 
that is $|\bm{x}_i - \bm{\chi}| = O(\ep) \ll 1$, we calculate as 
\begin{equation}\label{eq:phi-0-dx-chi-1}
\begin{aligned}
\frac{d\phi_1}{dt} = & \frac{\partial \phi_1}{\partial \bm{x}_1} \cdot [\bm{f}(\bm{x}_1) + \kappa (\bm{x}_2-\bm{x}_1) ] \\
= & \left( \bm{Z}(\phi_1(t)) + O(\ep) \right) \cdot  [\bm{f}(\bm{\chi}(\phi_1)) + \kappa (\bm{\chi}(\phi_2) - \bm{\chi}(\phi_1)) + O(\ep) ] \\
= & \mu + \bm{Z}(\phi_1(t)) \cdot \kappa (\bm{\chi}(\phi_2) - \bm{\chi}(\phi_1)) + O(\ep).
\end{aligned}
\end{equation}
}

\ccyan{Let us adopt the approximation approach from \cite{Kuramoto}. 
For sufficiently large time $t$, $(u_1,w_1)$ and $(u_2,w_2)$ tend to the limit cycle.  
Then $\bm{Z}(\phi_1(t)) \cdot \kappa (\bm{\chi}(\phi_2) - \bm{\chi}(\phi_1))$ can be regarded as a perturbation term,
which is approximately replaced by its average in $(t,t+T)$, 
that is 
\[
\begin{aligned}
& \frac{1}{T} \int_t^{t+T} \bm{Z}(\phi_1(t')) \cdot \kappa (\bm{\chi}(\phi_2(t')) - \bm{\chi}(\phi_1(t')))~dt' \\
= & \frac{1}{T} \int_0^T  \bm{Z}(\phi_1(t)+\mu t') \cdot \kappa (\bm{\chi}(\phi_2(t) + \mu t')-\bm{\chi}(\phi_1(t) + \mu t'))~dt' + O(\ep) \\
= & \int_0^T  \bm{Z}(\eta + \phi_1(t)-\phi_2(t)) \cdot \kappa (\bm{\chi}(\eta)-\bm{\chi}(\eta + \phi_1(t)-\phi_2(t)))~d\eta =: \Gamma(\phi_1(t) - \phi_2(t)), 
\end{aligned}
\]
where we have ignored the small term $O(\ep)$. 
Replacing $\bm{Z}(\phi_1(t)) \cdot \kappa (\bm{\chi}(\phi_2) - \bm{\chi}(\phi_1))$ by the average $\Gamma(\phi_1(t) - \phi_2(t))$ in \eqref{eq:phi-0-dx-chi-1} yields 
\begin{equation}\label{eq:phi-0-dx-chi-2}
\begin{aligned}
\frac{d\phi_1}{dt} = \mu + \Gamma(\phi_1(t) - \phi_2(t)).
\end{aligned}
\end{equation}
We have presented a rough derivation of the phase equation for $\phi_1$ above ($\phi_2$ can be treated in the same way). 
The obtention of the closed-form of $\Gamma(\cdot)$ requires technical calculation, which is omitted here. 
One can refer to \cite{Kuramoto} for more rigorous and detailed mathematical arguments. 
For simplicity, we replace $\Gamma(\cdot)$ by $\sin(2\pi (\cdot))$ without losing the essentiality of the model. }

\ccyan{In summary, 
we chose the Kuramoto model as the basic model to study the synchronization behavior of two-coupled cardiomyocytes: 
\begin{equation}\label{eq:phi-0-dx-chi-3}
\begin{aligned}
& \frac{d\phi_1}{dt} = \mu + \sin(2\pi(\phi_1(t) - \phi_2(t))), \\
& \frac{d\phi_2}{dt} = \mu + \sin(2\pi(\phi_2(t) - \phi_1(t))). 
\end{aligned}
\end{equation}
\begin{remark}
We apply the Kuramoto model owing to its wide application in studying the oscillators' synchronization. 
But from the mathematical points of view, 
it is worth to consider other interaction terms besides of $\sin(2\pi(\phi_j-\phi_i))$, for example $(\phi_j-\phi_i)/|\phi_j-\phi_i|$, $(\phi_j-\phi_i)^\alpha$ and so on. 
\end{remark}
}
\begin{figure}
\begin{center}
  \subfigure{%
     \begin{overpic}[width=0.24\linewidth]{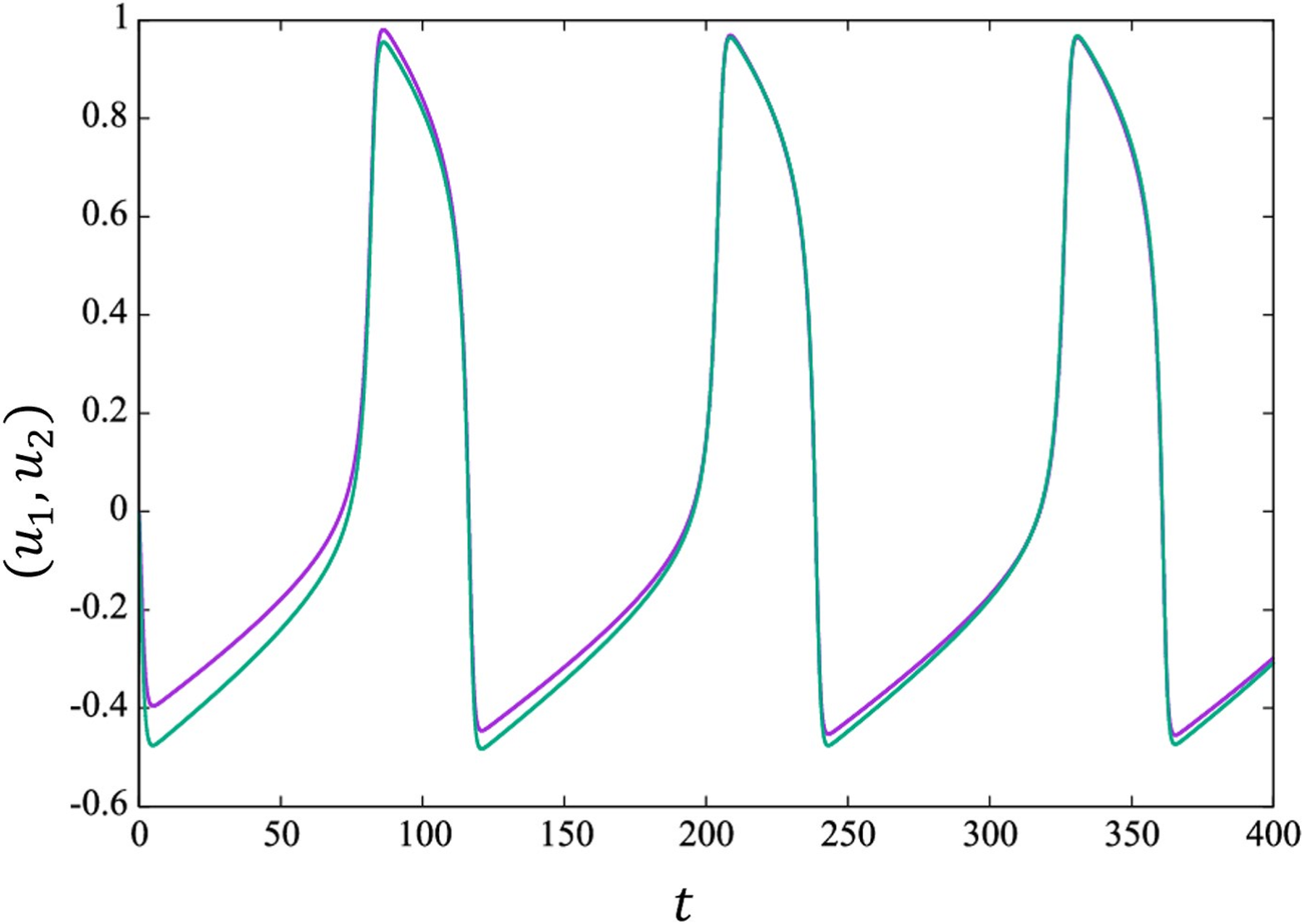}
     \put(-4,70){\tiny \bf (a)}
     \end{overpic}}%
   \hspace*{0.3cm}
  \subfigure{%
     \begin{overpic}[width=0.24\linewidth]{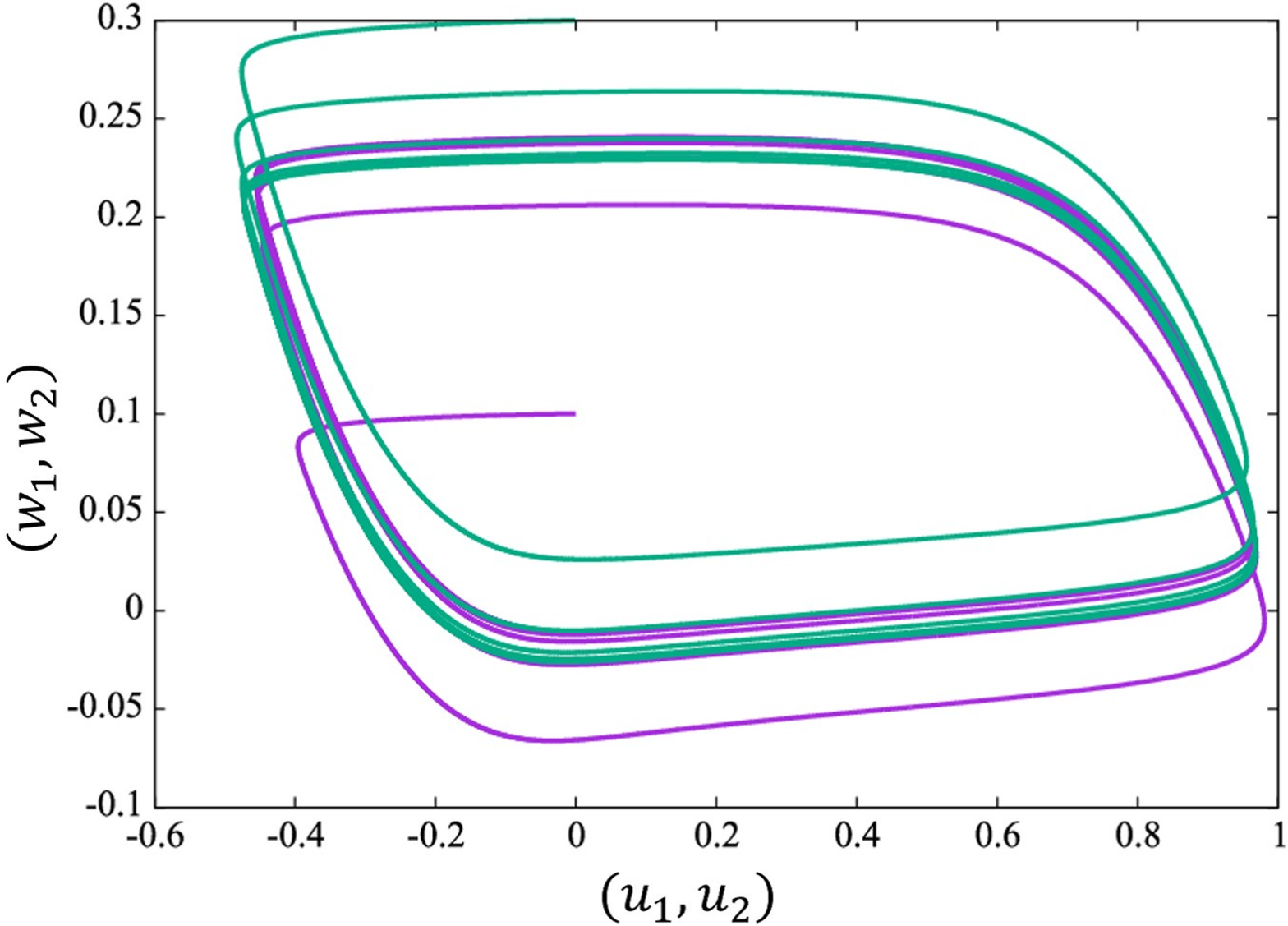}
     \put(-4,70){\tiny \bf (b)}
     \end{overpic}}%
   \hspace*{0.2cm}
     \subfigure{%
     \begin{overpic}[width=0.24\linewidth]{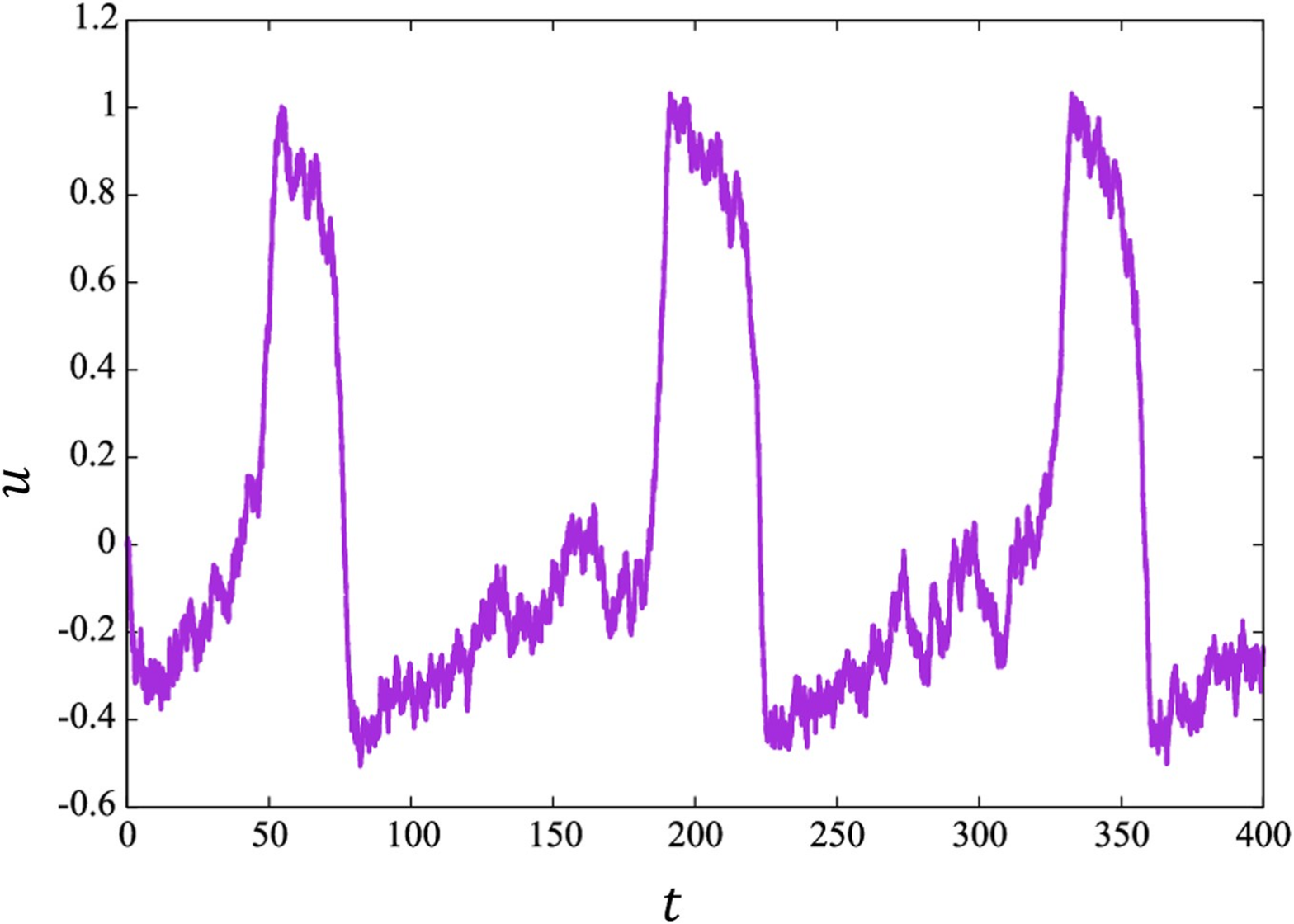}
     \put(-4,70){\tiny \bf (c)}
     \end{overpic}}%
   \hspace*{0.3cm}
  \subfigure{%
     \begin{overpic}[width=0.24\linewidth]{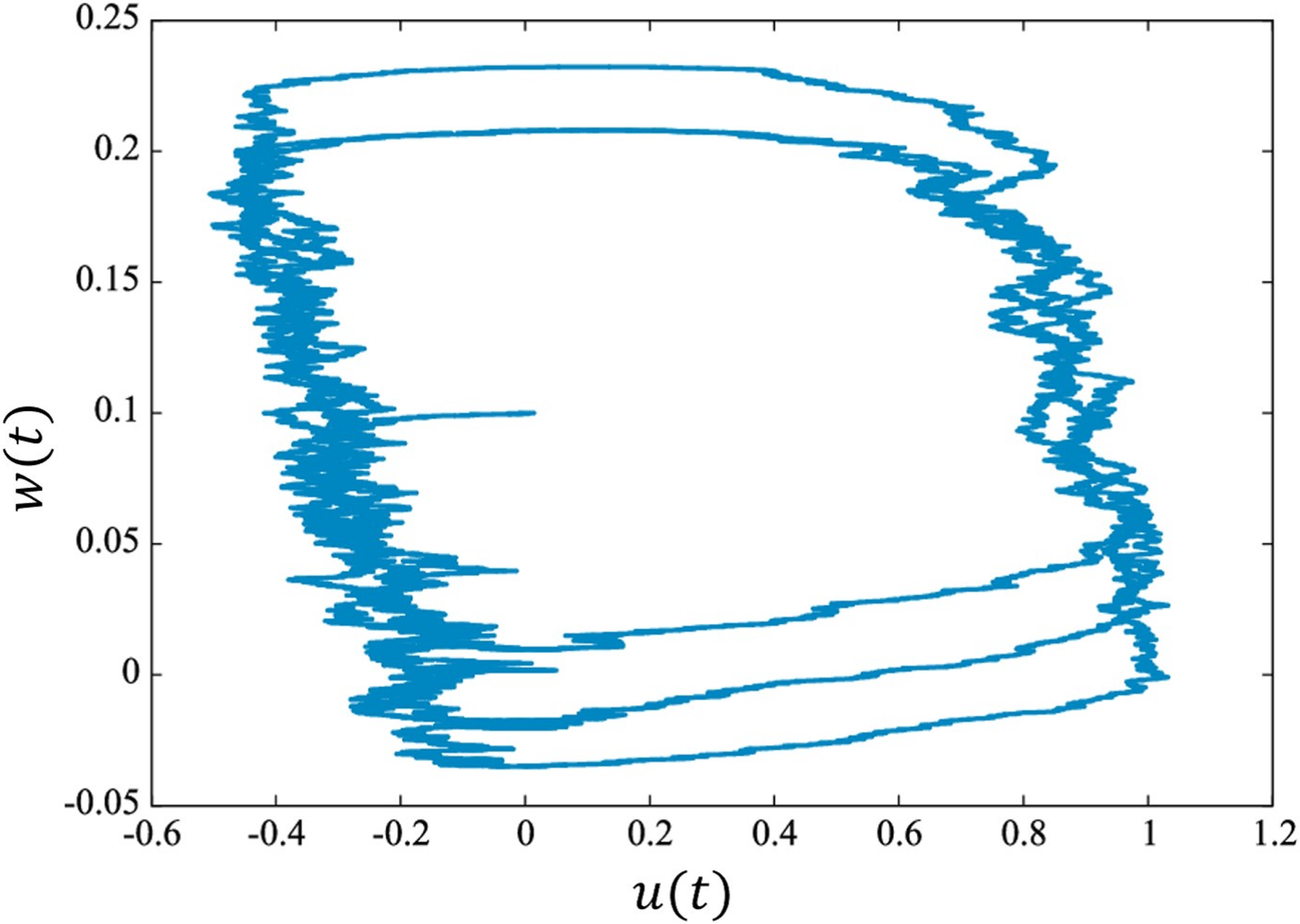}
     \put(-4,70){\tiny \bf (d)}
     \end{overpic}}%
 \end{center}
 \caption{\ccyan{(a) The synchronization of $u_1(t)$ and $u_2$ for the two-coupled FitzHugh-Nagumo models \eqref{eq:FN-2cell}. (b) The periodic trajectories $\{(u_i(t),w_i(t))\}_{i=1,2}$ have the same limit cycles. 
 (c) The periodic solution $u(t)$ of the FitzHugh-Nagumo model for one cell with noise \eqref{eq:phi-0-noise}. (d) The periodic trajectories $(u(t),w(t))$ of \eqref{eq:phi-0-noise}.}
 }
\label{fig:FN-2}
\end{figure}

\subsection{\ccyan{The FitzHugh-Nagumo  and phase models with noise}}
\ccyan{The biological experiments (Figure~\ref{fig:exp-2} (b)) show that the beating intervals of cardiomyocytes are not perfectly periodic, 
which indeed are effected by noise. 
Therefore, it is necessary to consider the FN model with noise: 
\begin{subequations}\label{eq:FN-noise}
\begin{align}
& du = [u(u-a)(1-u) - w]dt + \sigma_1 dW(t), \\
& dw = \tau(u - bw)dt,
\end{align}
\end{subequations}
where $\sigma_1>0$ denotes the noise strength, 
and $dW(t)$ the white noise ($W(t)$ is the standard Brownian motion). 
A realization of the membrane current $u(t)$ and the trajectories $(u(t), w(t))$ of \eqref{eq:FN-noise} is presented in Figure~\ref{fig:FN-2} (c)(d). 
A simulation of the distribution of the beating interval $\Delta t$ is plotted in Figure~\ref{fig:FN-phi-noise} (b)) with mean value $\mathbf{E}(\Delta t) \approx 108.88$ and standard variance $\sqrt{\mathbf{Var}(\Delta t)} \approx 17.662$.}

\ccyan{Since the relationship between the distribution of beating interval and the parameters $(a,b,\tau,\sigma_1)$ of FN model has not been understood fully, 
the phase model with noise is applicable to study the beating process of cardiomyocyte, 
which is stated as follows:    
\begin{subequations}\label{eq:phi-0-noise}
\begin{align}
& d\phi(t) = \mu dt + \sigma dW(t), \\
& \phi(0) = 0, 
\end{align}
\end{subequations}
where $\sigma>0$ denotes the noise strength. 
}

\ccyan{Define the beating interval $\Delta t := t^{(k)}-t^{(k-1)}$ with $t^{(k)}$ the first passage time that $\phi(t) = k$. 
Or equivalently, we set $\phi(t)=0$ when $\phi(t-)=1$ (the phase jumps to $0$ when reaching $1$), 
and $\Delta t = t^{(k)}-t^{(k-1)}$ with $t^{(k)} := \inf \{t > t^{(k-1)} \mid \phi(t-) =1 \}$ $(t^{(0)} = 0)$. 
By stochastic calculus, one can verify that     
\begin{equation}\label{eq:1-c-trad-t^1-E-Var}
\mathbf{E}(\Delta t) = \frac{1}{\mu},  \quad \mathbf{Var}(\Delta t) = \frac{\sigma^2}{\mu^3}, \quad \mathbf{CV}(\Delta t) = \sqrt{\frac{ \mathbf{Var}(\Delta t)}{ \mathbf{E}(\Delta t)^2}} = \frac{\sigma}{\sqrt\mu}.  
\end{equation}
}
\ccyan{By \eqref{eq:1-c-trad-t^1-E-Var}, together with $\mathbf{E}(\Delta t) = 108.88$ and $\sqrt{\mathbf{Var}(\Delta t)} = 17.662$ (the simulation from FN model), 
we first compute the parameter $(\mu, \sigma) = (0.0092, 0.016)$,   
and then carry out the numerical simulation of \eqref{eq:phi-0-noise} and plot the distribution of the beating interval in Figure~\ref{fig:FN-phi-noise} (c). 
Although two distributions, Figure~\ref{fig:FN-phi-noise} (b) and (c), have the same mean value and variance, 
the density functions do not consistent with each other well. 
In view of the trajectory of $u(t)$ in Figure~\ref{fig:FN-2} (c), 
when $u(t)$ increases from $0.2$ to $1$ and decreases from $0.6$ to $-0.1$ rapidly, 
the noise has little effect to the dynamic of $u$, and also to the period of oscillation cycle. 
In other words, when the action potential (the pulsation of $u$, or called beating) occurs, 
the noise effect is somehow inhibited such that the pulsation cannot be reversed by the noise. 
This irreversibility has not been captured by the phase model \eqref{eq:phi-0-noise}, 
which may be the main reason causing the inconsistency between the distributions of Figure~\ref{fig:FN-phi-noise} (b) and (c). 
To address issue, in the next section, 
we will propose a phase model incorporating the irreversibility after beating.   }

\begin{figure}
\begin{center}
    \subfigure{%
     \begin{overpic}[width=0.31\linewidth]{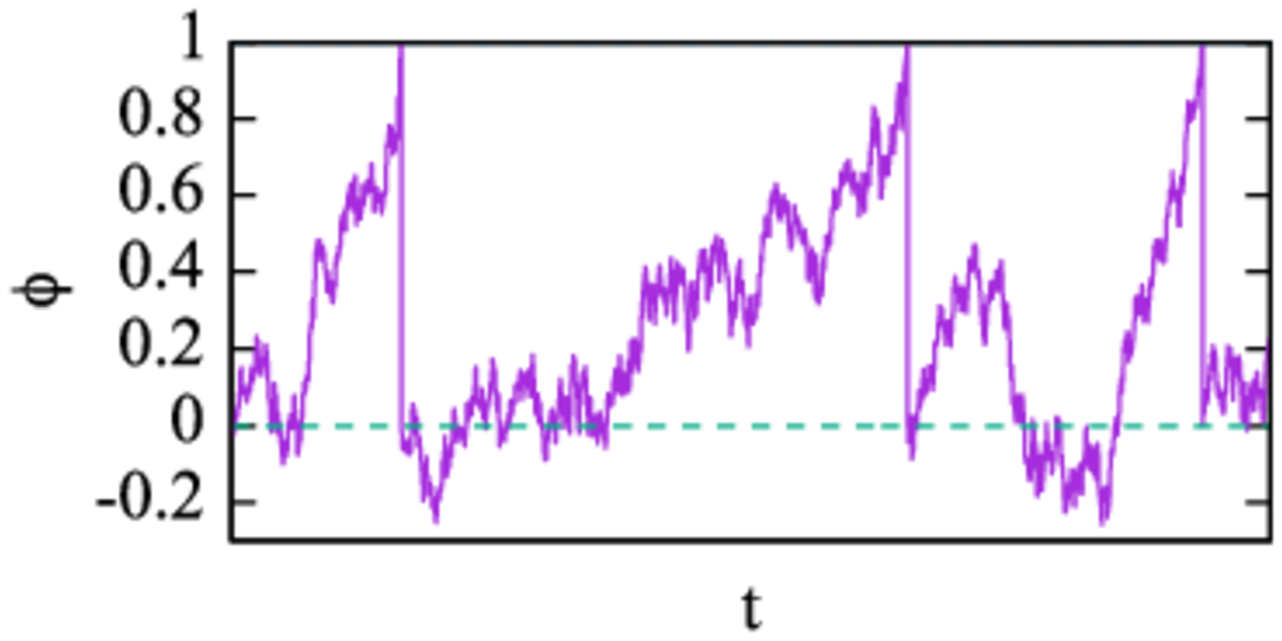}
     \put(10,50){\small \bf (a)}
     \end{overpic}}%
   \hspace*{0.1cm}
\subfigure{%
     \begin{overpic}[width=0.31\linewidth]{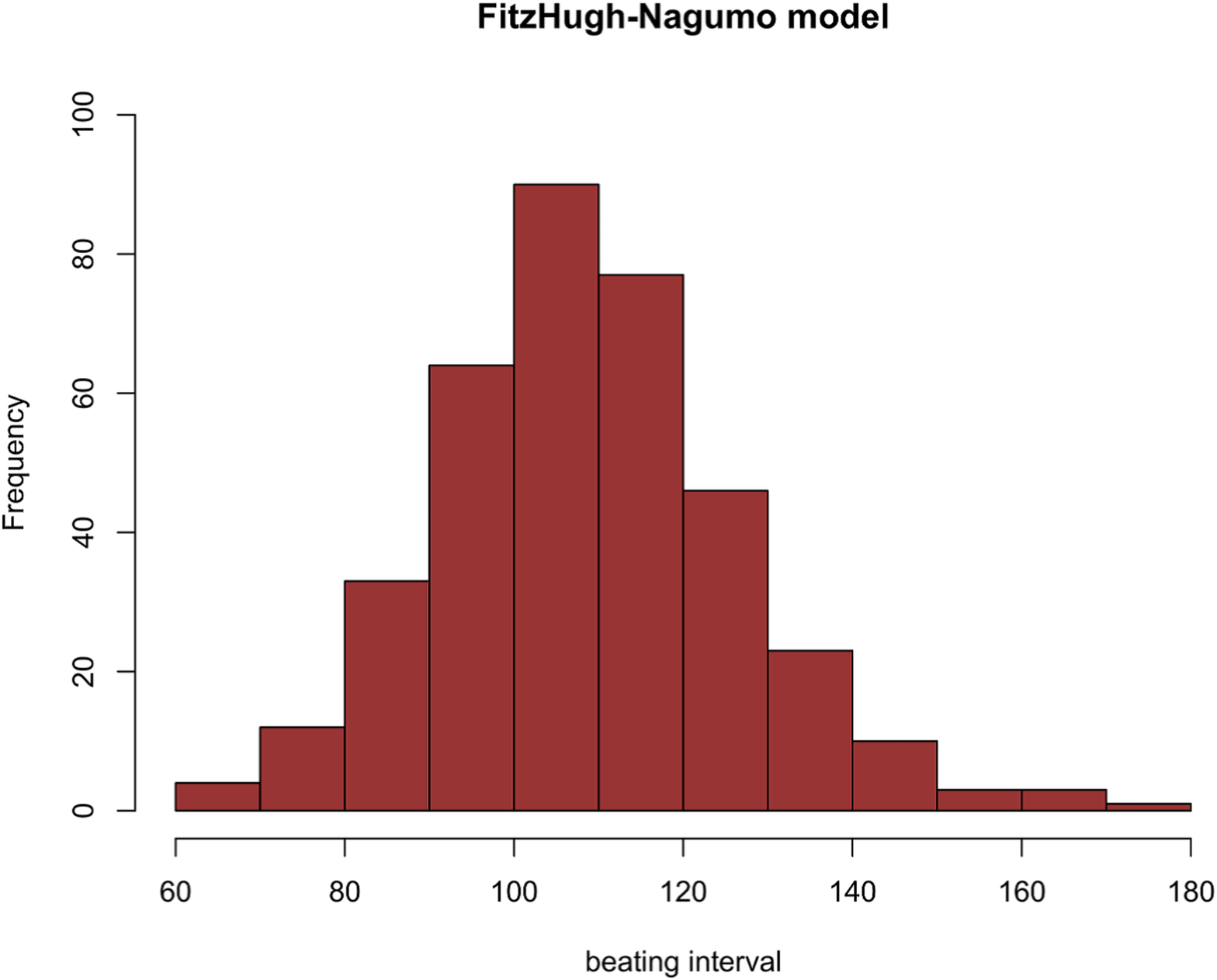}
     \put(10,78){\small \bf (b)}
     \end{overpic}}%
   \hspace*{0.1cm}
\subfigure{%
     \begin{overpic}[width=0.31\linewidth]{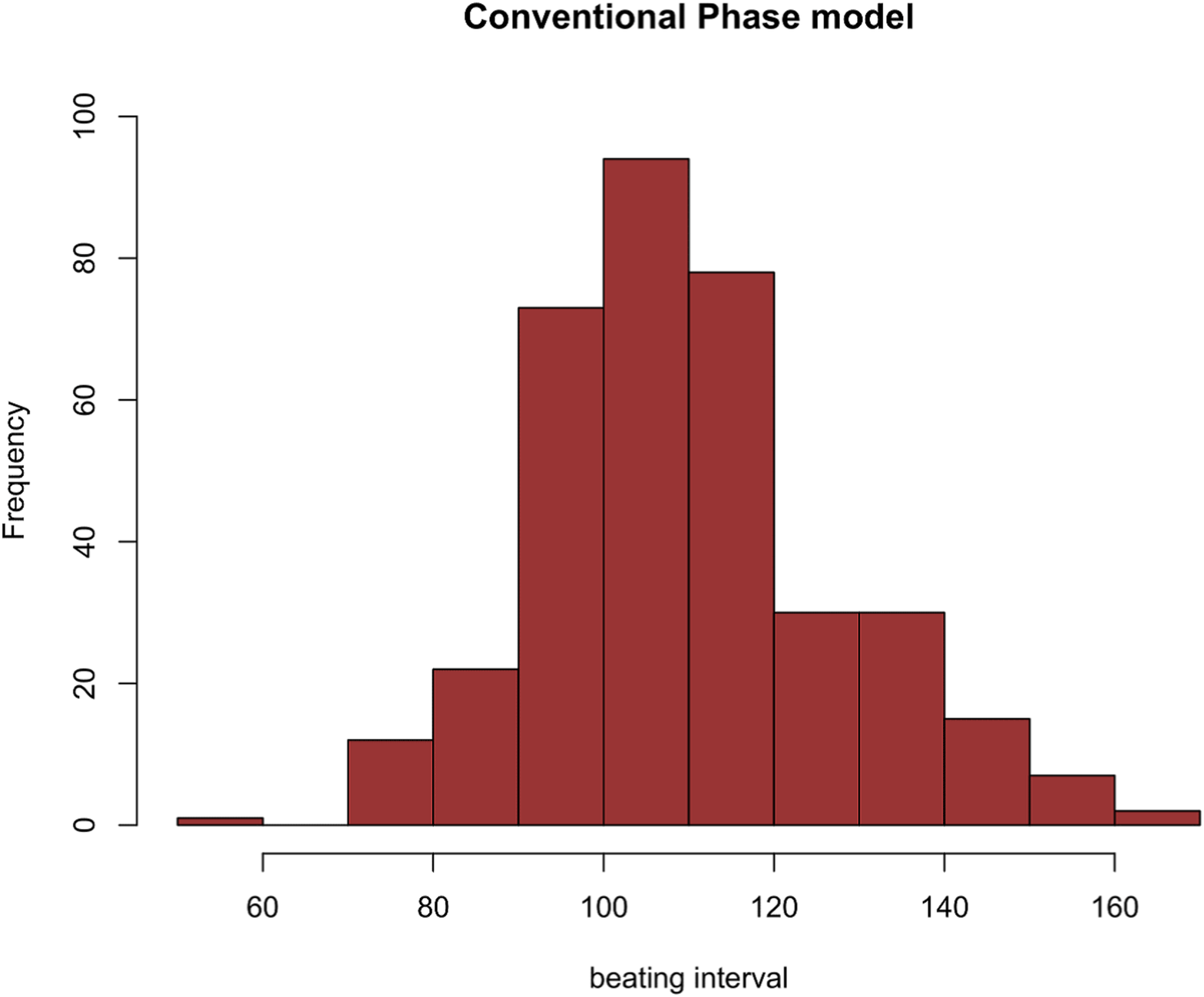}
     \put(10,78){\small \bf (c)}
     \end{overpic}}%
 \end{center}
 \caption{\ccyan{(a) A realization of the phase $\phi$ from \eqref{eq:FN-noise}. 
 (b) The distribution of the beating intervals $\Delta t$ from the FitzHugh-Nagumo model \eqref{eq:FN-noise} with $(a,b,\tau,\sigma) = (-0.1,0.5,0.01,0.1)$. 
 (c) The distribution of the beating intervals from the phase model \eqref{eq:phi-0-noise} with $\mu = \frac{1}{\mathbf{E}(\Delta t)}$ and $\sigma = \sqrt{ \mathbf{Var}(\Delta t) \mu^3}$ according to the formula \eqref{eq:1-c-trad-t^1-E-Var}. }}
\label{fig:FN-phi-noise}
\end{figure}
%

\section{The phase models for an isolated cell} \label{sec:1-c} 
Regarding the single-isolated cardiomyocyte as an oscillator with phase $\phi$, 
we say the cell beats at time $t$ when $\phi(t-) = 1$, and we let $\phi(t)$ jumps to $0$ immediately after beating to start a new oscillation cycle (beating interval). 
In view of the dynamics of $u$ in Figure~\ref{fig:FN-2} (c), 
the beating ($\phi(t-) = 1, \phi(t)=0$) corresponds to the action potential (i.e., $u$ increases quickly from $0.2$ to $1$) which cannot be reversed by noise effect. 
Hence, we shall enforce an inhibition to the noise effect at $\phi(t)=0$ such that $\phi$ cannot be dragged backward by noise. 
To incorporate the irreversibility after beating, 
we consider the stochastic phase model with the \emph{reflective boundary}:  
\begin{subequations}\label{eq:1-c}
\begin{align}
& d \phi(t) = \mu dt + \sigma d W(t) + dL(t), \label{eq:1-c-a} \\
& \phi(0) = 0, \label{eq:1-c-b} \\
& \phi(t) = 0 \quad \text{ when } \phi(t-) = 1, \label{eq:1-c-c}
\end{align}
\end{subequations}
where $L(t)$ is the process to prevent $\phi(t)$ being driven backward by noise when $\phi(t)=0$.  
$L(t)$ indeed describes the reflective boundary at $\phi=0$ (cf. \cite{Harrison85, Lions84, Skorokhod61}).  
To state the definition of $L(t)$, 
we set $t^{(k)}$ the $k$-th time that $\phi(t)$ approaches $1$ ($k = 1,2,\ldots$):   
\begin{equation}\label{eq:1-c-t^k}
t^{(k)} = \inf \{ t > t^{(k-1)} \ : \ \phi(t-) = 1 \}, \quad t^{(0)} = 0. 
\end{equation}
We say the cell beats at $t^{(k)}$ ($k = 1,2,\ldots$). 
In view of \eqref{eq:1-c-c}, when $\phi$ approaches $1$ (i.e. $\phi(t^{(k)}-) = 1$), 
$\phi$ immediately jumps to $0$, i.e., $\phi(t^{(k)}) = 0$, 
and a new oscillation cycle begins. 

The rigorous definition of $L(t)$ is given as follows (see Figure~\ref{fig:X-phi-L} (a)): 
\begin{itemize}
\item[(L1)] $L(0) = 0$, and $L(t)$ is a nondecreasing, continuous process for $t^{(k-1)} \le t < t^{(k)}$ such that $\phi(t) \ge 0$;
\item[(L2)] $L(t)$ increases only when $\phi(t) = 0$.
\end{itemize}
\ccyan{
In simulation, a simple approach \cite{Slominski94} is to reset $\phi(t_n)=0$ when $\phi(t_n) = \phi(t_{n-1}) + \mu \Delta t + \sqrt{\Delta t} \sigma \eta < 0$, 
where $\Delta t$ denotes the time-step ($\Delta t \ll 1$), $t_n = n \Delta t$ $(n=1,2,\cdots,)$ and $\eta \sim \mathcal{N}(0,1)$. 
}
\ccyan{\begin{remark}
Assume that the phase state $\phi_0$ $(0<\phi_0<1)$ corresponds to the promptly decreasing of the membrane current $u$ from $0.6$ to $-0.1$, 
which also cannot be reversed by noise. 
And one may also input a reflective boundary at $\phi=\phi_0$ to control the noise, 
such that $\phi$ cannot be driven backward by noise when $\phi= \phi_0$. 
But in this paper, we only implement the reflective boundary at $\phi=0$, 
which seems sufficient for application \cite{Hayashi}. 
\end{remark}}
During the first beating process, i.e., $t \in [0, t^{(1)})$, integrating \eqref{eq:1-c-a} yields 
\[
\phi(t) = \underbrace{\mu t + \sigma W(t)}_{=:X(t)} + L(t), 
\]
which means $\phi(t)$ is in fact a combination of the $(\mu,\sigma)$-Brownian motion $X(t)$ and the process $L(t)$. 
We plot an example of $(\phi(t), X(t), L(t))$ for $t \in [0,t^{(1)})$ in Figure~\ref{fig:X-phi-L} (a). 
Since $\phi$ returns to $0$ instantly when approaching $1$, 
$\phi(t)$ is a renewal process for $t \ge t^{(k)}$ ($k = 1,2,\cdots$) (see Figure~\ref{fig:X-phi-L} (b)),  
and the beating intervals (oscillation period) $\Delta t^{(k)} := t^{(k)} - t^{(k-1)}$ ($k = 1,2,\ldots$) are independent, identically distributed random variables.  
Hence, we only need to investigate $\Delta t^{(1)} = t^{(1)}  - 0$.  

\begin{figure}
\begin{center}
    \subfigure{%
     \begin{overpic}[width=0.33\linewidth]{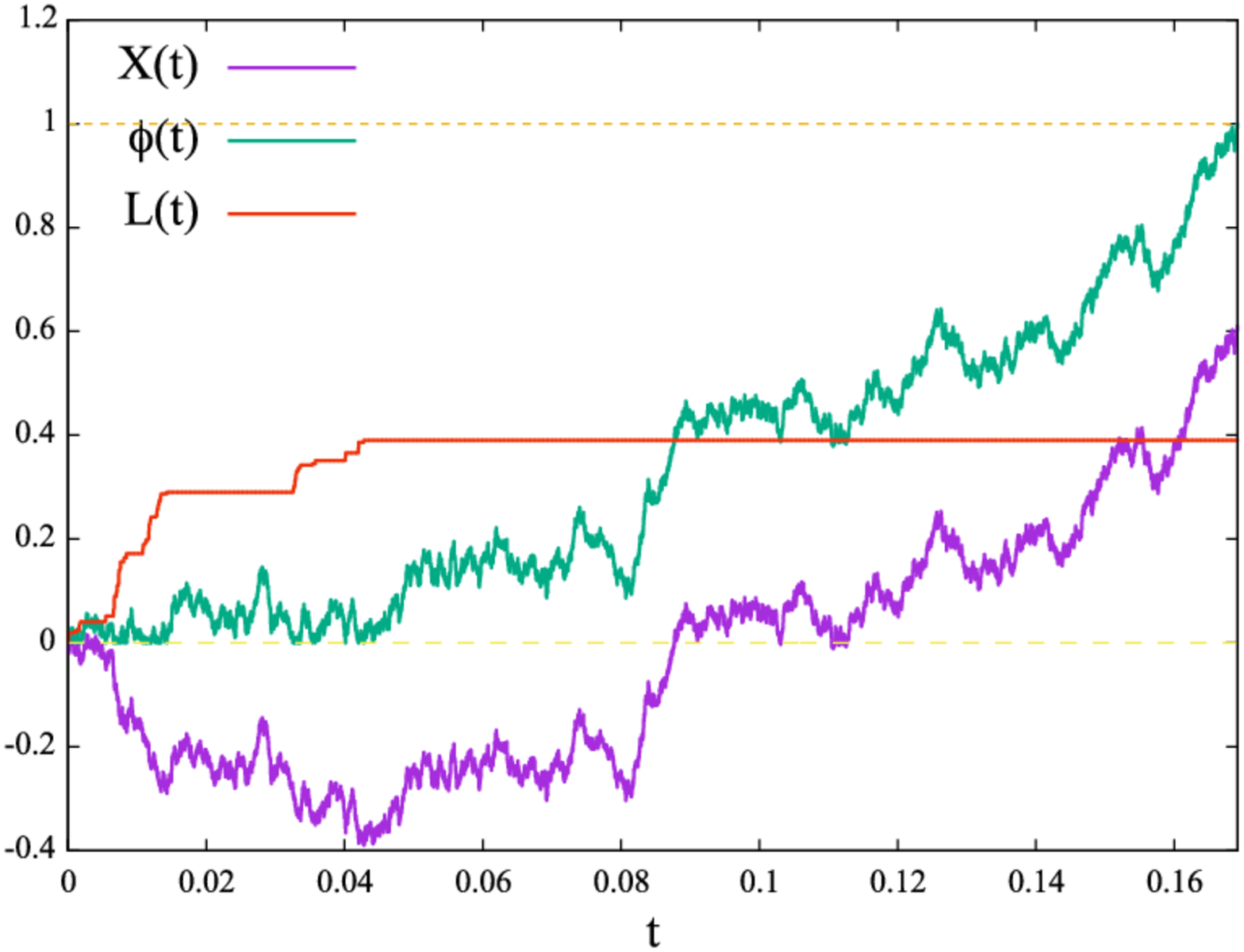}
     \put(1,75){\small \bf (a)}
     \end{overpic}}%
   \hspace*{0.1cm}
     \subfigure{%
     \begin{overpic}[width=0.33\linewidth]{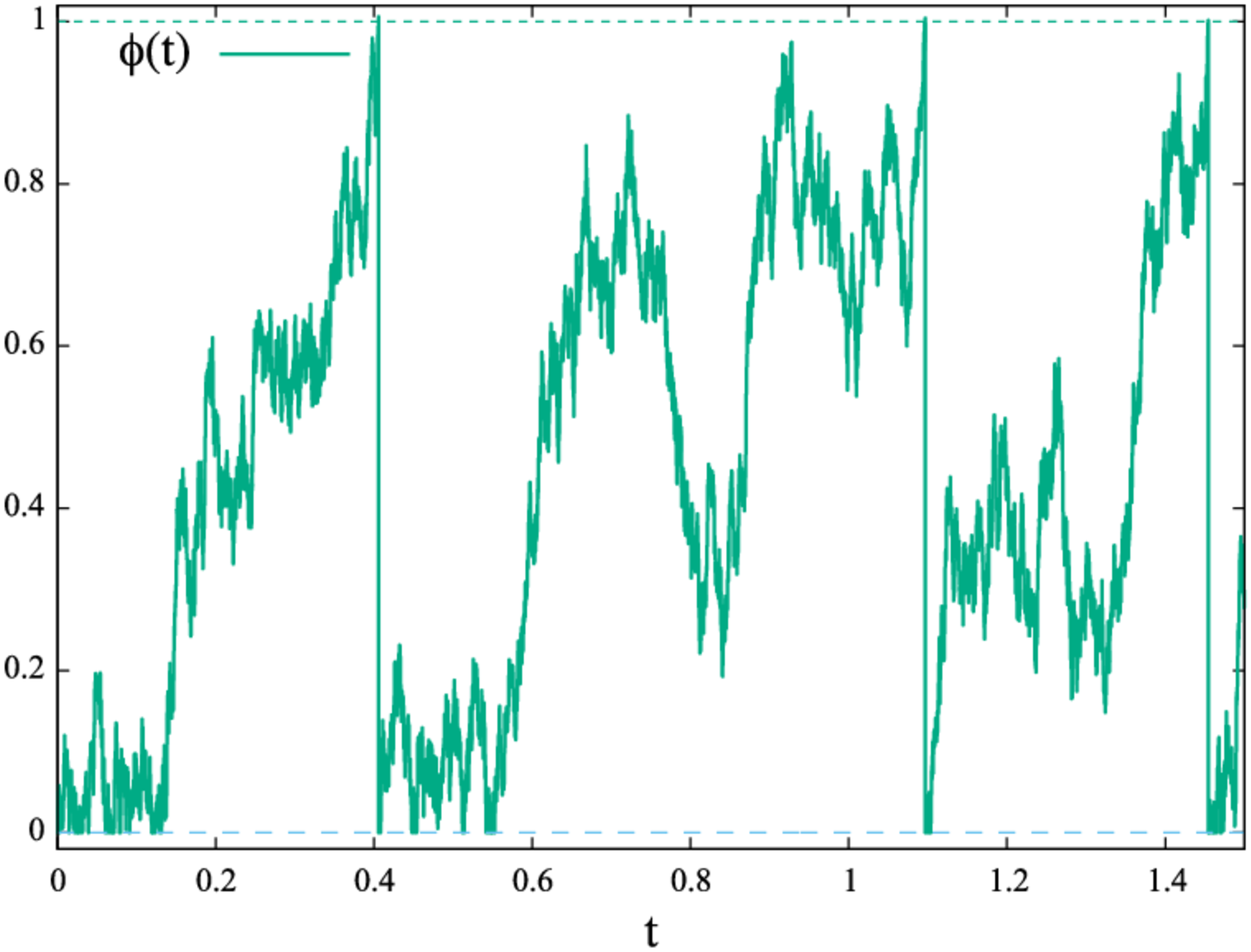}
     \put(1,75){\small \bf (b)}
     \end{overpic}}%
\subfigure{%
     \begin{overpic}[width=0.31\linewidth]{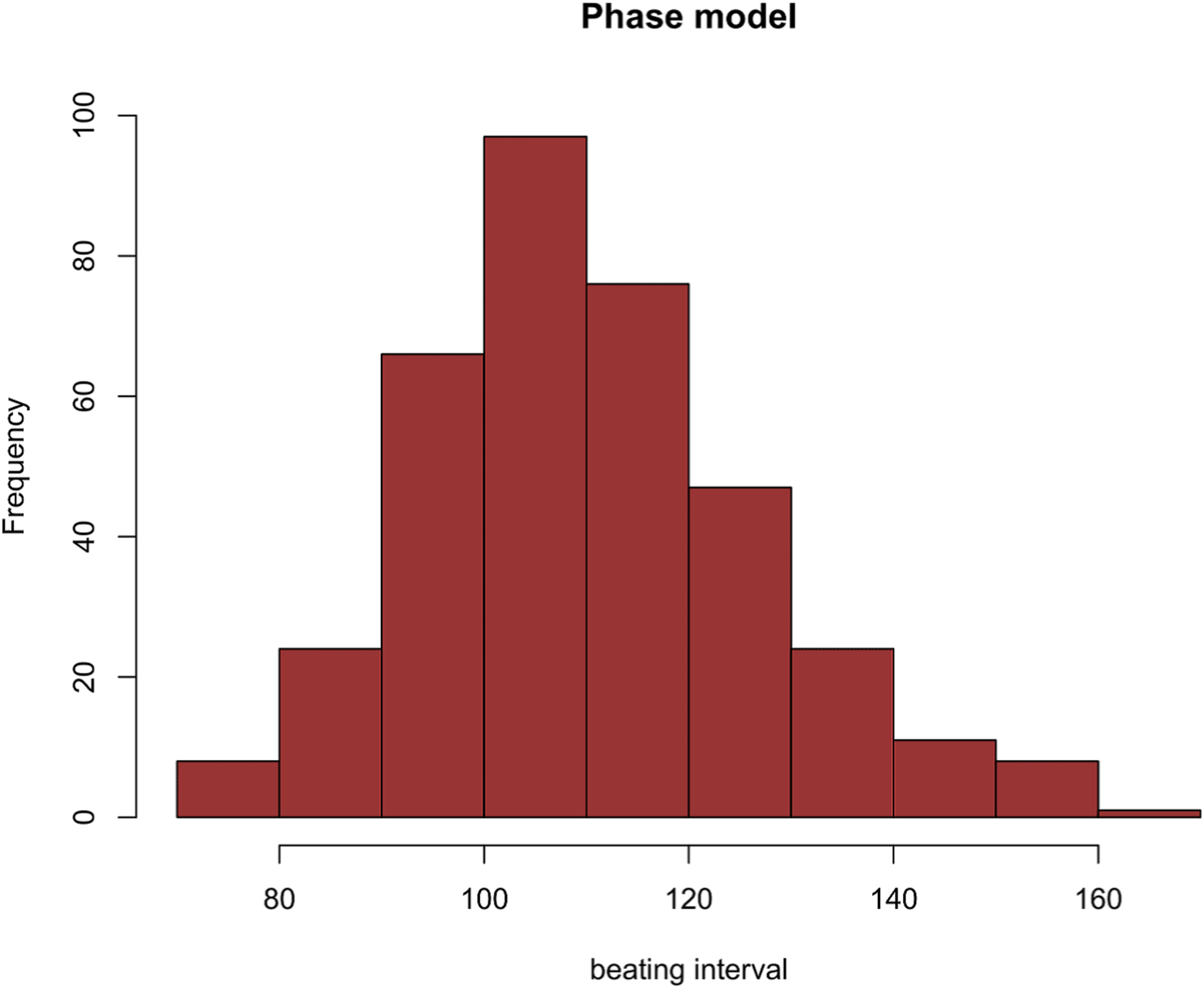}
     \put(10,78){\small \bf (c)}
     \end{overpic}}%

 \end{center}
 \caption{\ccyan{(a) A realization of $(\phi(t), X(t), L(t))$ with $(\mu, \sigma) = (0.3,1)$, where $X(t)$ is the $(\mu,\sigma)$-Brownian motion, 
$\phi(t) = X(t) + L(t)$ and $L(t)$ increases only when $\phi(t)=0$ such that $\phi \ge 0$ always holds true. 
 (b) An example of $\phi(t)$ with $(\mu, \sigma) = (0.3,1)$ by numerical simulation \cite{Slominski94}. 
 (c) The distribution of the beating interval $t^{(1)}$ of \eqref{eq:1-c} with the identical mean value and variance of the FitzHugh-Nagumo model.}}
\label{fig:X-phi-L}
\end{figure}
\begin{theorem}\label{th:1C-E-Var}
For $\mu \ge 0$, $\sigma > 0$, we have 
\begin{equation}\label{eq:1-c-t^1-E}
\mathbf{E}( t^{(1)} ) = \left\{ \begin{aligned} 
				& 1/\sigma^2 &  \text{ for } \mu = 0,  \\
				& (\theta -1 + e^{-\theta}) / (\mu \theta) & \text{ for } \mu > 0,				\end{aligned} \right.
\end{equation}
\begin{equation}\label{eq:1-c-t^1-Var}
\mathbf{Var}( t^{(1)} ) = \left\{ \begin{aligned} 
				& 2/(3\sigma^4) & \text{ for } \mu = 0,  \\
				& e^{-\theta}(e^{-\theta} -5e^\theta + 2\theta e^\theta + 4 + 4\theta)/(\mu^2 \theta^2) & \text{ for } \mu > 0, 				\end{aligned} \right.
\end{equation}
where $\theta = 2\mu/\sigma^2$. 
\end{theorem}
\begin{remark}
Throughout this paper, we always consider the non-negative intrinsic frequency $(\mu \ge 0)$ and positive noise strength $(\sigma>0)$. 
Passing to the limit $\mu \rightarrow 0$, 
one can validate $\frac{\theta -1 + e^{-\theta}}{\mu \theta} \rightarrow \frac{1}{\sigma^2}$  and $\frac{ e^{-\theta}(e^{-\theta} -5e^\theta + 2\theta e^\theta + 4 + 4\theta)}{\mu^2 \theta^2} \rightarrow \frac{2}{3\sigma^4}$.
The coefficient of variance (CV) is given by: 
\begin{equation}\label{eq:1-c-t^1-CV}
\mathbf{CV}^2( t^{(1)} ) = \frac{ \mathbf{Var}(\Delta t) }{ \mathbf{E}^2(\Delta t) } = \left\{ \begin{aligned} 
				& 2/3 \quad & \text{ for } \mu = 0,  \\
				& K(\theta) \quad & \text{ for } \mu > 0, 				\end{aligned} \right.
\end{equation}
where $K(\theta) = \frac{e^\theta (e^{-\theta} -5e^\theta + 2\theta e^\theta + 4 + 4\theta) }{(1+\theta e^\theta - e^\theta)^2}$ and $K(\theta) \uparrow 2/3$ as $\theta \downarrow 0$.  
\end{remark}
\ccyan{\begin{remark}
We can compute the mean value and CV of the beating interval from the experimental data. 
From \eqref{eq:1-c-t^1-CV}, we first determine $\theta$. 
Then, with $\theta$ and the mean value, the coefficient $(\mu, \sigma)$ can be calculated by \eqref{eq:1-c-t^1-E}. 
The applicability of our model \eqref{eq:1-c} to the bio-experimental data has been discussed in \cite{Hayashi}. 
Here, we only compare the simulation with the FN model presented in Section~2. 
Using $\mathbf{E}(\Delta t) = 108.88$ and $\sqrt{\mathbf{Var}(\Delta t)} = 17.662$ from the simulation of the FN model (see Figure~\ref{fig:FN-phi-noise} (b)), 
we calculate the corresponding $(\mu,\sigma)$ of \eqref{eq:1-c}, and carry out the simulation to plot the distribution of the beating interval (see Figure~\ref{fig:X-phi-L} (c)).   
\end{remark}}
\ccyan{\begin{remark}
Besides of the reflective boundary, 
the phase-dependent noise strength $\sigma(\phi)$ with $\sigma(0)=0$ is also considerable.   
In fact, this approach has been applied to the Langevin equation with noise modeling the ion channels' dynamic \cite{Dangerfield12} for the Hodgkin-Huxley formulation, 
where the noise effect is inhibited when the proportion of the opened channels reaches $0$ or $1$ such that the proportion is bounded in $[0,1]$.
\end{remark}}
\begin{proof}[Proof of Theorem~\ref{th:1C-E-Var}]
For any function $g(x)$ in $[0,1]$ with continuous differential $\frac{\partial g}{\partial x}$ and $\frac{\partial^2 g}{\partial x^2}$, Ito's formula yields ($0 \le t \le t^{(1)}$): 
\begin{equation}\label{eq:ito-1}
\begin{aligned}
g(\phi(t-)) & = g(\phi(0)) + \int_0^t \left[ \mu \frac{\partial }{\partial x} + \frac{\sigma^2}{2} \frac{\partial^2 }{\partial x^2} \right] g(\phi(s))~ds \\
& + \int_0^t \sigma \frac{\partial g}{\partial x}(\phi(s))~dW(s) + \int_0^t  \frac{\partial g}{\partial x}(\phi(s))~dL(s).
\end{aligned}
\end{equation}
Since $L(t)$ is nondecreasing and increases only when $\phi(t) = 0$ (see (L1)(L2)), 
\begin{equation}\label{eq:ito-1-a}
\int_0^t  \frac{\partial g}{\partial x}(\phi(s))~dL(s) = \int_{\{ 0\le s \le t \ \mid \ \phi(s) = 0\}}  \frac{\partial g}{\partial x}(0)~dL(s). 
\end{equation}
Now, let $g(x)$ be the solution of: 
\begin{subequations}\label{eq:g-1}
\begin{align}
& \left[ \mu \frac{\partial }{\partial x} + \frac{\sigma^2}{2} \frac{\partial^2 }{\partial x^2} \right] g(x)  = -1 \quad \text{ for } \quad 0<x<1, \label{eq:g-1-a} \\
& g(1) = 0, \quad   \frac{\partial g}{\partial x}(0) = 0. \label{eq:g-1-b}
\end{align} 
\end{subequations}
In view of $\phi(0) = 0$ and $\phi(t^{(1)}-) = 1$, from \eqref{eq:g-1}, \eqref{eq:ito-1-a} and \eqref{eq:ito-1}, 
we obtain 
\begin{equation}\label{eq:ito-1-g-1}
\begin{aligned}
0 & = g(0) +\underbrace{ \int_0^{t^{(1)}} -1 ~ds}_{=-t^{(1)}} + \int_0^{t^{(1)}} \sigma \frac{\partial g_1}{\partial x}(\phi(s))~dW(s).  
\end{aligned}
\end{equation}
Because the expectation of an It\^{o}'s integral is zero (\cite{Evans13, Mckean69}), 
\[
\mathbf{E} \left[ \int_0^{t^{(1)}} \sigma \frac{\partial g_1}{\partial x}(\phi(s))~dW(s) \right] = 0,
\]
which, together with \eqref{eq:ito-1-g-1}, gives   
\begin{equation}\label{eq:ito-1-g-1-a}
\mathbf{E} (t^{(1)}) = g(0). 
\end{equation}
Therefore, the obtention of $\mathbf{E} (t^{(1)})$ reduces to solve the boundary value problem \eqref{eq:g-1}. 
In fact, 
\begin{equation}\label{eq:g-1-c}
g(x) = \left\{ \begin{aligned}
			& \frac{1-x^2}{\sigma^2} & \text{ for } \mu = 0, \\
			& \frac{1-x}{\mu} - \frac{ \sigma^2(e^{-\frac{2\mu x}{\sigma^2}} - e^{-\frac{2\mu}{\sigma^2}}) }{2\mu^2} & \text{ for } \mu > 0.  
\end{aligned} \right.
\end{equation}
Moreover, one can validate that $\frac{1-x}{\mu} - \frac{\sigma^2(e^{-\frac{2\mu x}{\sigma^2}} - e^{-\frac{2\mu}{\sigma^2}}) }{2\mu^2} \rightarrow  \frac{1-x^2}{\sigma^2}$ as $\mu \rightarrow 0$. 
Hence, we conclude 
\[
\mathbf{E}(t^{(1)}) = g_1(0) = \left\{ \begin{aligned}
			& \frac{1}{\sigma^2} & \text{ for } \mu = 0, \\
			& \frac{1}{\mu} - \frac{\sigma^2(1- e^{-2\mu/\sigma^2})}{2\mu^2} & \text{ for } \mu > 0.  
\end{aligned} \right.
\]
We have derived \eqref{eq:1-c-t^1-E}. 
Next, let us turn attention to the variance $\mathbf{Var}(t^{(1)})$. 

In view of $\mathbf{Var}(t^{(1)}) = \mathbf{E}[(t^{(1)})^2] - [\mathbf{E}(t^{(1)})]^2$, 
what left is to calculate $\mathbf{E}[(t^{(1)})^2]$. 
From \eqref{eq:ito-1-g-1},  
\begin{equation}\label{eq:ito-1-g-1-a^2}
\begin{aligned}
(t^{(1)})^2 = & (g(0))^2 + \left( \int_0^{t^{(1)}} \sigma \frac{\partial g}{\partial x}(\phi(s))~dW(s)\right)^2 + 2 g(0) \int_0^{t^{(1)}} \sigma \frac{\partial g}{\partial x}(\phi(s))~dW(s).
\end{aligned}
\end{equation}
Taking the expectation of \eqref{eq:ito-1-g-1-a^2}, 
and noting that $\mathbf{E}[\int_0^{t^{(1)}} \sigma \frac{\partial g}{\partial x}(\phi(s))~dW(s)] = 0$, 
we have (by Ito's isometry)
\begin{equation}\label{eq:ito-1-g-1-a^2-1}
\begin{aligned}
\mathbf{E}[(t^{(1)})^2] = & (g(0))^2 + \mathbf{E}\left( \int_0^{t^{(1)}} \sigma \frac{\partial g}{\partial x}(\phi(s))~dW(s)\right)^2  =  (g(0))^2 + \mathbf{E} \left[ \int_0^{t^{(1)}} \sigma^2 \left| \frac{\partial g}{\partial x}(\phi(s)) \right|^2~ds \right], 
\end{aligned}
\end{equation}
which, together with \eqref{eq:ito-1-g-1-a}, yields  
\begin{equation}\label{eq:ito-1-g-1-a^2-1-a}
\begin{aligned}
\mathbf{Var}(t^{(1)}) = \mathbf{E} \left[ \int_0^{t^{(1)}} \sigma^2 \left| \frac{\partial g}{\partial x}(\phi(s)) \right|^2~ds \right].
\end{aligned}
\end{equation}
It remains to calculate the right hand side of \eqref{eq:ito-1-g-1-a^2-1-a}.
 
For any subset $A$ in the interval $[0,1)$, let $1_A(x)$ be the characteristic function for $A$ (i.e., $1_A(x) = 1$ for $x$ in $A$, and $1_A(x) = 0$ for otherwise). 
Defining the measure 
\begin{equation}\label{eq:pi}
\pi(A) := \mathbf{E} \left[ \int_0^{t^{(1)}} 1_A (\phi(s))~ds \right]/ \mathbf{E}[t^{(1)}],  
\end{equation}
we rewrite \eqref{eq:ito-1-g-1-a^2-1-a} into 
\begin{equation}\label{eq:ito-1-g-1-a^2-2}
\begin{aligned}
\mathbf{Var}(t^{(1)}) = & \mathbf{E} \left[ \int_0^{t^{(1)}} \int_0^1 \sigma^2 \left| \frac{\partial g}{\partial x}(\phi(s)) \right|^2 1_{dx}(\phi(s)) \right] \\
= & \int_0^1 \sigma^2 \left| \frac{\partial g}{\partial x}(x) \right|^2 \mathbf{E} \left[ \int_0^{t^{(1)}} 1_{dx}(\phi(s))~ds \right] = \mathbf{E}(t^{(1)}) \int_0^1 \sigma^2 \left| \frac{\partial g}{\partial x}\right|^2 \pi(dx).
\end{aligned}
\end{equation}
Since there exists a probability density function $p(x)$ satisfying  
\begin{equation}\label{eq:def-p}
\pi(A) = \int_A p(x)~dx, \quad \pi(dx) = p(x)dx.
\end{equation}
we are left with the task of finding $p(x)$.
In the following, we derive $p(x)$ in two cases: (i) $\mu = 0$, (ii) $\mu > 0$.  

(i) $\mu = 0$. 
Substituting $g = 1-x$ into \eqref{eq:ito-1}, 
we calculate as  
\begin{equation}\label{eq:p-1}
\begin{aligned}
-1 & = g(1) - g(0) = \mathbf{E}(g(\phi(t^{(1)-})) - g(\phi(0))) \\
& = 0 + \underbrace{\mathbf{E}\left( \int_0^{t^{(1)}} -1dW(s) \right)}_{=0} + \mathbf{E}\left( \int_0^{t^{(1)}} -1dL(s) \right)  = - E(L(t^{(1)})). 
\end{aligned}
\end{equation}
Substituting $g(x) = e^{\lambda x}$ into \eqref{eq:ito-1}, 
noting that $\frac{\partial g}{\partial x} = \lambda e^{\lambda x}$ and $\left[ \mu \frac{\partial }{\partial x} + \frac{\sigma^2}{2} \frac{\partial^2 }{\partial x^2} \right]  g = \frac{\sigma^2}{2} \lambda^2 e^{\lambda x}$,  
we deduce 
\begin{equation}\label{eq:p-2}
\begin{aligned}
e^\lambda - 1 & = g(1) - g(0) = \mathbf{E}(g(\phi(t^{(1)-})) - g(\phi(0))) \\
& = \lambda^2 \frac{\sigma^2}{2}  \mathbf{E}\left( \int_0^{t^{(1)}} e^{\lambda \phi(s)}~ds \right) + \sigma \mathbf{E} \underbrace{\left( \int_0^{t^{(1)}}  \lambda e^{\lambda \phi(s)} dW(s) \right)}_{=0}   + \mathbf{E}\left( \int_0^{t^{(1)}} \lambda e^{\lambda \phi(s)} dL(s) \right). 
\end{aligned}
\end{equation}
Since $L(t)$ increases only when $\phi(t)=0$, 
\begin{equation}\label{eq:p-2-1}
\begin{aligned}
\mathbf{E}\left( \int_0^{t^{(1)}} \lambda e^{\lambda \phi(s)} dL(s) \right) &  = \mathbf{E}\left( \int_{\{0<s<t^{(1)} \mid  \phi(s)=0\}} \lambda e^{\lambda \cdot 0} dL(s) \right) \\
& =  \lambda E(L(t^{(1)})) = \lambda \quad (\text{by \eqref{eq:p-1}}). 
\end{aligned}
\end{equation}
It follows from \eqref{eq:p-2} and \eqref{eq:p-2-1} that 
\begin{equation}\label{eq:p-3}
\lambda^2 \frac{\sigma^2}{2}  \mathbf{E}\left( \int_0^{t^{(1)}} e^{\lambda \phi(s)}~ds \right) = e^\lambda - 1  - \lambda.
\end{equation}
Meanwhile,  
\begin{equation}\label{eq:p-4}
\begin{aligned}
& \mathbf{E}\left( \int_0^{t^{(1)}} e^{\lambda \phi(s)}~ds \right) = \mathbf{E}\left( \int_0^{t^{(1)}} \int_0^1 e^{\lambda x} 1_{dx}(\phi(s))~dx~ds \right) \\
= & \int_0^1 e^{\lambda x} \mathbf{E}\left( \int_0^{t^{(1)}}1_{dx}(\phi(s))~ds \right) =  \mathbf{E}(t^{(1)}) \int_0^1 e^{\lambda x} p(x)~dx.
\end{aligned}
\end{equation}
We have obtained $\mathbf{E}(t^{(1)}) = \sigma^{-2}$. 
It follows from \eqref{eq:p-3} and \eqref{eq:p-4} that  
\begin{equation}\label{eq:p-5}
\int_0^1 e^{\lambda x} p(x)~dx = \frac{2}{\lambda^2}(e^\lambda - 1  - \lambda) \quad \quad \text{ for all } \lambda.
\end{equation}
The left hand side of \eqref{eq:p-5} is the Laplace transform of $p(x)$, 
which implies 
\begin{equation}\label{eq:p-6}
p(x) = 2(1-x).
\end{equation}

(ii) $\mu > 0$. Via a similar argument to (i), 
instead of \eqref{eq:p-5}, one can obtain that: 
\[
\int_0^1 e^{\lambda x} (\mu \lambda + \frac{\sigma^2}{2} \lambda^2) p(x)~dx = \frac{1}{\mathbf{E}(t^{(1)})} (e^\lambda - 1  - \lambda) + \mu \lambda,
\]
with $\mathbf{E}(t^{(1)})=\frac{1}{\mu} - \frac{\sigma^2(1- e^{-2\mu/\sigma^2})}{2\mu^2}$ from \eqref{eq:1-c-t^1-E}. 
Then, one can verify that 
\begin{equation}\label{eq:p-8}
p(x) = \frac{\theta(e^\theta - e^{\theta x})}{1+\theta e^\theta - e^\theta} \quad  \text{ for } \mu > 0, \quad \theta = 2 \mu/\sigma^2.
\end{equation}

Hence, we have obtained the density function $p(x)$ for $\mu \ge 0$.  
Substituting \eqref{eq:p-8} into \eqref{eq:ito-1-g-1-a^2-2}, and together with \eqref{eq:g-1-c}, 
we obatin 
\[
\begin{aligned}
& \mathbf{Var}(t^{(1)}) =  \mathbf{E}[t^{(1)}] \int_0^1 \sigma^2 \left| \frac{\partial g}{\partial x} \right|^2 p(x)~dx =  \left\{  \begin{aligned} & \frac{2}{3\sigma^4} & \text{ for } \mu = 0, \\ 
					& \frac{-5 + e^{-2\theta} + 4e^{-\theta} + 4\theta e^{-\theta} + 2\theta}{\mu^2 \theta^2} & \text{ for } \mu > 0.
		\end{aligned} \right.
\end{aligned}
\]
Hence, we have proved \eqref{eq:1-c-t^1-Var}.  
\end{proof}
\cmag{\begin{remark}
For any $x \in [0,1]$, $g(x)$ of \eqref{eq:g-1-c} represents the expectation of the beating interval of the oscillator with initial phase $\phi(0)=x$. 
\end{remark}}
\cmag{\begin{remark}
Noting that $\phi$ is a renewal process for $t \ge t^{(k)}$ ($k = 0,1,2,\cdots$), 
according to the renewal theory (cf. \cite[Chapter 9 (1.22) (2.25)]{Cinlar13}), 
$p(x)$ of \eqref{eq:def-p} is indeed the probability density of the distribution $\phi(t)$ in $[0,1]$ as $t \rightarrow \infty$. 
Let $\tilde{p}(x,t)$ denote the probability density of the distribution of $\phi$ at time $t$. 
$\tilde{p}(x,t)$ satisfies the Fokker-Planck equation, or called the forward equation: 
\begin{subequations}\label{eq:forward-equation}
\begin{align}
& \frac{\partial \tilde{p}}{\partial t} - \left[ \frac{\sigma^2}{2}  \frac{\partial^2 }{\partial x^2} -\mu \frac{\partial }{\partial x} \right] \tilde{p} = 0 \quad (t,x) \in (0,\infty) \times (0,1), \label{eq:forward-equation-a} \\
& \tilde{p}(1,t) = 0, \label{eq:forward-equation-b} \\
& \left.\left[ \frac{\sigma^2}{2}  \frac{\partial \tilde{p}}{\partial x}(x,t) - \mu   \tilde{p} (x,t) \right]\right|_{x=0}^{x=1} = 0, \label{eq:forward-equation-c} \\
& \tilde{p}(x,0) = \delta(x), \label{eq:forward-equation-d}
\end{align}
\end{subequations}
where $\delta(x)$ denotes the Dirac Delta function and \eqref{eq:forward-equation-d} follows from the initial state of $\phi$, i.e., $\phi(0) = 0$. 
Since $\phi(t)$ jumps to $0$ immediately when approaching $1$, 
the density of $\phi(t)$ at $x=1$ is zero and the flux of the density at $x=0,1$ are equal to each other, 
which correspond to the boundary conditions \eqref{eq:forward-equation-b} and \eqref{eq:forward-equation-c} respectively. 
Moreover, \eqref{eq:forward-equation-c} ensures the conservation $\int_0^1 \tilde{p}(x,t)~dx = \int_0^1 \tilde{p}(x,0)~dx = 1$ for all $t > 0$.  
The obtention of \eqref{eq:forward-equation} follows from the classical argument (cf. \cite[\S 3.5]{Mckean69}). 
Passing to the limit $t \rightarrow \infty$, 
one can validate that $\tilde{p}(x,t)$ converges to the stationary state, i.e., the solution of 
\begin{subequations}\label{eq:forward-equation-sta}
\begin{align}
& - \left[ -\mu \frac{\partial }{\partial x}+  \frac{\sigma^2}{2}  \frac{\partial^2 }{\partial x^2} \right]  p = 0 \quad   x \in (0,1), \label{eq:forward-equation-sta-a} \\
& p(1) = 0, \quad \left[ \frac{\sigma^2}{2} \left. \frac{\partial p}{\partial x} - \mu p \right] \right|_{x=0}^{x=1} = 0, \label{eq:forward-equation-sta-b}\\
& \int_0^1 p(x)~dx = 1. \label{eq:forward-equation-sta-c}
\end{align}
\end{subequations}
One can validate that $p(x)$ given by \eqref{eq:p-6} and \eqref{eq:p-8} indeed satisfies \eqref{eq:forward-equation-sta} for $\mu=0$ and $\mu >0$, respectively. 
\end{remark}}

\section{The phase models for two coupled cells} \label{sec:2-c} 
As explained in Section~2, 
the Kuramoto model is an applicable tool to investigate the synchronization beating of two-coupled cardiomyocytes. 
The conventional Kuramoto model with noise effect for tow-coupled oscillators $\{\bar{\phi}_i\}_{i=1,2}$ is presented as follows: 
\begin{subequations}\label{eq:2-c-trad}
\begin{align}
& d \bar{\phi}_i(t) = \mu_i dt + A_{i,j} f(\bar{\phi}_j - \bar{\phi}_i) dt + \sigma_i d W_i(t), \\
& \bar{\phi}_i(0) = 0, 
\end{align}
\end{subequations}
where $i, j=1,2$, $i \neq j$, $f(x) = \sin(2\pi x)$, 
$(\mu_i,\sigma_i)$ denotes the intrinsic frequency and noise strength for cell (oscillator) $i$, 
$A_{i,j}$ the coefficient describing the strength of reaction between cell $i$ and cell $j$ (\cblue{$A_{i,j} \ge 0$}), 
and $\{W_i(t)\}_{i=1,2}$ the two independent standard Brownian motions. 

However, in general case, 
the above model may be inadequate to capture the essential properties of cardiomyocytes' synchronization. 
First, the irreversibility of beating should be taken into account. 
Second, the cardiomyocyte can be induced to beat by the neighboring cells' action potential. 
\ccyan{In addition, after beating the cardiomyocyte enters into a \emph{refractory}, 
during which the cell cannot be induced to beat.  
The length of refractory depends on the membrane potential, or more precisely, the concentrations of Ca$^{2+}$, K$^+$, Na$^+$ ions interior and exterior of the membrane. }

To incorporate the irreversibility of beating, induced beating and refractory, 
we modify the conventional Kuramoto model \eqref{eq:phi-0-dx-chi-3} as follows. 
Let $\{\phi_i\}_{i=1,2}$ be the phase of two cardiomyocytes, satisfying 
\begin{subequations}\label{eq:2-c}
\begin{align}
& d \phi_i(t) = \mu_i dt + A_{i,j} f(\phi_j - \phi_i) dt + \sigma_i d W_i(t) + dL_i(t), \\
& \phi_i(0) = 0,  
\end{align}
\end{subequations}
where the process $L_i(t)$ imposes the reflective boundary for $\phi_i$. 
Let $t_i^{(k)}$ be the $k$-th passage time that cell $i$ beats. 
Then we call $[t_i^{(k)}, t_i^{(k+1)})$ the $k$-th beating interval (or oscillation cycle) of cell $i$.   
$L_i$ is defined by:  
\begin{itemize}
\item[(L1)] $L_i$ is continuous and nondecreasing during each beating interval of $\phi_i$; 
\item[(L2)] $L_i$ increases only when $\phi_i = 0$. 
\end{itemize}
Since the refractory period associates with the membrane potential ($u$ of FN model), 
which corresponds to the phase $\phi_i$, 
for simplicity, we set a refractory threshold $B_i$ (\cblue{$0 \le B_i <1$}), 
and implement the induced beating and refractory by: 
\begin{itemize}
\item[(IND)] If cell $i$ is out of refractory, cell $i$ beats promptly when the neighbor (cell $j$) beats spontaniously, 
in other words, if $\phi_i(t-) > B_i$ and $\phi_j(t-) = 1$, then $\phi_j(t) = \phi_i(t) = 0$ (both two phases jump to $0$ after beating to start a new oscillation cycle); 
\item[(REF)] If cell $i$ is in refractory and the neighbor cell $j$ beats spontaniously, then cell $i$ will not be induced to beat, namely, if $\phi_i(t-) \le B_i$ and $\phi_j(t-) = 1$, then $\phi_j(t) = 0$ and $\phi_i(t) =  \phi_i(t-)$ ($\phi_j$ jumps to $0$ but $\phi_i$ keeps going). 
\end{itemize}
%
\cblue{\begin{remark} \label{rk:4-1}
One weak point of the conventional model \eqref{eq:2-c-trad} is that the synchronization has been treated ``ambiguous'' or ``approximately'', 
because the possibility of $\bar{\phi}_i(t) = \bar{\phi}_j(t) = 1$ is zero, 
and one can only expect that both two cells beat with tiny time-delay, namely, $\bar{\phi}_i(t_i) = 1$, $\bar{\phi}_j(t_j) = 1$ with $|t_i - t_j| \approx 0$. 
To guarantee this ``approximated'' synchronization, 
one should take sufficiently small noise strength $\sigma_i$ and large enough reaction coefficient $A_{i,j}$ (see Section~\ref{sec:2-c-2}). 
\end{remark}}

Thanks to the induced beating (IND), we have a rigorous mathematical definition of the synchronization. 
Let $\mathbb{t}^{(k)}$ denote the time of $k$-th synchronized beating, i.e., 
\begin{equation}\label{eq:t-1-syn}
\mathbb{t}^{(k)} := \inf \{ t > \mathbb{t}^{(k-1)} \ : \ \phi_i(t-) = 1, \ \phi_j(t) > B_j, \ i = 1 \text{ or } 2, \ j \neq i\} \quad (\mathbb{t}^{(0)}= 0).
\end{equation}
In view of $\phi_1(\mathbb{t}^{(k)}) = \phi_2(\mathbb{t}^{(k)}) = 0$, 
$\Phi(t) :=(\phi_1(t), \phi_2(t))$ is a renewal stochastic process for $t > \mathbb{t}^{(k)}$ $(k=0,1,2,\cdots)$. 
Therefore, the beating intervals $\{ \mathbb{t}^{(k+1)}-\mathbb{t}^{(k)} \}_{k\ge 0}$ are independent and identically distributed (i.i.d.).  
To obtain the expected value and variance of synchronized beating interval, 
we only need to investigate $\mathbb{t}^{(1)}$. 

In Figure~\ref{fig:2-c} (a)(b) and (c)(d), 
we plot two examples of $(\bar{\phi}_1, \bar{\phi}_2)$ and $(\phi_1, \phi_2)$ respectively. 
For small noise strength and large reaction coefficients, for example, $\sigma_1 = \sigma_2 = 0.2$, $A_{1,2} = A_{2,1} = 6$, 
the conventional model \eqref{eq:2-c-trad} (Figure~\ref{fig:2-c} (b)) and the proposed model \eqref{eq:2-c} (L1)(L2)(IND)(REF) (Figure~\ref{fig:2-c} (d)) have similar solution behavor.  
In view of Figure~\ref{fig:2-c} (b), the non-positive phase ($\bar{\phi}_i \le 0)$ is ignorable, 
and the time-delay between two cells' beating is very tiny. 
Therefore, the roles of the reflective boundary and induced beating of our model are negligible. 
However, when the noise strength is not so small and the reaction coefficients is not large enough, 
the conventional model $(\bar{\phi}_1, \bar{\phi}_2)$ may have no synchronization (see Figure~\ref{fig:2-c} (a)).  
wherea Figure~\ref{fig:2-c} (c) shows the synchronization owing to the induced beating, 
and the significant role of the reflective boundary.   

\begin{figure}[h]
 \begin{center}
     \begin{overpic}[width=0.43\linewidth]{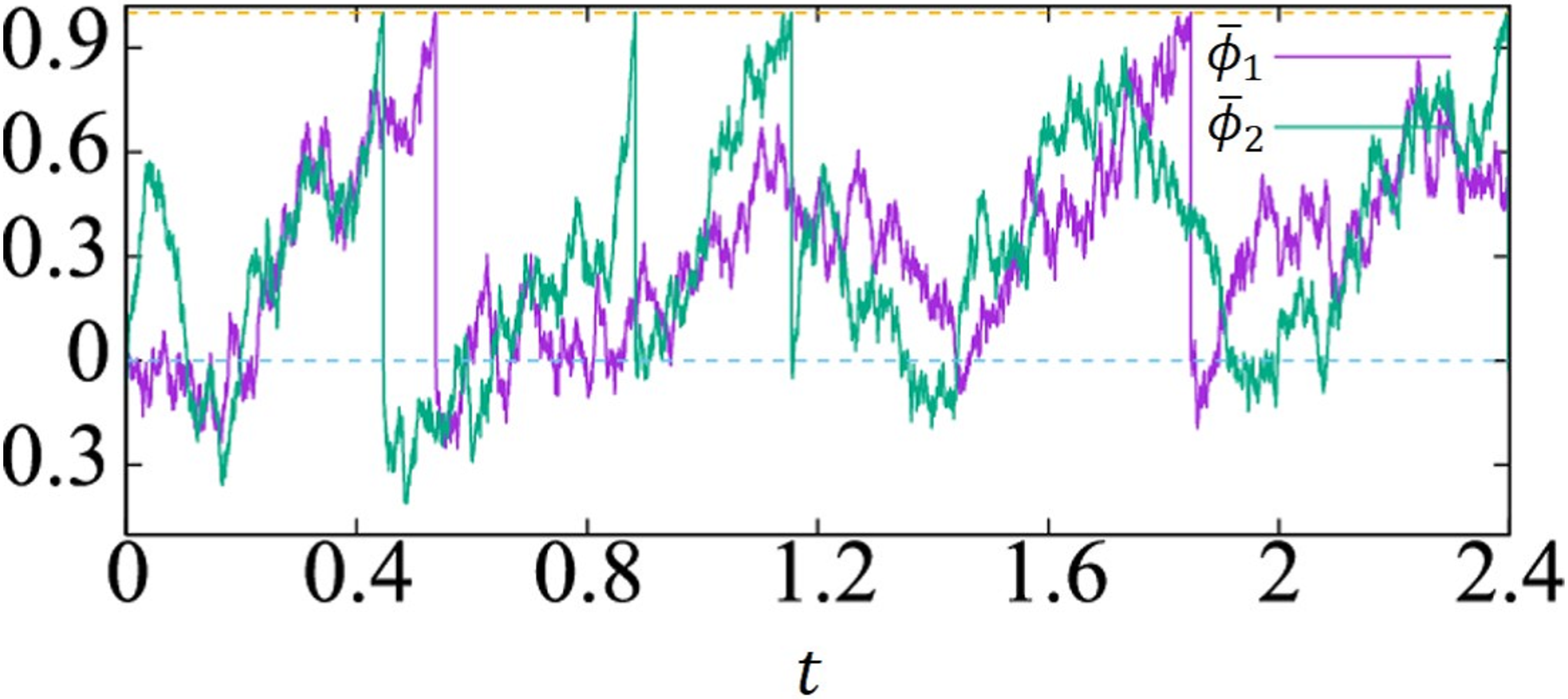}
     \put(-4,48){\small \bf (a)}
     \end{overpic}
 \hspace*{0.5cm}
%
     \begin{overpic}[width=0.43\linewidth]{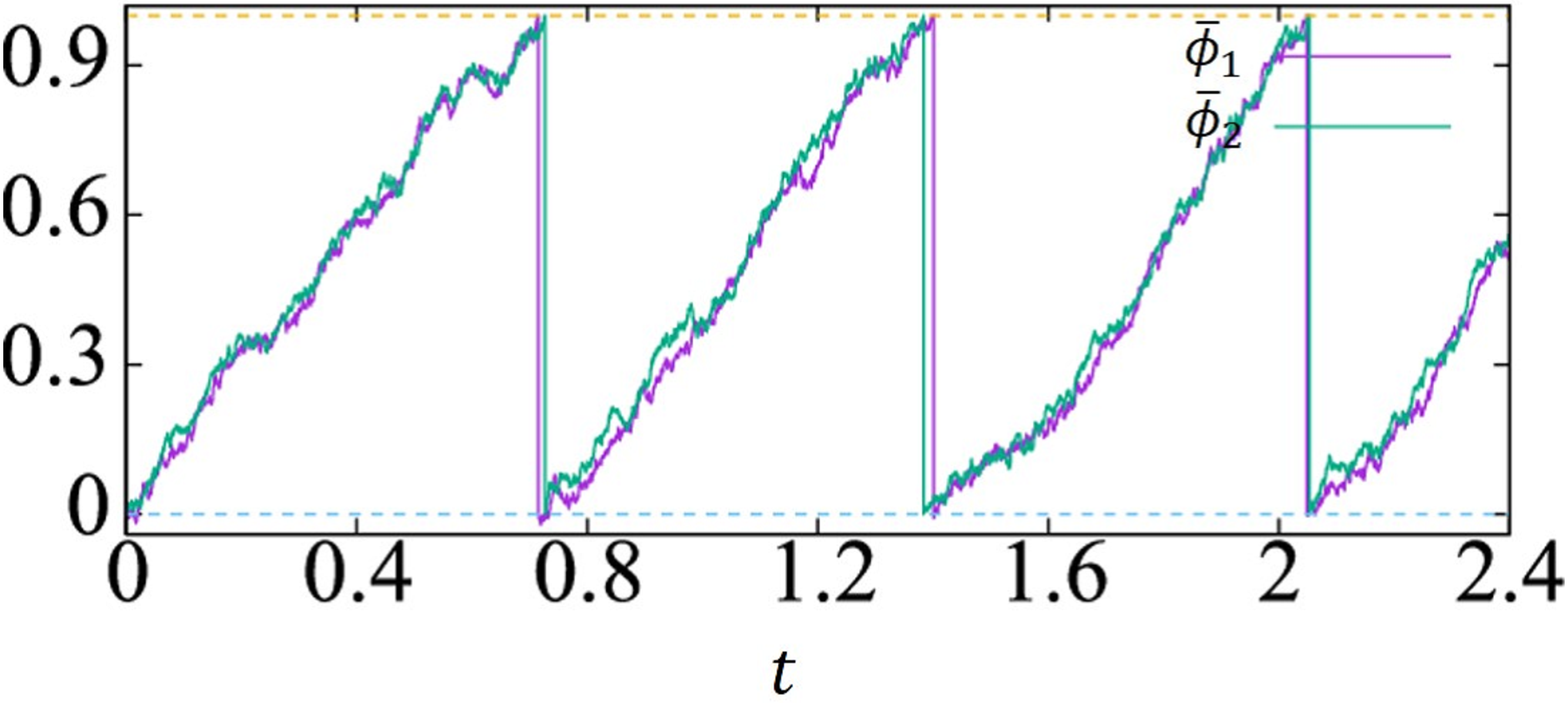}
     \put(-4,48){\small \bf (b)}
     \end{overpic}
     \begin{overpic}[width=0.43\linewidth]{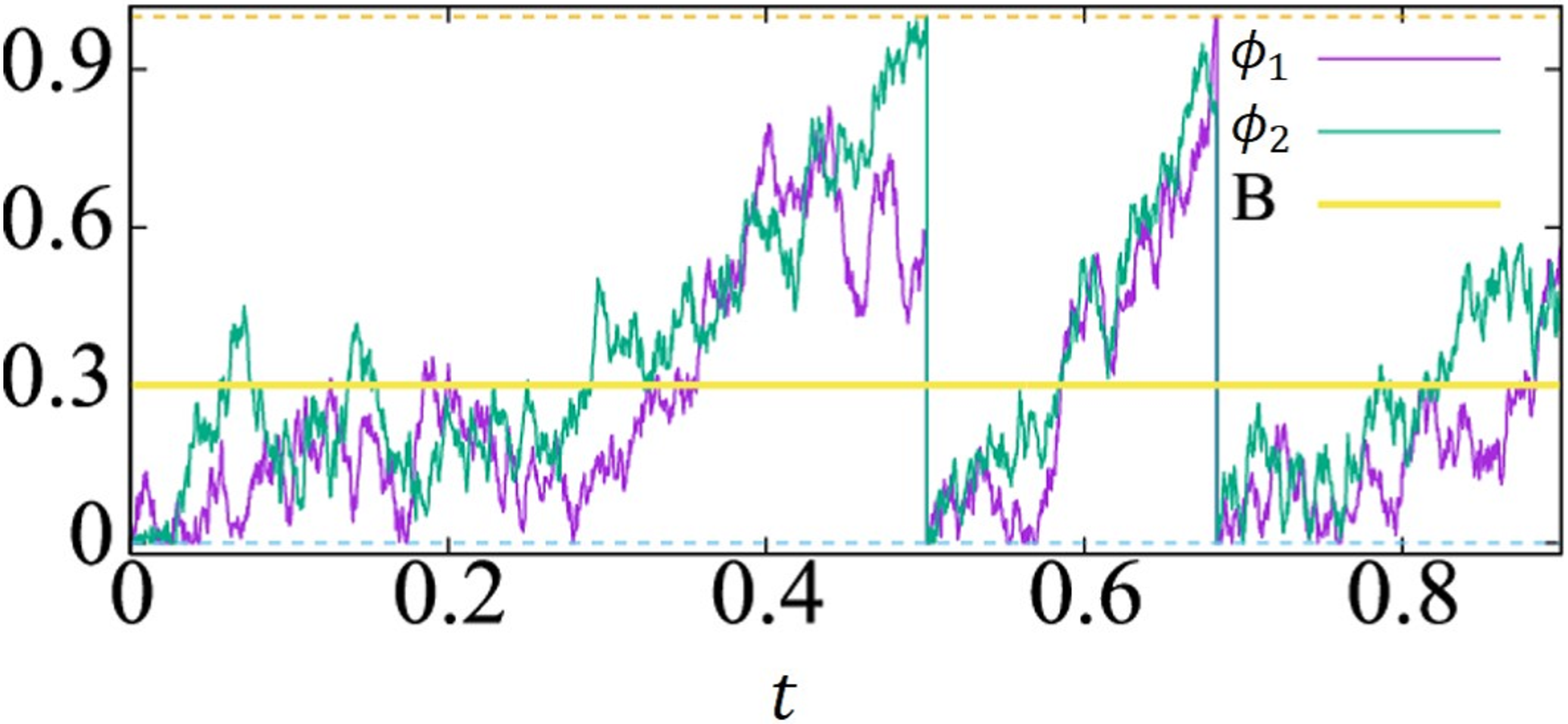}
     \put(-4,48){\small \bf (c)}
     \end{overpic}
 \hspace*{0.5cm}
%
     \begin{overpic}[width=0.43\linewidth]{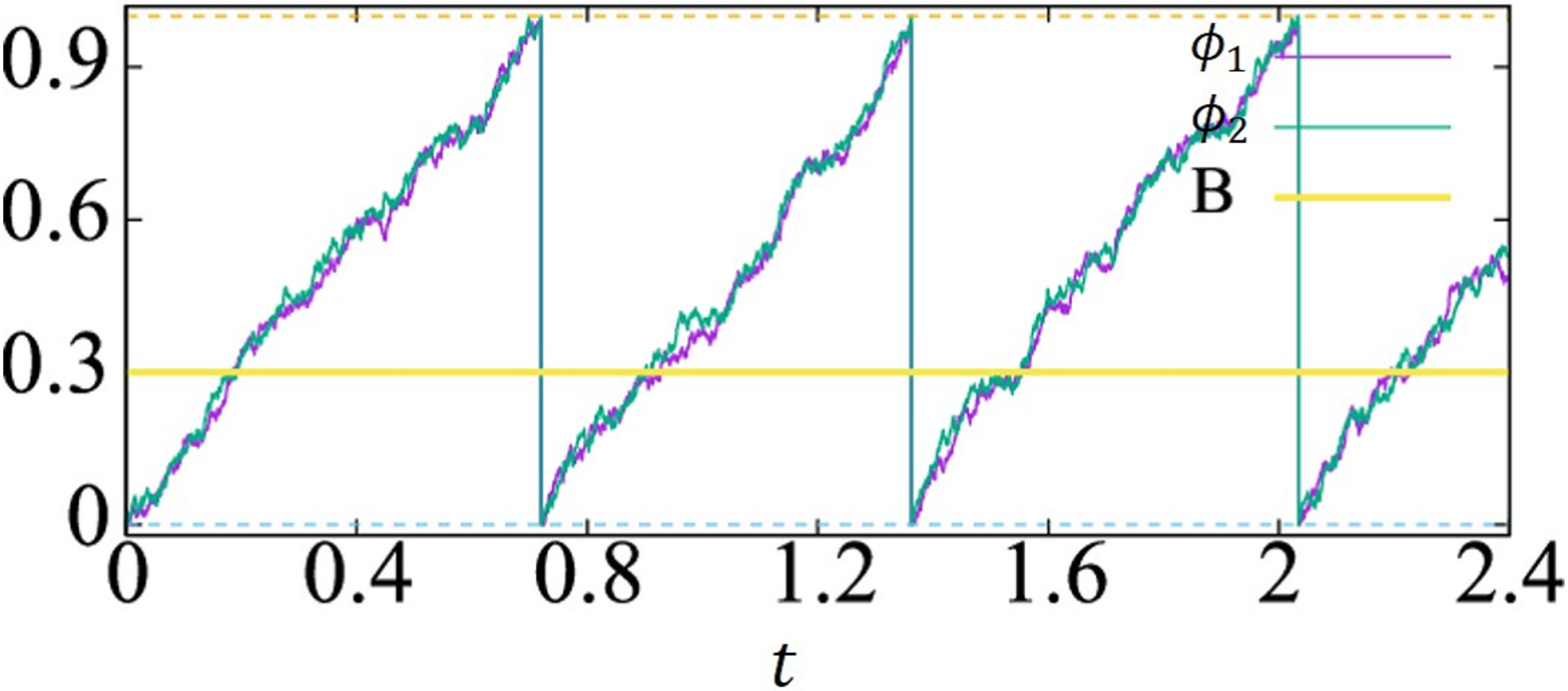}
     \put(-4,48){\small \bf (d)}
     \end{overpic}
 \end{center}
\caption{The trajectories of $(\bar{\phi}_1, \bar{\phi}_2)$ and $(\phi_1, \phi_2)$ with $(\mu_1,\mu_2) = (1,2)$, $A_{1,2} = A_{2,1} = A$, $\sigma_1= \sigma_2 = \sigma$, $B_1 = B_2 = B = 0.3$. (a)(c) $(\sigma, A) = (1, 2)$. (b)(d) $(\sigma, A) = (0.2,6)$.}
\label{fig:2-c}
\end{figure}

We intend to calculate the expected value and variance for the synchronized beating interval $\mathbb{t}^{(1)}$. 
First, for the proposed model, applying the Ito's calculus and the renewal theory, 
we derive the PDEs associated with $\mathbf{E}(\mathbb{t}^{(1)})$  and $\mathbf{Var}(\mathbb{t}^{(1)})$. 
However, the closed-form of the PDEs' solutions are non-trivial to derive. 
Next, we consider the case that the role of the reflective boundary and induced beating is negligible, 
where the conventional and proposed models have little difference, 
and we obtain the relationship between the parameters and the CV of beating intervals.    
\subsection{The expectation and variance of the synchronized beating interval}
Besides of the synchronized beating (induced beating (IND)), 
let us pay attention to the single-beating (REF), or called independent-beating,  
where only one cell is beating and the other is in refractory.  
Assume that before the first synchronized beating, 
cell $i$ beats independently for $J_i$ times. 
Let $t_i^{(m)}$ $(m \le J_i)$ be the passage time of $m$-th single-beating of cell $i$, that is, 
\begin{equation}\label{eq:2-c-t_i^m}
t_i^{(m)} := \inf \{ t > t_i^{(m-1)} \ : \ \phi_i(t-) = 1, \ \phi_j(t)<B_j, \ t<\mathbb{t}^{(1)}\}, \quad t_i^{(0)} = 0 \quad (j \neq i). 
\end{equation}
For $t < \mathbb{t}^{(1)}$, $\phi_i(t)$ is a stochastic process with jump at $\{t_i^{(m)}\}_{m=1}^{J_i}$, 
where $\phi_i(t_i^{(m)}-)=1$ and $\phi_i(t_i^{(m)})=0$. 
Setting $\Phi(t):=(\phi_1(t), \phi_2(t))$, 
and applying the It\^{o}'s formula for the stochastic process with jumps (cf. \cite{Harrison85}), 
we have:
for any $g(x_1,x_2)$ $(0 \le x_1,x_2 \le 1)$ with continuous differential $\frac{\partial g}{\partial x_i}$ and $\frac{\partial^2 g}{\partial x_i \partial x_j}$ ($i,j = 1,2$), 
\begin{equation}\label{eq:ito-2}
\begin{aligned}
& g(\Phi(\mathbb{t}^{(1)}-))  = g(\Phi(0)) + \int_0^{\mathbb{t}^{(1)}} \sum_{i = 1}^2  \left[ (\mu_i + A_{i,j} f(x_j- x_i) )\frac{\partial}{\partial x_i} + \frac{\sigma_i^2}{2} \frac{\partial^2}{\partial x_i^2} \right] g(\Phi(s))~ds \\
& \quad + \int_0^{\mathbb{t}^{(1)}} \sum_{i=1}^2 \sigma_i \frac{\partial g}{\partial x_i}(\Phi(s))~dW_i(s) + \sum_{i=1}^2 \int_0^{\mathbb{t}^{(1)}} \frac{\partial g}{\partial x_i}(\Phi(s))~dL_i(s) \\
& \quad - \sum_{m=1}^{J_1} \left[ g(1,\phi_2(t_1^{(m)}-)) - g(0,\phi_2(t_1^{(m)})) \right] -  \sum_{m=1}^{J_2} \left[ g(\phi_1(t_2^{(m)}-), 1) - g(\phi_1(t_2^{(m)}), 0) \right],
\end{aligned}
\end{equation}
where $\phi_i(t_j^{(m)}) = \phi_i(t_j^{(m)}-) \le B_i$, that is, cell $j$ is in refractory while cell $i$ is beating ($\phi_i(t_i^{(m)}-)=1$).  

Since $L_i(t)$ is non-decreasing in the time intervals $\{(t_i^{(m)}, t_i^{(m+1)})\}_{m=0}^{J_{i-1}}$ and $(t_i^{(J_i)}, \mathbb{t}^{(1)})$, 
and $L_i(t)$ increases only when $\phi_i(t) = 0$, 
we see that 
\begin{equation}\label{eq:ito-2-a}
\int_0^{\mathbb{t}^{(1)}} \frac{\partial g}{\partial x_i}(\Phi(s))~dL_i(s) = \int_{\{0<s<\mathbb{t}^{(1)} \mid \phi_i(s)=0\}} \frac{\partial g}{\partial x_i}(\Phi(s))~dL_i(s)
\end{equation} 
Let $g$ be the solution of \eqref{eq:2-c-g}:
\begin{subequations}\label{eq:2-c-g}
\begin{align}
& \sum_{i = 1}^2  \left[ (\mu_i + A_{i,j} f(x_j- x_i) )\frac{\partial}{\partial x_i} + \frac{\sigma_i^2}{2} \frac{\partial^2}{\partial x_i^2} \right]  g = -1 \quad & \text{ for } 0 < x_1,x_2 < 1, \label{eq:2-c-g-1}  \\
& \frac{\partial g}{\partial x_i} = 0 \quad & \text{ for } x_i = 0, \ i = 1,2, \label{eq:2-c-g-b} \\
& g(1,x_2) = 0 \quad & \text{ for } B_2 < x_2 \le 1, \label{eq:2-c-g-c} \\ 
& g(x_1,1) = 0 \quad & \text{ for }  B_1 < x_1 \le 1, \label{eq:2-c-g-d} \\
& g(x_1,0) = g_2(x_1,1) \quad & \text{ for } 0 \le x_1 \le B_1, \label{eq:2-c-g-e} \\
& g(0,x_2) = g_2(1,x_2) \quad & \text{ for } 0 \le x_2 \le B_2. \label{eq:2-c-g-f}
\end{align}
\end{subequations} 
It follows from \eqref{eq:2-c-g-b} and \eqref{eq:ito-2-a} that 
\[
\int_0^{\mathbb{t}^{(1)}} \frac{\partial g}{\partial x_i}(\Phi(s))~dL_i(s) = 0.
\]
Since the induced beating happens at $t = \mathbb{t}^{(1)}$, 
we have $\phi_1(\mathbb{t}^{(1)}-) = 1, \phi_2(\mathbb{t}^{(1)}-) > B_2$ or $\phi_2(\mathbb{t}^{(1)}-) = 1, \phi_1(\mathbb{t}^{(1)}-) > B_1$,  
which, together with \eqref{eq:2-c-g-c} and \eqref{eq:2-c-g-d}, implies 
\[
g(\Phi(\mathbb{t}^{(1)}-)) = 0. 
\]
Furthermore, in view of $\phi_i(t_j^{(m)}) = \phi_i(t_j^{(m)}-)  \le B_i$, \eqref{eq:2-c-g-e} and \eqref{eq:2-c-g-f} guarantee that 
\[
\sum_{m=1}^{J_1} \left[ g(1,\phi_2(t_1^{(m)}-)) - g(0,\phi_2(t_1^{(m)})) \right] = 0, \quad   \sum_{m=1}^{J_2} \left[ g(\phi_1(t_2^{(m)}-), 1) - g(\phi_1(t_2^{(m)}), 0) \right] = 0. 
\]
With $\Phi(0) = (0,0)$ and \eqref{eq:2-c-g-1}, 
we rewrite \eqref{eq:ito-2} into:
\begin{equation}\label{eq:ito-2-b}
\begin{aligned}
0 = g(0,0) - \mathbb{t}^{(1)} + \int_0^{\mathbb{t}^{(1)}} \sum_{i=1}^2 \sigma_i \frac{\partial g}{\partial x_i}(\Phi(s))~dW_i(s) 
\end{aligned}
\end{equation}
Taking the expectation, we have 
\begin{equation}\label{eq:ito-2-c}
\mathbf{E}[\mathbb{t}^{(1)}] = g(0,0)  + \underbrace{\mathbf{E} \left[\int_0^{\mathbb{t}^{(1)}} \sum_{i=1}^2 \sigma_i \frac{\partial g}{\partial x_i}(\Phi(s))~dW_i(s)  \right]}_{=0} = g(0,0). 
\end{equation}
Hence, the obtention of $\mathbf{E}(\mathbb{t}^{(1)})$ reduces to solve the PDE \eqref{eq:2-c-g} 
Next, let us turn attention to the variance $\mathbf{Var}[\mathbb{t}^{(1)}]$. 
From \eqref{eq:ito-2-b}, we have 
\[
(\mathbb{t}^{(1)})^2 = (g(0,0))^2  + \left[\int_0^{\mathbb{t}^{(1)}} \sum_{i=1}^2 \sigma_i \frac{\partial g}{\partial x_i}(\Phi(s))~dW_i(s) \right]^2 + 2 g(0,0)\int_0^{\mathbb{t}^{(1)}} \sum_{i=1}^2 \sigma_i \frac{\partial g}{\partial x_i}(\Phi(s))~dW_i (s).
\]
Taking the expectation of the above equation yields:  
\[
\begin{aligned}
\mathbf{E}[(\mathbb{t}^{(1)})^2] & = (g(0,0))^2  +  \mathbf{E} \left[\int_0^{\mathbb{t}^{(1)}} \sum_{i=1}^2 \sigma_i \frac{\partial g}{\partial x_i}(\Phi(s))~dW_i(s) \right]^2 + 2 g(0,0)  \mathbf{E} \underbrace{\left[ \int_0^{\mathbb{t}^{(1)}} \sum_{i=1}^2 \sigma_i \frac{\partial g}{\partial x_i}(\Phi(s))~dW_i (s) \right]}_{=0} \\
& = (g(0,0))^2  + \mathbf{E}\left[ \int_0^{\mathbb{t}^{(1)}} \sum_{i=1}^2 \sigma_i^2 \left|\frac{\partial g}{\partial x_i}(\Phi(s))\right|^2~ds \right] \quad \text{(by Ito's isometry)}, 
\end{aligned}
\]
which, together with $\mathbf{Var}[\mathbb{t}^{(1)}] = \mathbf{E}[(\mathbb{t}^{(1)})^2] - (\mathbf{E}[(\mathbb{t}^{(1)})])^2$ and \eqref{eq:ito-2-c}, implies  
\begin{equation}\label{eq:ito-2-d}
\begin{aligned}
\mathbf{Var}[\mathbb{t}^{(1)}]  =  \mathbf{E}\left[ \int_0^{\mathbb{t}^{(1)}} \sum_{i=1}^2 \sigma_i^2 \left|\frac{\partial g}{\partial x_i}(\Phi(s))\right|^2~ds \right].  
\end{aligned}
\end{equation}
Noting that $\Phi(t) = (\phi_1(t), \phi_2(t))$ returns to $(0,0)$ after every synchronization, 
$\Phi(t)$ is a renewal process. 
According to the renewal theory \cite{Cinlar13}, the right-hand side of \eqref{eq:ito-2-d} is evaluated as:    
\[
\begin{aligned}
\mathbf{Var}(\mathbb{t}^{(1)}) = \mathbf{E}(\mathbb{t}^{(1)}) \int_0^1 \int_0^1 \sum_{i=1}^2 \sigma_i^2 \left|\frac{\partial g}{\partial x_i}\right|^2 p~dx_1dx_2, 
\end{aligned}
\]
where $p(x_1,x_2)$ denotes the distribution density of $\Phi = (\phi_1,\phi_2) \in [0,1)^2$ as $t \rightarrow \infty$.  

\cmag{To derive the equations concerned about $p(x_1,x_2)$, we interpret the model \eqref{eq:2-c} (L1)(L2)(IND)(REF) from a physical point of view. 
Regard $(\phi_1,\phi_2)$ as the position of a particle in $[0,1)^2$,   
which moves with velocity $[\mu_1 + A_{1,2}f(\phi_2(t)-\phi_1(t)), \mu_2 + A_{2,1}f(\phi_1(t)-\phi_2(t))]^\top$, 
and is effected by noise $[\sigma_1 dW_1(t), \sigma_2 dW_2(t)]^\top$. 
The initial position of the particle is $(0,0)$. 
If the particle reaches the boundary $\{x_1 = 1, \ x_2 > B_2\} \cup \{x_1 > B_1, \ x_2 = 1\}$, 
then it jumps to point $(0,0)$ immediately. 
On the other hand, if the particle approaches the boundary $\{x_1 = 1, \ 0 <  x_2 \le B_2\}$ (resp. $\{0<x_1 \le B_1, \ x_2 = 1\}$), 
it jumps to position $(0, x_2)$ (resp. $(x_1, 0)$) instantly.  
Moreover, the movement reflects when touching the boundary $\{(x_1,x_2) \mid x_1 = 0 \text{ or } x_2 = 0\}$.  
In Figure~\ref{fig:Phi} (a)(b), we plot two trajectories of the particle.  
}

\begin{figure}[h]
 \begin{center}
     \begin{overpic}[width=0.4\linewidth]{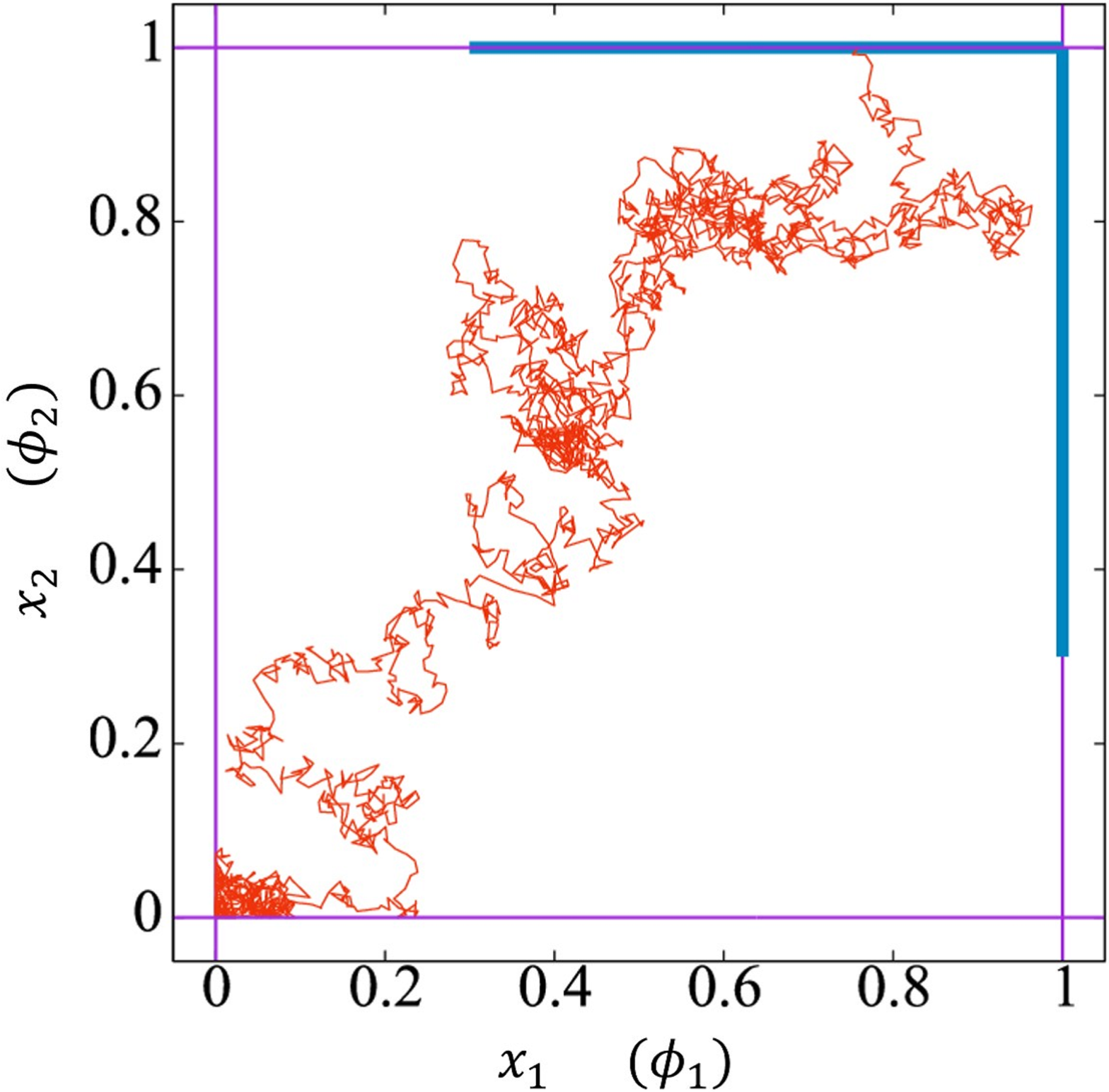}
     \put(-1,95){\small \bf (a)}
     \end{overpic}
 \hspace*{0.5cm}
%
     \begin{overpic}[width=0.4\linewidth]{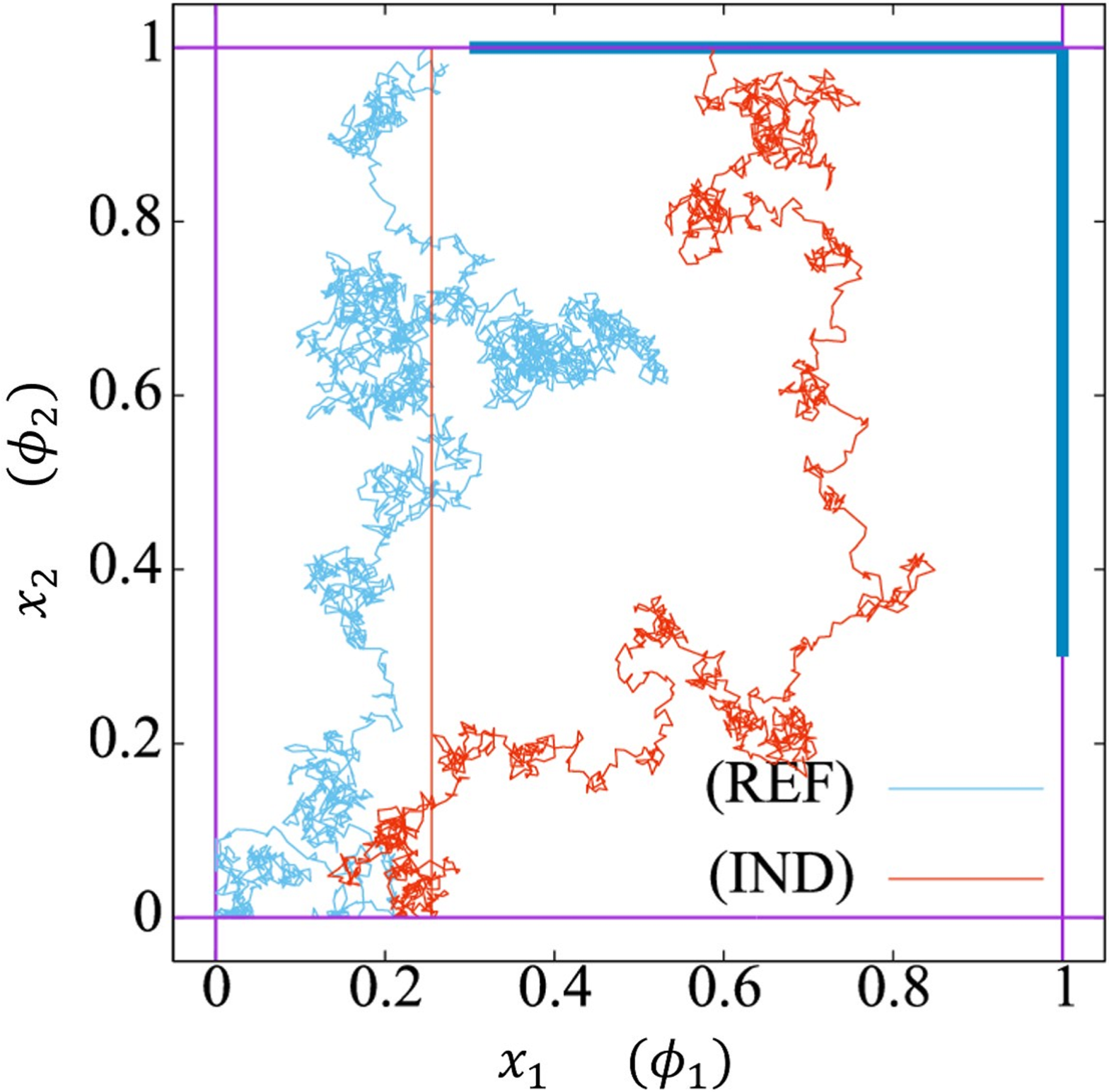}
     \put(-1,95){\small \bf (b)}
     \end{overpic}
 \end{center}
\caption{\cmag{Two realizations of the trajectories $(\phi_1, \phi_2)$ with $B_1=B_2 = 0.3$. The reflective boundary is imposed on $\{x_1 = 0\} \cup \{x_2 = 0\}$. (a) The particle touches the boundary (blue color) $\{ x_1 > B_1, x_2 = 1 \}$, which will jumps to point $(0, 0)$. (b) The particle first touches the boundary $\{ 0< x_1 \le B_1, x_2 = 1 \}$ (cyan trajectory), and jumps to point $(x_1, 0)$. Then, the particle approaches $\{ x_1 > B_1, x_2 = 1 \}$ (red trajectory), which will jumps to point $(0, 0)$. The red trajectories $(\phi_1, \phi_2)$ correspond to the induced bearing (IND),} \cmag{whereas the cyan one corresponds to the independent beating (REF).}}
\label{fig:Phi}
\end{figure}

\cmag{Therefore, $p(x_1,x_2)$ represents the distribution density of the particle in $[0,1)^2$ as $t \rightarrow \infty$. 
Now, let $\tilde{p}(x_1,x_2, t)$ denote the distribution density of the particle at time $t$. 
Via a similar argument to \cite[see Section 3.5]{Mckean69}, 
one can prove that $\tilde{p}(x_1,x_2, t)$ satisfies the following Fokker-Planck equation (or the forward equation): for 
$i,j = 1,2$, $i \neq j$,
\begin{subequations}\label{eq:p-2c}
\begin{align}
& 
	\frac{\partial \tilde{p}}{\partial t}  + \sum_{i=1}^2 \frac{\partial}{\partial x_i} \tilde{\mathcal{F}}_i = \delta(x_1,x_2) \sum_{i=1}^2 \int_{\{x_i=1, \ B_j \le x_j < 1\}} \tilde{\mathcal{F}}_i~ds
	& \text{ for } 0<x_1,x_2<1, \ t > 0, \label{eq:p-2c-a} \\
& \tilde{p}(x_1,x_2,t) = 0 & \text{ on } \{x_1=1\} \cup \{x_2 = 1\} , \label{eq:p-2c-b} \\
& \tilde{\mathcal{F}}_1 = 0 
	& \text{ on }  \{x_1=0, B_2 \le x_2 < 1 \}, \label{eq:p-2c-c} \\
& \tilde{\mathcal{F}}_2 = 0 
	& \text{ on }  \{x_2=0, B_1 \le x_1 < 1 \}, \label{eq:p-2c-d} \\
& \tilde{\mathcal{F}}_1|_{x_1=0}^{x_1=1} = 0 
	&  \text{ for }   \{0 < x_2 \le B_2\}, \label{eq:p-2c-e} \\
& \tilde{\mathcal{F}}_2|_{x_2=0}^{x_2=1} = 0 
	& \text{ for }  \{0 < x_1 \le B_1\}, \label{eq:p-2c-f} \\
& \tilde{p}(x_1,x_2,0) = \delta(x_1,x_2),  & \text{ for } \{0<x_1,x_2<1\},  \label{eq:p-2c-g}
\end{align}
\end{subequations}
where $\tilde{\mathcal{F}}_i : = (\mu_i + A_{i,j}f(x_j-x_i))\tilde{p} - \frac{\sigma_i^2}{2} \frac{\partial \tilde{p}}{\partial x_i}$ denotes the $i$-th component of flux, and $\delta(x_1,x_2)$ the Dirac Delta function. 
\eqref{eq:p-2c-g} means the initial position of the particle is $(0,0)$.  
Since the particle jumps to $(0,0)$ or $\cup_{i,j=1,2}^{i \neq j}\{ x_i = 0, 0 < x_j < B_j \}$ instantly when touching the boundary $\cup_{i=1,2} \{x_i = 1\}$, 
the density of the particle on $\cup_{i=1,2} \{x_i = 1\}$ is zero, namely \eqref{eq:p-2c-b}, 
which implies 
\[
\sum_{i=1}^2 \int_{\{x_i=1, B_j \le x_j < 1\}} -\frac{\sigma_i^2}{2}\frac{\partial \tilde{p}}{\partial x_i}~ds = \sum_{i=1}^2 \int_{\{x_i=1, B_j \le x_j < 1\}} \tilde{\mathcal{F}}_i ~ds.   
\]
Hence, the right-hand side of \eqref{eq:p-2c-a} represents the total flux of $\tilde{p}$ which touches the boundary $\{x_1=1, B_2 \le x_2 < 1\} \cup \{B_1 \le x_1 < 1, x_2=1\} $ and then jumps to point $(0,0)$ immediately. 
Putting together with the boundary conditions \eqref{eq:p-2c-c}--\eqref{eq:p-2c-g}, 
and in view of $\int_{[0,1)^2} \tilde{p}(x_1,x_2,0)~dx = \int_{[0,1)^2} \delta(x_1,x_2) ~dx =1$, 
one can validate the conservation law:  
\[
\frac{d}{dt}  \int_{[0,1)^2} \tilde{p}(x_1,x_2, t)~dx = 0 \quad (\text{equivalently } \int_{[0,1)^2} \tilde{p}(x_1,x_2, t)~dx = 1) \quad \text{ for all } t > 0. 
\]
The zero-flux boundary conditions \eqref{eq:p-2c-c}, \eqref{eq:p-2c-d} correspond to the reflective boundary, 
whereas \eqref{eq:p-2c-e} (resp. \eqref{eq:p-2c-f}) describes that the particle jumps from $(1, x_2)$ to $(0, x_2)$ with $0 < x_2 \le B_2$ (resp. from $(x_1,x_2) = (x_1, 1)$ to $(x_1, 0)$ with $0 < x_1 \le B_1$). 
}

As $t \rightarrow \infty$, one can show that $\tilde{p}(x_1,x_2, t)$ converges to the stationary state $p(x_1,x_2)$, 
which satisfies: for $i,j = 1,2$, $i \neq j$,
\begin{subequations}\label{eq:p-2c-s}
\begin{align}
& 
	\sum_{i=1}^2 \frac{\partial}{\partial x_i} \mathcal{F}_i = \delta(x_1,x_2)  \sum_{i=1}^2 \int_{\{x_i=1, \ B_j \le x_j < 1\}}  \mathcal{F}_i ~ds 
   & \text{ for } 0<x_1,x_2<1, \label{eq:p-2c-s-a} \\
& p(x_1,x_2) = 0 & \text{ on } \{x_1=1\} \cup \{x_2 = 1\} , \label{eq:p-2c-s-b} \\
& \mathcal{F}_1 = 0  
	& \text{ on }  \{x_1=0, B_2 \le x_2 < 1\}, \label{eq:p-2c-s-c} \\
& \mathcal{F}_2 = 0  
	& \text{ on }  \{x_2=0, B_1 \le x_1 < 1\}, \label{eq:p-2c-s-d} \\
& \mathcal{F}_1|_{x_1=0}^{x_1=1} = 0  
	&  \text{ for }   \{0 < x_2 \le B_2\}, \label{eq:p-2c-s-e} \\
&  \mathcal{F}_2|_{x_2=0}^{x_2=1} = 0 
	& \text{ for }  \{0 < x_1 \le B_1\}, \label{eq:p-2c-s-f} 
\end{align}
\end{subequations}
where $\mathcal{F}_i : = (\mu_i + A_{i,j}f(x_j-x_i))p - \frac{\sigma_i^2}{2} \frac{\partial p}{\partial x_i}$ denotes the $i$-th component of flux ($\mathcal{F}_i = - \frac{\sigma_i^2}{2} \frac{\partial p}{\partial x_i}$ on $\{x_i = 1\}$ by \eqref{eq:p-2c-s-b}).  
As the stationary state of $\tilde{p}$, $p$ also satisfies $\int_{[0,1)^2}p~dx_1dx_2 = 1$. 

From the above argument, we conclude: 
\begin{proposition}\label{prop:2C-E-Var}
The expectation and variance of the synchronized beating interval are given by 
\begin{subequations}\label{eq:2c-t-syn-E-Var}
\begin{align}
& \mathbf{E}(\mathbb{t}^{(1)}) = g(0,0), \label{eq:2c-t-syn-E-Var-a} \\
& \mathbf{Var}(\mathbb{t}^{(1)}) = \mathbf{E}(\mathbb{t}^{(1)}) \int_{[0,1)^2} \sum_{i=1}^2 \sigma_i^2 \left| \frac{\partial g}{\partial x_i} \right|^2 p~ dx_1 dx_2, \label{eq:2c-t-syn-E-Var-b}
\end{align}
\end{subequations}
where $g, p$ are the solutions to \eqref{eq:2-c-g} and \eqref{eq:p-2c-s}, respectively. 
\end{proposition}
In the case of single isolated cardiomyocyte (Section~3), 
we have obtained $g$ and $p$ in closed-form. 
However, it is non-trivial to solve the two-dimensional PDEs \eqref{eq:2-c-g} and \eqref{eq:p-2c}.  
On the elementary case that $\mu_1 = \mu_2 = 0$, $A_{1,2} = A_{2,1} = 0$ and $B_1 = B_2 = 0$, 
\eqref{eq:2-c-g} is reduced to: 
\begin{subequations}\label{eq:2-c-g-elementary}
\begin{align}
& \sum_{i = 1}^2   \frac{\sigma_i^2}{2} \frac{\partial^2 g}{\partial x_i^2} = -1 \quad & \text{ for } 0 < x_1,x_2 < 1, \label{eq:2-c-g-elementary-a}  \\
& \frac{\partial g}{\partial x_i} = 0 \quad & \text{ for } x_i = 0, \ i = 1,2, \label{eq:2-c-g-elementary-b} \\
& g(1,x_2) = 0 \quad & \text{ for } 0 < x_2 \le 1, \label{eq:2-c-g-elementary-c} \\ 
& g(x_1,1) = 0 \quad & \text{ for }  0 < x_1 \le 1. \label{eq:2-c-g-elementary-d}
\end{align}
\end{subequations} 
Apparently, the eigenvalues $\{\lambda_{mn}\}_{m,n=0}^\infty$ and eigenfunctions  $\{g_{mn}\}_{m,n=0}^\infty$ for theoperator -$\sum_{i=1}^2 \frac{\sigma_i^2}{2} \frac{\partial^2}{\partial x_i^2}$ under the boundary conditions \eqref{eq:2-c-g-elementary-b}--\eqref{eq:2-c-g-elementary-d} are given by: 
\[
\lambda_{mn} = \frac{\sigma_1^2}{2} (m+ \frac{1}{2})^2 \pi^2 + \frac{\sigma_2^2}{2} (n+ \frac{1}{2})^2 \pi^2 , \quad g_{mn} = \cos((m+ \frac{1}{2}) \pi x_1 )  \cos((n+ \frac{1}{2}) \pi x_2 ). 
\]
Then, there exist constants $\{a_{mn}\}_{m,n = 0}^\infty$ such that $g=\sum_{m,n=0}^\infty a_{mn} g_{mn} $ is the solution of \eqref{eq:2-c-g-elementary}.  
Substituting $g=\sum_{m,n=0}^\infty a_{mn} g_{mn} $ into \eqref{eq:2-c-g-elementary-a}, 
and calculating the integration $\int_{[0,1)^2} \eqref{eq:2-c-g-elementary-a}  \times g_{m'n'}~dx_1dx_2$ for $m',n' = 0,1,2,\cdots$, 
one can derive that  
\[
a_{mn} = \frac{4(-1)^{m+n}}{\lambda_{mn} (m+ \frac{1}{2}) (n+ \frac{1}{2}) \pi^2}, \quad m,n = 0,1,2,\cdots.
\]
Hence, we get the expected value of the synchronized beating interval 
\[
\mathbf{E}(\mathbb{t}^{(1)}) = g(0,0) = \sum_{m,n=0}^\infty a_{mn} g_{mn}(0,0) =  \sum_{m,n=0}^\infty \frac{4(-1)^{m+n}}{\lambda_{mn} (m+ \frac{1}{2}) (n+ \frac{1}{2}) \pi^2} < \infty . 
\]
In above, we have derive $\mathbf{E}(\mathbb{t}^{(1)})$ for the case with zero intrinsic frequencies $\{\mu_i\}$, 
zero reaction coefficients $\{A_{i,j}\}$, and zero refractory thresholds $\{B_i\}$. 
However, for the general case, the closed-form of $g$ and $p$ are difficult to obtain,    
where one can compute the numerical solutions using the finite difference/element method (see Figure~\ref{fig:g-p-2} for a numerical example of $g$ and $p$).

\begin{figure}
 \begin{center}
  \subfigure{%
     \begin{overpic}[width=0.4\linewidth]{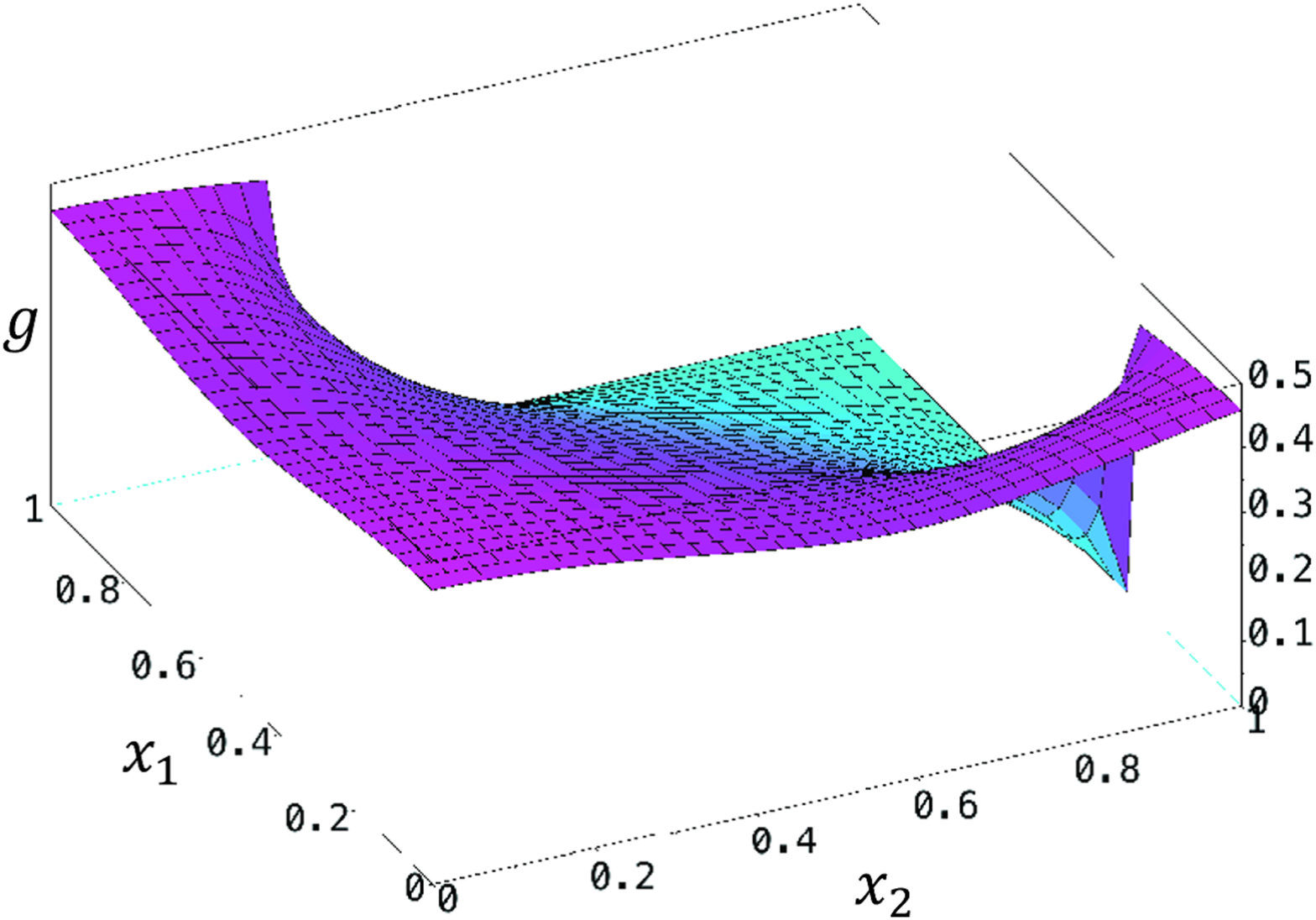}
     \put(0,70){\small \bf (a)}
     \end{overpic}}%
   \hspace*{0.4cm}
  \subfigure{%
     \begin{overpic}[width=0.4\linewidth]{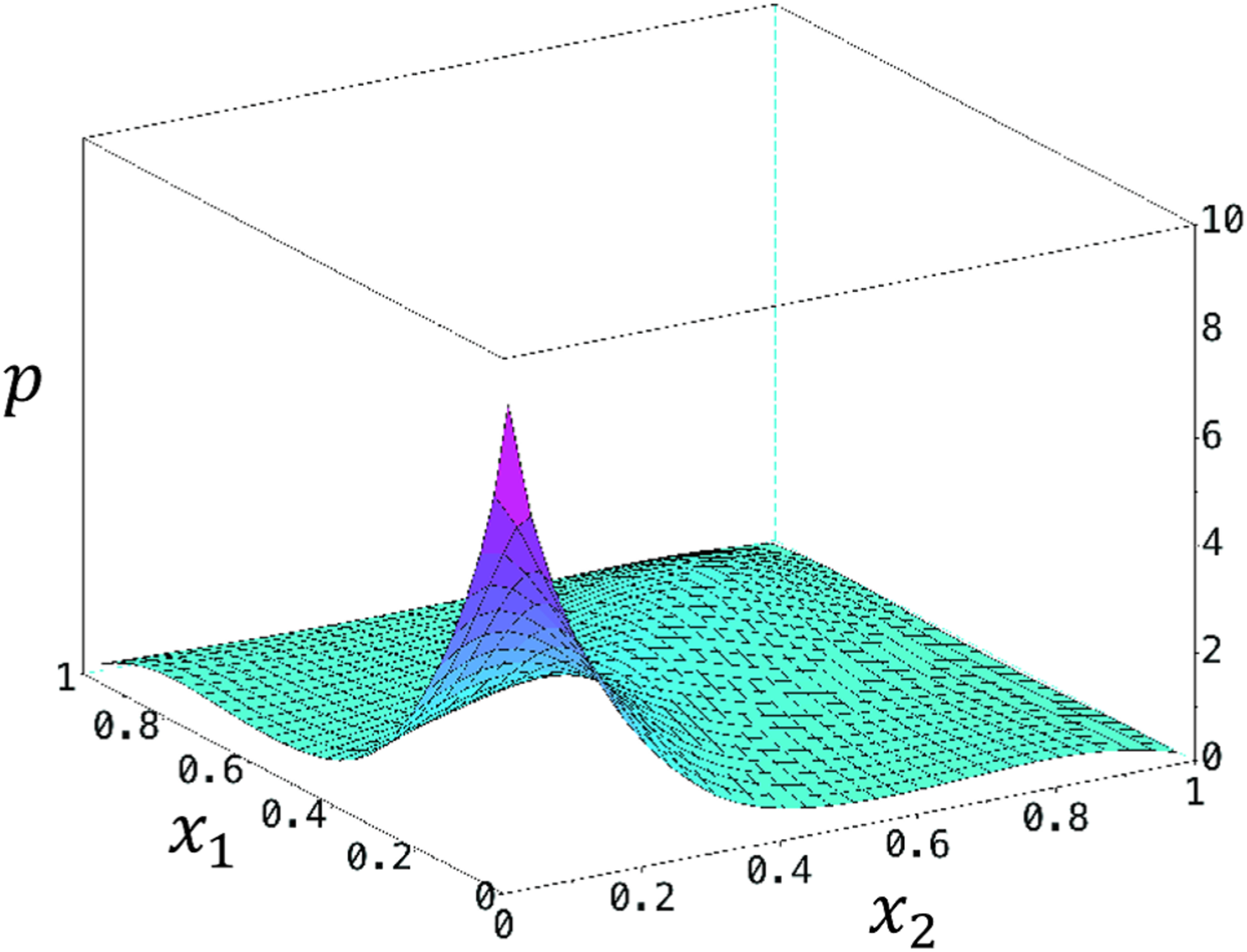}
     \put(0,68){\small \bf (b)}
     \end{overpic}}%
 \end{center}
\caption{(a) The profile of $g(x_1,x_2)$ with $\mathbf{E}(\mathbb{t}^{(1)}) = g_2(0,0) \approx 0.43$. (b) The profile of $p(x_1,x_2)$. Here, we set $(\mu_1,\mu_2) = (1,2)$, $\sigma_1=\sigma_2 = 1$, $A_{1,2} = A_{2,1}=2$, $B_1=B_2 = 0.3$. }
\label{fig:g-p-2}
\end{figure}
\subsection{The synchronized beating of the conventional model}\label{sec:2-c-2}
In view of Figure~\ref{fig:2-c} (b)(d), 
when the noise strength is sufficiently small and the reaction coefficients are large enough, 
the role of the reflective boundary and induced beating is ignorable,  
such that there is no much difference between the proposed model \eqref{eq:2-c} (L1)(L2)(IND)(REF) and the conventional model \eqref{eq:2-c-trad}. 

In this section, we shall pay attention to the conventional model \eqref{eq:2-c-trad}. 
Since the probability for the ``exact synxhronization'' $\bar{\phi}_1(t)=\bar{\phi}_2(t)=1$ is zero, 
we can only consider the ``approximated synchronization'', i.e., $\bar{\phi}_1(t_1)=\bar{\phi}_2(t_2)=1$ with $t_1 \approx t_2$. 
Let the $k$-th beating time for oscillator $i$ be the $k$-th passage time that $\bar{\phi}_i = 1$: 
\[
\mathbf{t}_i^{(k)} := \inf\{t > \mathbf{t}_i^{(k-1)} \ : \ \bar{\phi}_i(t) = 1 \} \quad \quad (\mathbf{t}_i^{(0)}=0).
\]
%
\begin{remark}\label{rk:no-jump}
In view of $f(\bar{\phi}_j- \bar{\phi}_i)=\sin(2\pi (\bar{\phi}_j- \bar{\phi}_i))$, 
it is equivalent that we remove the setting that $\bar{\phi}_i$ jumps to $0$ when reaching $1$, 
and define the $k$-th ``synchronized'' beating time $\mathbf{t}_i^{(k)}$ as the first passage time that $\bar{\phi}_i=k$. 
For the convenience of the discussion, 
we temporarily remove the enforcement that $\bar{\phi}_i(t)=0$ if $\bar{\phi}_i(t-)=1$ in the following argument of this section. 
Hence, the $k$-th beating time of oscillator $i$ is redefined by: 
\[
\mathbf{t}_i^{(k)} := \inf\{t > \mathbf{t}_i^{(k-1)} \ : \ \bar{\phi}_i(t) = k \}.
\]
In addition, we assume the ``approximated synchronization'' occurs, saying $\mathbf{t}_1^{(k)} \approx \mathbf{t}_2^{(k)}$. 
\end{remark}
\cmag{To ensure the ``approximated synchronization'', 
we assume that $|\bar{\phi}_1 - \bar{\phi}_2 | \ll 1$. 
In fact, we show that for sufficiently large reaction coefficients $\{A_{i,j} \}$ and small enough noise strength $\{\sigma_i\}$,  
one can guarantee that $|\mathbf{E}(\bar{\phi}_1 - \bar{\phi}_2)| \le \epsilon_1 \ll1$ and $\mathbf{Var}(\bar{\phi}_1 - \bar{\phi}_2) \le \epsilon_2 \ll1$. }

\cmag{Subtracting the following two equations with each other 
\[
\begin{aligned}
d \bar{\phi}_1 = \mu_1 dt + A_{1,2} \sin(2\pi (\bar{\phi}_2 - \bar{\phi}_1))dt + \sigma_1 dW_1(t), \\
d \bar{\phi}_2 = \mu_2 dt + A_{2,1} \sin(2\pi (\bar{\phi}_1 - \bar{\phi}_2))dt + \sigma_2 dW_2(t),
\end{aligned}
\]
we get 
\[
d (\bar{\phi}_1 - \bar{\phi}_2) = (\mu_1 - \mu_2)dt + (A_{1,2}+A_{2,1}) \sin(2\pi (\bar{\phi}_2 - \bar{\phi}_1)) dt + \sigma_1 dW_1(t) - \sigma_2 dW_2(t).
\]
For $|\bar{\phi}_1 - \bar{\phi}_2 | \ll 1$, we adopt the approximation $\sin(2\pi (\bar{\phi}_2(t) - \bar{\phi}_1(t))) \approx 2\pi (\bar{\phi}_2(t) - \bar{\phi}_1(t))$.  
Then the above equation becomes
\[
d [(\bar{\phi}_1 - \bar{\phi}_2(t))] = (\mu_1 - \mu_2)dt + 2\pi (A_{1,2}+A_{2,1})  (\bar{\phi}_2 - \bar{\phi}_1) dt + \sigma_1 dW_1(t) - \sigma_2 dW_2(t),
\]
which is equivalent to
\[
d [e^{2\pi (A_{1,2}+A_{2,1})t}(\bar{\phi}_1 - \bar{\phi}_2) = e^{2\pi(A_{1,2}+A_{2,1})t} [(\mu_1 - \mu_2)dt + \sigma_1 dW_1(t) - \sigma_2 dW_2(t)].
\]
With the initial value $\bar{\phi}_1(0) - \bar{\phi}_2(0)=0$, 
we find that 
\begin{equation}\label{eq:phi-1-2-diff}
\bar{\phi}_1(t) - \bar{\phi}_2(t) = (\mu_1 - \mu_2)\frac{1 - e^{ -2\pi (A_{1,2}+A_{2,1})t} }{2\pi (A_{1,2}+A_{2,1})} + \int_0^t e^{2\pi (A_{1,2}+A_{2,1})(s-t)} [\sigma_1 dW_1(s) - \sigma_2 dW_2(s)] .
\end{equation}
Taking the expectation of \eqref{eq:phi-1-2-diff} yields 
\begin{equation}\label{eq:phi-1-2-diff-E}
\mathbf{E}[\bar{\phi}_1(t) - \bar{\phi}_2(t) ] =  (\mu_1 - \mu_2)\frac{1 - e^{ -2\pi (A_{1,2}+A_{2,1})t} }{2\pi (A_{1,2}+A_{2,1})} + 0. 
\end{equation}
Thus, for sufficiently large $\{A_{i,j}\}$ such that  $2\pi (A_{1,2}+A_{2,1}) \ge \frac{|\mu_1 - \mu_2|}{\ep_1}$ $(0<\ep_1 \ll1)$, 
$|\mathbf{E}[\bar{\phi}_1 - \bar{\phi}_2] |\le \ep_1$ is guaranteed. 

To derive the sufficient condition for $\mathbf{Var}(\bar{\phi}_1 - \bar{\phi}_2) \le \epsilon_2 \ll1$, 
from \eqref{eq:phi-1-2-diff}, \eqref{eq:phi-1-2-diff-E}, 
we calculate as 
\[
\begin{aligned}
(\bar{\phi}_1(t) - \bar{\phi}_2(t))^2 = & (\mathbf{E}[\bar{\phi}_1(t) - \bar{\phi}_2(t) ])^2 + \left( \int_0^t e^{2\pi (A_{1,2}+A_{2,1})(s-t)} [\sigma_1 dW_1(s) - \sigma_2 dW_2(s)] \right)^2 \\
& + 2 \mathbf{E}[\bar{\phi}_1(t) - \bar{\phi}_2(t) ] \int_0^t e^{2\pi (A_{1,2}+A_{2,1})(s-t)} [\sigma_1 dW_1(s) - \sigma_2 dW_2(s)],
\end{aligned}
\]
which implies 
\[
\begin{aligned}
\mathbf{E}[(\bar{\phi}_1 - \bar{\phi}_2)^2] & = (\mathbf{E}[\bar{\phi}_1 - \bar{\phi}_2 ])^2 + \mathbf{E}\left[ \int_0^t e^{2\pi (A_{1,2}+A_{2,1})(s-t)}\sigma_1 dW_1(s) \right]^2 \\
& + \mathbf{E}\left[ \int_0^t e^{2\pi (A_{1,2}+A_{2,1})(s-t)}\sigma_2 dW_2(s) \right]^2 \\
& + 2 \underbrace{\mathbf{E}\left[ \int_0^t e^{2\pi (A_{1,2}+A_{2,1})(s-t)}\sigma_1 dW_1(s) \right]}_{=0} \underbrace{\mathbf{E}\left[ \int_0^t e^{2\pi (A_{1,2}+A_{2,1})(s-t)}\sigma_2 dW_2(s) \right] }_{=0} \\
& + 2 \mathbf{E}[\bar{\phi}_1(t) - \bar{\phi}_2(t) ]  \underbrace{\mathbf{E} \left[ \int_0^t e^{2\pi (A_{1,2}+A_{2,1})(s-t)} [\sigma_1 dW_1(s) - \sigma_2 dW_2(s)] \right]}_{=0} 
\end{aligned}
\]
By Ito's isometry, we have 
\[
\begin{aligned}
& \mathbf{Var}[\bar{\phi}_1 - \bar{\phi}_2]  =  \mathbf{E}[(\bar{\phi}_1- \bar{\phi}_2)^2]  - (\mathbf{E}[\bar{\phi}_1 - \bar{\phi}_2 ])^2 \\
= &  \int_0^t e^{4\pi (A_{1,2}+A_{2,1})(s-t)} (\sigma_1^2 + \sigma_2^2) ds = (1-e^{-4\pi(A_{1,2}+A_{2,1}) t}) (\sigma_1^2 + \sigma_2^2).
\end{aligned}
\]
Therefore, for sufficiently small noise strength such that $(\sigma_1^2 + \sigma_2^2) \le \ep_2$, 
we have $\mathbf{Var}[\bar{\phi}_1(t) - \bar{\phi}_2(t)] \le \epsilon_2$. }

From now on, we tacitly assume that $\{A_{i,j}\}$ are sufficiently large and $\{\sigma_i\}$ are small enough such that the ``approximated'' synchronization ($\mathbf{t}_1^{(k)} \approx \mathbf{t}_j^{(k)}$) occurs. 
And we turn to investigate the CV of the beating intervals $\{\mathbf{t}_i^{(k)}\}$,  
where we employ the approximation approach proposed by \cite[(5)--(18)]{Kori}. 

Let us briefly introduce the idea of \cite{Kori}. 
For a very large time scale, one can approximate the stable synchronization oscillation system by the linear system:  
\begin{equation}\label{eq:syn-phi-i-pre}
\phi_i^{\text{syn}} (t) = \mu^{\text{syn}} t + \psi_i^{\text{syn}}, \quad i=1,2,
\end{equation}
where the phase functions $\{\phi_i^{\text{syn}}\}_{i=1,2}$ are called the synchronized solutions, 
with the intrinsic synchronized frequency $\mu^{\text{syn}}$ and initial state $\psi_i^{\text{syn}}$ satisfying 
\begin{equation}\label{eq:psi-i-j}
A_{i,j}\sin(2\pi(\psi_j^{\text{syn}} - \psi_i^{\text{syn}} )) = \mu_i - \mu^{\text{syn}}, \quad i,j=1,2, \quad i \neq j.
\end{equation}
For \eqref{eq:syn-phi-i-pre}, we have the synchronized beating interval $\tau = 1/\mu^{\text{syn}}$, 
and 
\begin{equation}\label{eq:phi-i-t+tau}
1 = \phi_i^{\text{syn}} (t+ \tau) - \phi_i^{\text{syn}} (t).  
\end{equation}
Here, we take $\tau$ as the mean value of the beating intervals $\{\mathbf{t}_i^{(k)}\}_{k=1,2,\cdots}$. 
Since one oscillation cycle of $\bar{\phi}_i$ corresponds to the increasement of $\bar{\phi}_i$ by $1$, 
according to the discussion of \cite{Kori}, 
the variance of beating intervals $\{ \mathbf{t}_i^{(k)} \}_{k=1,2,\cdots}$ is proportional to the variance of $\bar{\phi}_i(t+\tau) - \bar{\phi}_i(t) -1$ as $t \rightarrow \infty$.  
Therefore, the CV of $\{\mathbf{t}_i^{(k)}\}$ can be approximated by 
\begin{equation}\label{eq:CV-i-app-0}
\mathbf{CV}_i := \sqrt{\lim_{t \rightarrow \infty} \mathbf{E}[(\bar{\phi}_i(t+\tau) - \bar{\phi}_i(t) -1)^2]}. 
\end{equation}
Setting the notation $\xi_i(t) := \bar{\phi}_i(t) - \phi_i^{\text{syn}} (t)$, 
from \eqref{eq:phi-i-t+tau} and \eqref{eq:CV-i-app-0}, 
we see that  
\begin{equation}\label{eq:CV-i-app}
\mathbf{CV}_i^2 =  \lim_{t \rightarrow \infty} \mathbf{E}[(\xi_i(t+\tau) - \xi_i(t))^2] 
\end{equation}

Now the problem reduces to calculate $\mathbf{E}[(\xi_i(t+\tau) - \xi_i(t))^2]$.  
To this end, we first derive the equations for $\{\xi_i\}_{i=1,2}$: $i,j=1,2$, $i \neq j$, 
\[
\begin{aligned}
& d \xi_i(t) = (\mu_i - \mu^{\text{syn}})dt + A_{i,j} \sin(2\pi(\bar{\phi}_j(t) - \bar{\phi}_i(t)))dt + \sigma_idW_i(t), \\
& \xi_i(0) = \xi_i^0 :=\bar{\phi}_i(0) - \phi^{\text{syn}}_i(0) = -\psi_i^{\text{syn}}, 
\end{aligned}
\]
We assume that the difference between the synchronized solution $\phi_i^{\text{syn}}$ and phase $\bar{\phi}_i$ is small, i.e., 
\[
|\xi_i|=|\phi_i^{\text{syn}}-\bar{\phi}_i | \ll 1. 
\]
In view of 
\[
\begin{aligned}
\sin(2\pi(\bar{\phi}_j(t) - \bar{\phi}_i(t))) &  = \sin(2\pi(\xi_j(t) - \xi_i(t) + \phi_j^{\text{syn}}(t) - \phi_i^{\text{syn}}(t) )) = \sin(2\pi(\xi_j(t) - \xi_i(t) + \psi_j^{\text{syn}} - \psi_i^{\text{syn}} )) \\
& =  \sin(2\pi(\psi_j^{\text{syn}} - \psi_i^{\text{syn}} )) +\cos (2\pi(\psi_j^{\text{syn}} - \psi_i^{\text{syn}} )) ( \xi_j(t) - \xi_i(t) ) + O((\xi_j(t) - \xi_i(t) )^2),
\end{aligned}
\]
and neglecting the smaller quadratic term $O(|\xi_j(t)|^2)$ and $O( |\xi_i(t)|^2 )$, 
together with \eqref{eq:psi-i-j}, 
we obtain: 
\begin{subequations}\label{eq:xi}
\begin{align}
& d \xi_i(t) = b_{ij}(\xi_j(t) - \xi_i(t)) dt + \sigma_idW_i(t), \label{eq:xi-a} \\
& \xi_i(0) = \xi_i^0, \label{eq:xi-b}
\end{align}
\end{subequations}
where $b_{ij} := A_{i,j} \cos(2\pi(\psi_j^{\text{syn}} - \psi_i^{\text{syn}} )) $. 
\cblue{In the following, we assume $A_{i,j} > 0$ and  $|\psi_j^{\text{syn}} - \psi_i^{\text{syn}} | \ll 1$ such that $\cos(2\pi(\psi_j^{\text{syn}} - \psi_i^{\text{syn}} )) \approx 1$ and $b_{ij} \approx A_{i,j} >0$.} 

\cblue{From now on, we establish a new analysis utilizing the stochastic calculus, which is different to \cite{Kori}.  
Comparing with \cite{Kori}, we makes the improvement in two aspects: 
First, we present a rigorous mathematical calculation of $\lim_{t\rightarrow \infty}\mathbf{E}[(\xi_i(t+\tau) - \xi_i(t))^2]$.   
Second, our result shows a explicit relationship between the parameters $b_{ij}$ and the CV, 
which is of practical use to determine the suitable parameters $\{A_{i,j}\}$ (see Remark~\ref{rk:A-ij}). }
\cblue{\begin{proposition}\label{prop:CV-i}
We approximate the CV of the synchronized beating intervals $\{\mathbf{t}_i^{(k)} \}_{k=1,2,\cdots}$ by $\mathbf{CV}_i = \lim_{t \rightarrow \infty} \mathbf{E}[(\xi_i(t+\tau)- \xi_i(t))^2] $, 
where $\{ \xi_i\}_{i=1,2}$ is the solution of \eqref{eq:xi}.
For $A_{1,2}, A_{2,1} > 0$, and $\cos(2\pi (\psi_j^{\text{syn}} - \psi_i^{\text{syn}})) > 0$, 
we have: $i,j = 1,2$, $i \neq j$, 
\begin{equation}\label{eq:CV-i}
\begin{aligned}
\mathbf{CV}_i^2 := & \lim_{t \rightarrow \infty} \mathbf{E}[(\xi_i(t+\tau)- \xi_i(t))^2]  \\
= & b^{-2} \left\{ (b_{ij}\sigma_j^2 + b_{ji}\sigma_i^2)\tau + b^{-1} \left[ 2b_{ij} (b_{ji} \sigma_i^2 - b_{ij} \sigma_j^2) + b_{ij}^2 (\sigma_j^2 + \sigma_i^2) \right] (1-e^{-\tau b}) \right\}.
\end{aligned}
\end{equation}
\end{proposition}
}
\cblue{\begin{remark}\label{rk:A-ij}
In Section~3,  
we determine the intrinsic frequency and noise strength $(\mu_i,\sigma_i)$ for single-isolated cell $i$ $(i=1,2)$ by formulas \eqref{eq:1-c-t^1-E}--\eqref{eq:1-c-t^1-CV}, 
together with the mean value and variance/CV of the beating intervals obtained from the bio-experiments \cite{KojimaK06}.  
Coupling two cells (cell $1$ and cell $2$), we intend to find suitable coefficients $A_{1,2}$ and $A_{2,1}$ for the reaction terms. 
Assuming that the difference between the synchronized solution $\{\phi_i^{\text{syn}}\}_{i=1,2}$ is tiny ($|\phi_1^{\text{syn}} - \phi_2^{\text{syn}}| = |\psi_1^{\text{syn}} - \psi_2^{\text{syn}} | \approx 0$), 
and taking the approximation 
\[
b_{ij} = A_{i,j} \cos(2\pi (\psi_i^{\text{syn}} - \psi_j^{\text{syn}} )) \approx A_{i,j},  
\]
we see that  
\begin{equation}\label{eq:CV-i-A}
\begin{aligned}
\mathbf{CV}_i^2 = & (A_{i,j} + A_{j,i})^{-2} (A_{i,j}\sigma_j^2 + A_{j,i}\sigma_i^2)\tau \\
& + (A_{i,j} + A_{j,i})^{-3} \left[ 2A_{i,j}(A_{j,i}\sigma_i^2 - A_{i,j} \sigma_j^2) + A_{i,j}^2 (\sigma_j^2 + \sigma_i^2) \right] (1-e^{-\tau (A_{i,j} + A_{j,i})}).
\end{aligned}
\end{equation}
Meanwhile, the expetation and CV of the synchronized beating intervals, denoted by $\mathbf{T}$ and $\mathbf{CV}$, can be obtained from the bio-experiments \cite{KojimaK06}.  
Substituting $\tau=\mathbf{T}$ and $\mathbf{CV}_1 = \mathbf{CV}_2 = \mathbf{CV}$ into \eqref{eq:CV-i-A}, 
one can solve \eqref{eq:CV-i-A} ($i,j=1,2$ $i \neq j$) numerically to get the coefficients $A_{1,2}$ and $A_{2,1}$. 
\end{remark}
}
\cblue{\begin{proof}[Proof of Proposition~\ref{prop:CV-i}]
Setting the notations
\[
\bm{\xi} = \left[ \begin{array}{c}
\xi_1 \\
\xi_2 
\end{array}\right], \quad 
\bm{B} = \left[ \begin{array}{cc} 
b_{12} & -b_{12} \\ 
-b_{21} & b_{21} 
\end{array} \right], 
\]
\[
\bm{\xi}^0 = \left[ \begin{array}{c}
\xi_1^0 \\
\xi_2^0 
\end{array}\right], \quad 
\bm{W} = \left[ \begin{array}{c}
W_1 \\
W_2 
\end{array}\right], \quad 
\bm{\sigma} = \left[ \begin{array}{cc} 
\sigma_1 & 0 \\ 
0 & \sigma_2 
\end{array} \right], 
\]
we write \eqref{eq:xi} as follows: 
\[
d \bm{\xi} =  - \bm{B} \bm{\xi} dt + \bm{\sigma} d \bm{W}(t), \quad \bm{\xi}(0) = \bm{\xi}^0. 
\]
Multiplying the above equation with $e^{-\bm{B}t}$, we have 
\[
d(e^{\bm{B}t} \bm{\xi}) = e^{\bm{B}t} [\bm{\nu} dt  + \bm{\sigma} d \bm{W}],
\]
which implies 
\begin{equation}\label{eq:xi-sol-2-c}
\bm{\xi}(t) = e^{-\bm{B} t} \bm{\xi}^0 +  \int_0^t e^{-\bm{B}(t-s)} \bm{\sigma}~d\bm{W}(s). 
\end{equation}
Since the expectation of It\^{o}'s integral is zero, 
\[
\mathbf{E}[\bm{\xi}(\tau)] = e^{-\bm{B} t} \bm{\xi}^0. 
\]
One can validate that $\bm{B}$ has two sets of eigenvalue and eigenvector:
\[
\lambda_1=0 \quad \bm{u}_1=[1,1]^\top, \quad \lambda_2 = b := b_{12}+b_{21} \quad \bm{u}_2=[b_{12},-b_{21}]^\top.
\]
And we have 
\[
e^{-t\bm{B}} \bm{u}_1 = \bm{u}_1, \quad e^{-t\bm{B}} \bm{u}_2 = e^{-tb} \bm{u}_2
\]
With the help of $(\lambda_i, \bm{u}_i)_{i=1,2}$, we make the decompositions
\[
\begin{aligned}
& \bm{\xi}^0 = b^{-1} (b_{21}\xi_1^0 +  b_{12} \xi_2^0) \bm{u}_1 + b^{-1} (\xi_1^0 - \xi^0_2) \bm{u}_2, \\
& \bm{\sigma} d \bm{W} = b^{-1} (b_{21}\sigma_1 dW_1 + b_{12}\sigma_2 dW_2) \bm{u}_1 + b^{-1} (\sigma_1 dW_1 - \sigma_2 dW_2) \bm{u}_2,  
\end{aligned}
\]
substituting which into \eqref{eq:xi-sol-2-c}, we observe that 
\begin{equation}\label{eq:xi-sol-2-c-1}
\begin{aligned}
\xi_1(t) = & b^{-1} (b_{21} \xi_1^0 + b_{12} \xi_2^0) + b^{-1} (\xi^0_1 - \xi_2^0) e^{-tb} b_{12} \\
& + \int_0^t b^{-1} (b_{21} \sigma_1 dW_1(s) + b_{12} \sigma_2 dW_2(s)) + \int_0^t b^{-1} e^{-(t-s)b} b_{12} (\sigma_1 dW_1(s) - \sigma_2 dW_2(s)) \\
& = b^{-1} \mathcal{C}_1(t) +  b^{-1} \underbrace{ \sigma_1 \int_0^t (b_{21} + e^{-(t-s)b} b_{12}) dW_1(s)}_{=: \mathcal{W}_1(t)} + b^{-1}  \underbrace{ \sigma_2 \int_0^t  
( b_{12} - e^{-(t-s)b} b_{12} ) dW_2(s) }_{=: \mathcal{W}_2(t)} , 
\end{aligned}
\end{equation}
where $\mathcal{C}_1(t) := (b_{21} \xi_1^0 + b_{12} \xi_2^0) + b^{-1} (\xi^0_1 - \xi_2^0) e^{-tb} b_{12}$. 
Then, we see that 
\[
\begin{aligned}
\xi_1^2(t) =  b^{-2} [ \mathcal{C}_1^2(t) + \mathcal{W}_1^2(t) + \mathcal{W}_2^2(t)] + 2 b^{-2} [ \mathcal{C}_1(t) \mathcal{W}_1(t) +  \mathcal{C}_1(t) \mathcal{W}_2(t) + \mathcal{W}_1(t) \mathcal{W}_2^2(t)] 
\end{aligned}
\]
The expectation of It\^{o}'s integral $\mathcal{W}_i$ is zero, that is, 
\[
\mathbf{E}[\mathcal{W}_1(t)] = 0, \quad \mathbf{E}[\mathcal{W}_2(t)] = 0,
\]
together with the independency between $\mathcal{W}_1(t)$ and $\mathcal{W}_2(t)$ (because $W_1$ and $W_2$ are independent), 
which gives 
\[
\mathbf{E}[\mathcal{W}_1(t) \mathcal{W}_2(t)] = \mathbf{E}[\mathcal{W}_1(t)] \mathbf{E}[\mathcal{W}_2(t)]  = 0. 
\]
Moreover, by Ito's isometry, 
\[
\mathbf{E}[\mathcal{W}_1^2(t)] =  \int_0^t (b_{21} + e^{-(t-s)b} b_{12})^2 \sigma_1^2 ds, \quad \mathbf{E}[\mathcal{W}_2^2(t)] = \int_0^t  ( b_{12} - e^{-(t-s)b} b_{12} )^2 \sigma_2^2 ds. 
\]
Hence, we conclude 
\begin{equation}\label{eq:xi-sol-2-c-1^2}
\begin{aligned}
\mathbf{E}[\xi_1^2(t)] = & b^{-2}  \mathcal{C}_1^2(t) + b^{-2} \int_0^t (b_{21} + e^{-(t-s)b} b_{12})^2 \sigma_1^2 ds + b^{-2} \int_0^t  ( b_{12} - e^{-(t-s)b} b_{12} )^2 \sigma_2^2 ds \\
= & b^{-2} (b_{21} \xi_1^0 + b_{12} \xi_2^0 + (\xi^0_1 - \xi_2^0) e^{-tb} b_{12})^2 \\
& +  b^{-2} \left[  (b_{21}^2 \sigma_1^2 + b_{12}^2 \sigma_2^2)t + 2b_{12}(b_{21}\sigma_1^2 - b_{12}\sigma_2^2) \frac{1-e^{-tb}}{b}  + b_{12}^2(\sigma_1^2 + \sigma_2^2) \frac{1-e^{-2tb}}{2b} \right]. 
\end{aligned}
\end{equation}
In view of $\mathbf{E}[(\xi_1(t+\tau)- \xi_1(t))^2] = \mathbf{E}[\xi_1^2(t+\tau) + \xi_1^2(t) - 2\xi_1(t+\tau) \xi_1(t) ]$, 
it remains to calculate $\mathbf{E}[\xi_1(t+\tau) \xi_1(t) ]$. 
\[
\begin{aligned}
\mathbf{E} [\xi_1(t+\tau) \xi_1(t) ] = &  b^{-2}\mathbf{E} [ \mathcal{C}_1(t) \mathcal{C}_1(t+ \tau) + \mathcal{W}_1(t) \mathcal{W}_1(t+ \tau) + \mathcal{W}_2(t) \mathcal{W}_2(t+ \tau)] \\
& + b^{-2} \underbrace{\mathbf{E}[ \mathcal{C}_1(t) \mathcal{W}_1(t+ \tau) + \mathcal{C}_1(t) \mathcal{W}_2(t+ \tau) + \mathcal{W}_2(t) \mathcal{W}_1(t+ \tau) + \mathcal{W}_1(t) \mathcal{W}_2(t+ \tau)] }_{=0}.
\end{aligned}
\]
Let us pay attention to $\mathbf{E}[ \mathcal{W}_1(t) \mathcal{W}_1(t+ \tau) ]$. 
We divide $\mathcal{W}_1(t+ \tau)$ into 
\[
\mathcal{W}_1(t+ \tau) = \sigma_1 \int_0^t (b_{21} + e^{-(t+\tau-s)b} b_{12}) dW_1(s) +  \sigma_1 \int_t^{t+\tau} (b_{21} + e^{-(t+\tau-s)b} b_{12}) dW_1(s).
\]
The independency between $\int_t^{t+\tau} (b_{21} + e^{-(t+\tau-s)b} b_{12}) dW_1(s)$ and $\int_0^t (b_{21} + e^{-(t-s)b} b_{12}) dW_1(s)$ yields 
\[
\begin{aligned}
\mathbf{E}[ \mathcal{W}_1(t) \mathcal{W}_1(t+ \tau) ] & = \sigma_1^2 \left[ \left( \int_0^t (b_{21} + e^{-(t-s)b} b_{12}) dW_1(s) \right) \left( \int_0^t (b_{21} + e^{-(t+\tau-s)b} b_{12}) dW_1(s) \right) \right] \\
& = \sigma_1^2 \int_0^t (b_{21} + e^{-(t-s)b} b_{12}) (b_{21} + e^{-(t+\tau-s)b} b_{12})~ds \quad \quad (\text{by Ito's isometry}).
\end{aligned}
\]
Treating $\mathbf{E}[ \mathcal{W}_2(t) \mathcal{W}_2(t+ \tau) ]$ in a similar way,  
we get
\begin{equation}\label{eq:xi-sol-2-c-1^2-a}
\begin{aligned}
\mathbf{E} [\xi_1(t+\tau) \xi_1(t) ] = &  b^{-2} \mathcal{C}_1(t) \mathcal{C}_1(t+ \tau) \\
& + b^{-2} \sigma_1^2 \left[ b_{12}^2 t + b_{12}b_{21} (1+e^{-\tau b}) \frac{1-e^{-tb}}{b} + b_{12}^2 \frac{e^{-\tau b}(1-e^{-2tb})}{2b} \right] \\
& + b^{-2} \sigma_2^2 \left[ b_{12}^2 t - b_{12}^2 (1+e^{-\tau b}) \frac{1-e^{-tb}}{b} + b_{12}^2 \frac{e^{-\tau b}(1-e^{-2tb})}{2b} \right]. 
\end{aligned}
\end{equation}
Following from \eqref{eq:xi-sol-2-c-1^2}, \eqref{eq:xi-sol-2-c-1^2-a}, we find that  
\[
\begin{aligned}
\mathbf{E}[(\xi_1(t+\tau)- \xi_1(t))^2]  = & b^{-2} \left[ b_{12}^2(\xi_1^0 - \xi_2^0)^2 e^{-2tb} (1-e^{-\tau b}) + (b_{21}\sigma_1^2 + b_{12}\sigma_2^2) \tau \right] \\
& + b^{-2} \left[  \frac{2 b_{12} (b_{21} \sigma_1^2 - b_{12}\sigma_2^2)}{b}(1-e^{-\tau b}) \right] \\
& + b^{-2} \left[  \frac{b_{12}^2(\sigma_1^2 + \sigma_2^2)}{2b}(2 - 2e^{-\tau b} + e^{-2(t+\tau)b} - e^{-2tb})  \right]
\end{aligned}
\]
Passing to the limit $t \rightarrow \infty$ and in view of $b = b_{12} + b_{21} > 0$, we have 
\begin{equation}\label{eq:CV-1}
\begin{aligned}
(\mathbf{CV}_1)^2 = & \lim_{t \rightarrow \infty} \mathbf{E}[(\xi_1(t+\tau)- \xi_1(t))^2]  \\
= & b^{-2} \left\{ (b_{21}\sigma_1^2 + b_{12}\sigma_2^2)\tau + b^{-1} \left[ 2b_{12} (b_{21} \sigma_1^2 - b_{12} \sigma_2^2) + b_{12}^2 (\sigma_1^2 + \sigma_2^2) \right] (1-e^{-\tau b}) \right\}.
\end{aligned}
\end{equation}
Analogously to above argument, we can calculate $\mathbf{CV}_2$. 
\end{proof}
}
\section{The phase model for the $N$-cells network} \label{sec:N-c}
Let us extend the phase models of two-coupled cells to $N$-cells network. 
Figure~\ref{fig:cellnetwork} (a) shows two examples of cell-network constructed via the on-chip cellomics technology \cite{KanekoT07, KojimaK06}. 
Numbering the cells by $\{1,2,\cdots,N\}$, we denote by $\mathcal{N}_i$ the neighbors of cell $i$ (see Figure~\ref{fig:cellnetwork} (a)). 
In this section, we first introduce the phase model for $N$-cells network incorporating the irreversibility of beating (reflective boundary), induced beating and refractory. 

For the case with sufficiently large reaction coefficients and small enough noise strength,
the proposed model has similar behavior to the conventional model, and the synchronization is very stable,  
because the effects of reflective boundary, induced firing and refractory is ignorable (see Figure~\ref{fig:4-c} (b)(d)). 
Since the massive bio-experiments (cf. \cite{KanekoT07}) reveal that the CV of the synchronized beating intervals reduces as the network size increases (in other word, the synchronization is more stable if we add more cardiamyocytes to the network), 
we shall investigate the network-size-dependent CV of the synchronized beating intervals by the conventional model with a similar analysis to Section~4.2.

\begin{figure}[h]
 \begin{center}
 \subfigure{%
     \begin{overpic}[width=0.35\linewidth]{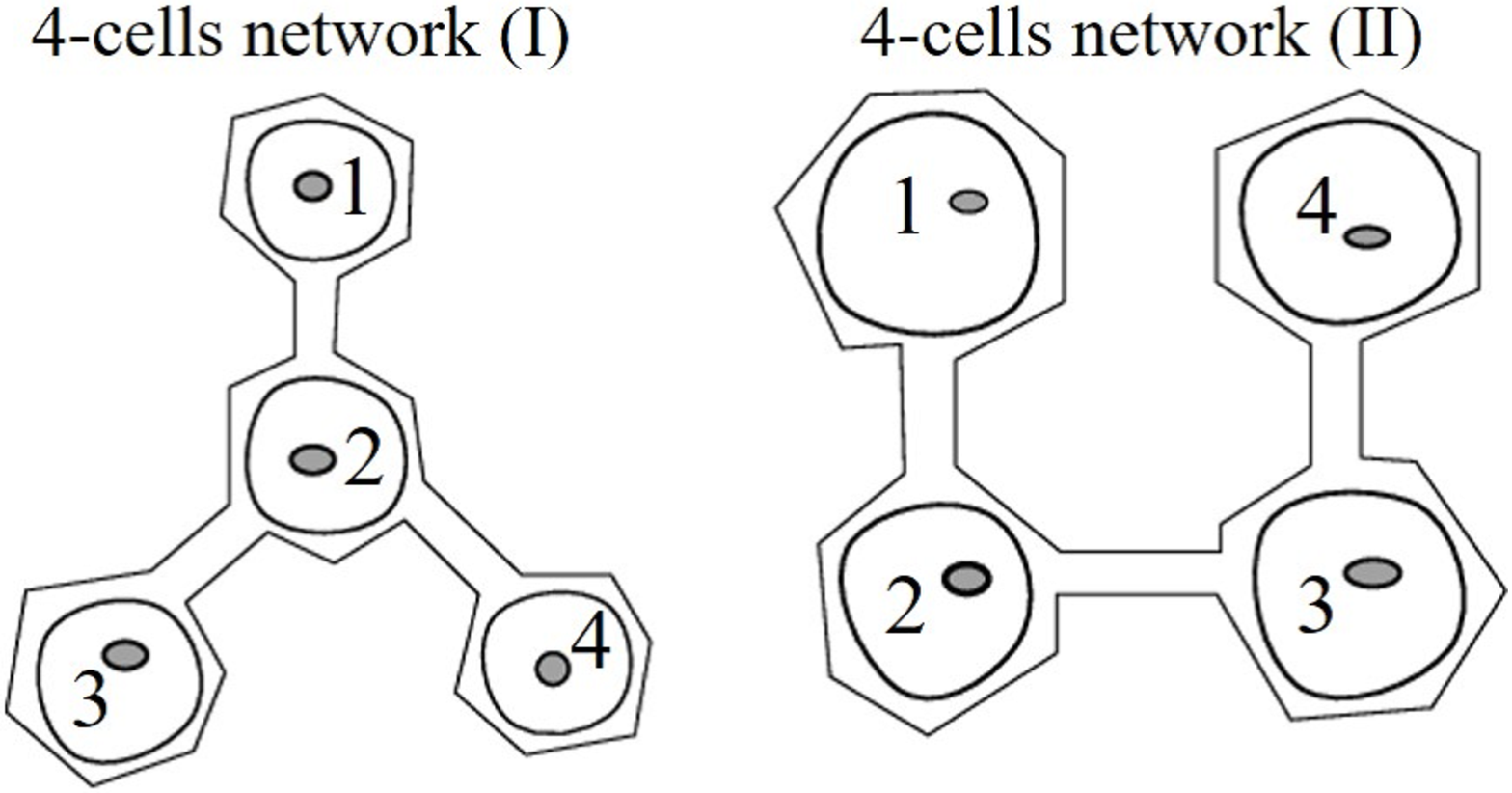}
     \put(-4,54){\bf \small{(a)}}
     %
     \end{overpic}
      }%
      \hspace*{0.4cm}
  \subfigure{%
     \begin{overpic}[width=0.22\linewidth]{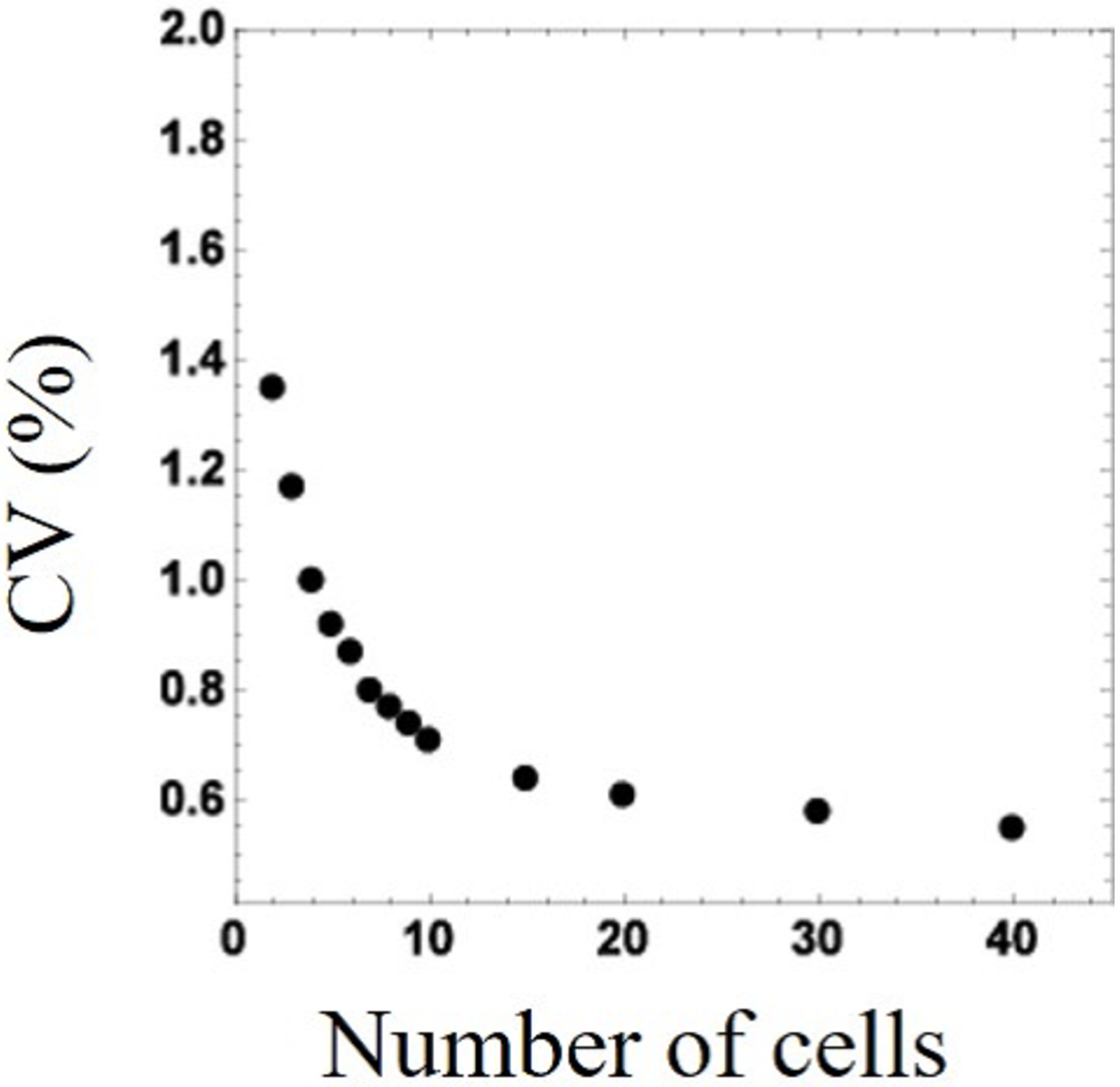}
     \put(-6,100){\small \bf (b)}
     \end{overpic}
      }%
   \hspace*{0.4cm}
  \subfigure{%
     \begin{overpic}[width=0.22\linewidth]{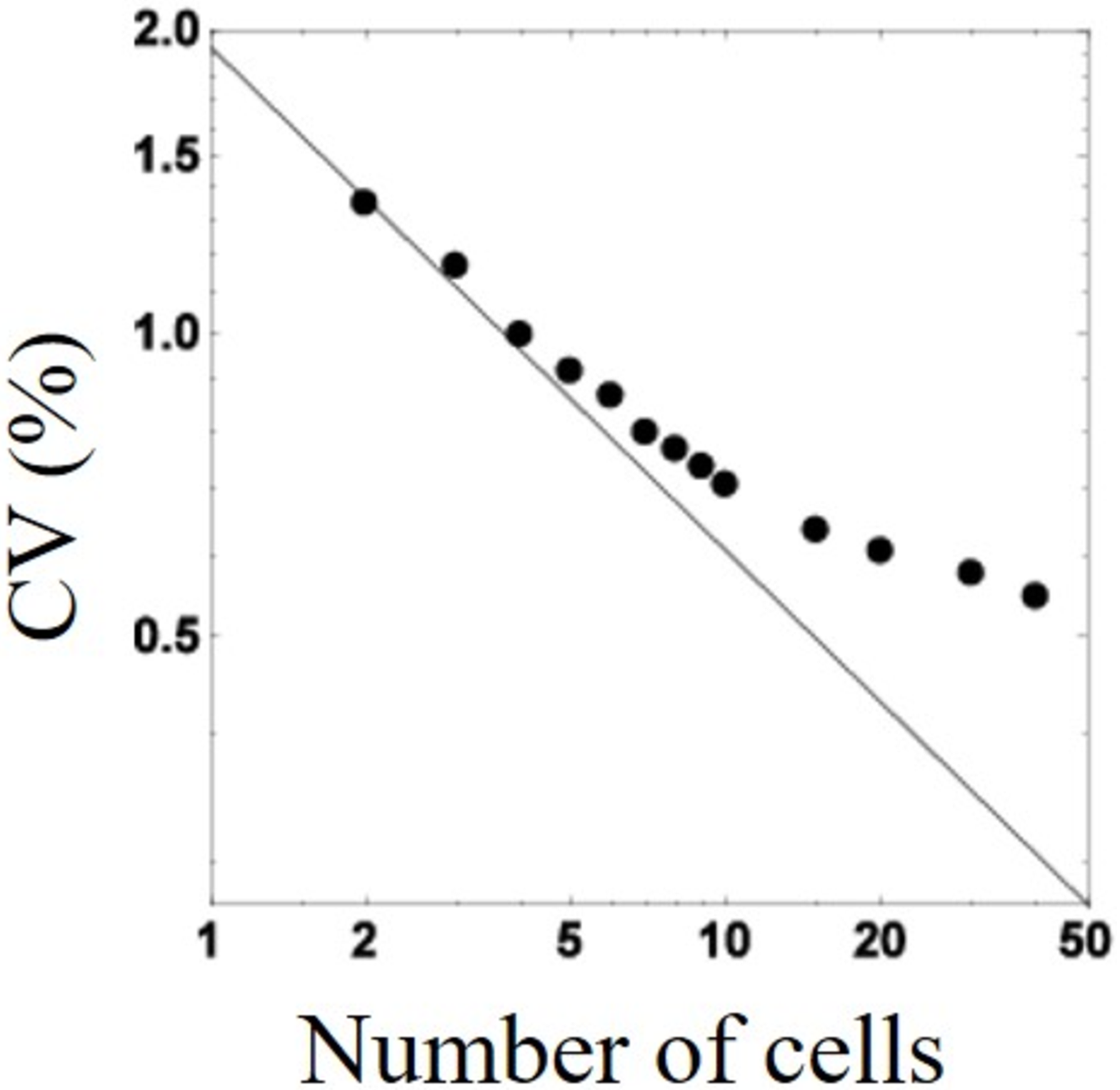}
     \put(-6,100){\small \bf (c)}
     \end{overpic}
      }%
 \end{center}
  \caption{(a) Tow examples of cell network. $4$-cells network (I): $\mathcal{N}_i = \{2\}$ ($i=1,3,4$), $\mathcal{N}_2 = \{1,3,4\}$. $4$-cells network (II):  $\mathcal{N}_1 = \{2\}$,  $\mathcal{N}_2 = \{1,3\}$,  $\mathcal{N}_3 = \{2,4\}$,  $\mathcal{N}_4 = \{3\}$.     
  (c) The size-dependent beating fluctuation. The CV of the synchronized beating intervals decreases as the cell number increases. 
  (d) The log-log scale of (c), where the black straight line represents $\propto N^{-1/2}$. 
  $N$ is the number of cells in the network.}
\label{fig:cellnetwork}
\end{figure}

\subsection{The phase model of $N$-cells network}
Let $(\phi_i, \mu_i,\sigma_i)$ denote the phase, intrinsic frequency and noise strength of cell $i$ ($i=1,\ldots,N$), and $A_{i,j}$ the coefficient of the reaction term between cell $i$ and $j$. 
For simplicity, we consider the network that all the cells are connected with each other, that is $\mathcal{N}_i = \{1,2,\cdots,N\}-\{i\}$. 
Then, the equations of $\{\phi_i\}_{i=1}^N$ are stated as follows: for $i = 1,2,\cdots, N$, 
\begin{subequations}\label{eq:N-c}
\begin{align}
& d \phi_i(t) = \mu_i dt + \sum_{j \in \mathcal{N}_i} A_{i,j} \sin(2\pi (\phi_j - \phi_i)) dt + \sigma_i d W_i(t) + dL_i(t), \\
& \phi_i(0) = 0,  
\end{align}
\end{subequations}
where $W_i(t)$ denotes the normal Brownian motion ($\{W_i\}_{i=1}^N$ are independent),  
and $L_i(t)$ the process implementing the reflective boundary (see (L1)(L2) of Section~4).  
When $\phi_i(t-) = 1$, we say cell $i$ beats spontaneously. 
At the same time, the neighboring cell $j$ ($j \in \mathcal{N}_i$) is induced to beat if $\phi_j(t-)>B_j$ (cell $j$ is out of refractory), 
where $B_j \in [0,1]$ denotes the refractory threshold of cell $j$. 
In this case, cell $i$ and cell $j$ have a synchronized beating. 
And after beating, both two phases jump to zero, that is, $\phi_i(t) = 0 $ and $\phi_j(t) = 0$.  
On the other hand, if $\phi_j(t-) \le B_j$, we say cell $j$ is in refractory and cannot be induced to beat, 
and we have $\phi_j(t) = \phi_j(t-)$.

\begin{figure}[h]
 \begin{center}
     \begin{overpic}[width=0.43\linewidth]{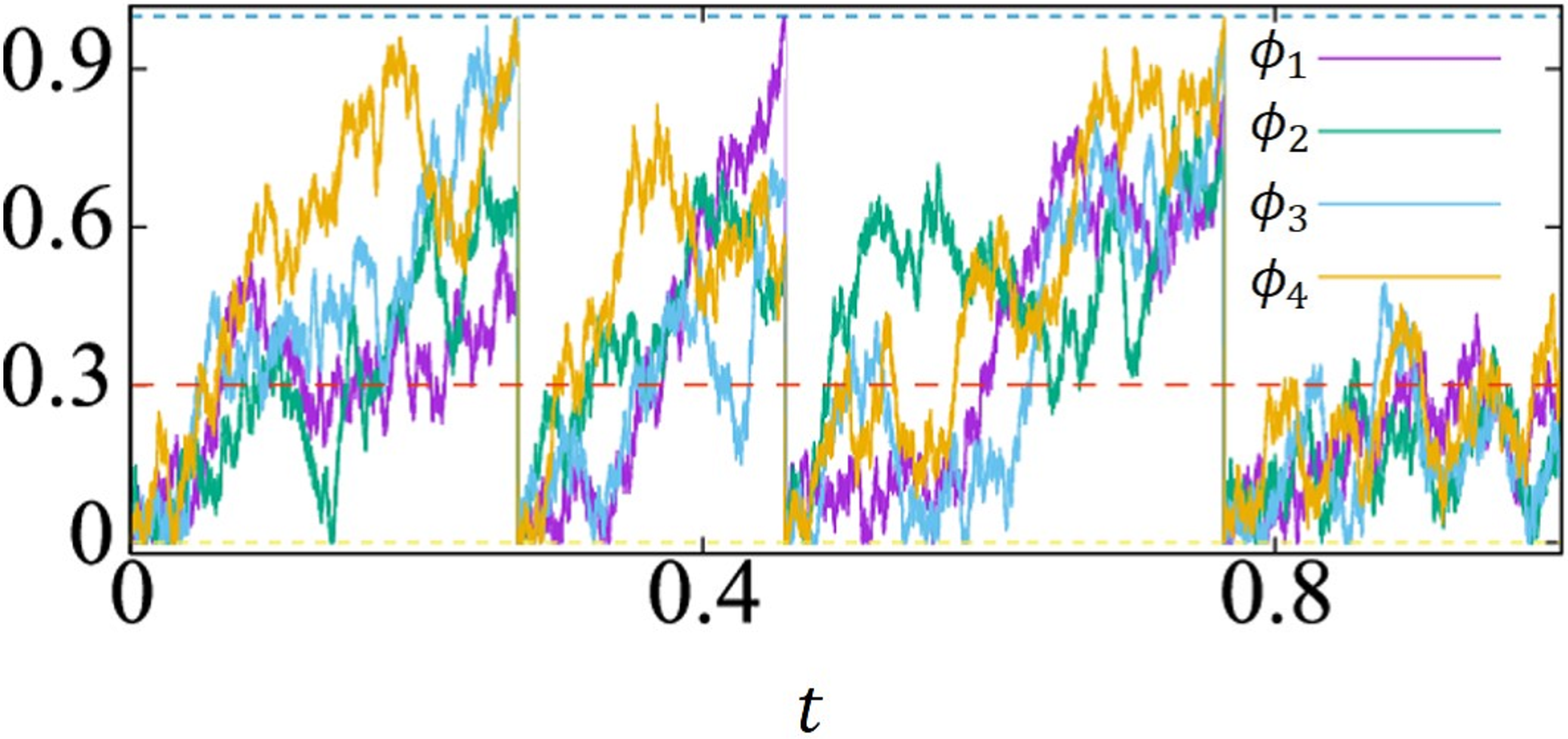}
     \put(-4,48){\small \bf (a)}
     \end{overpic}
 \hspace*{0.5cm}
%
     \begin{overpic}[width=0.43\linewidth]{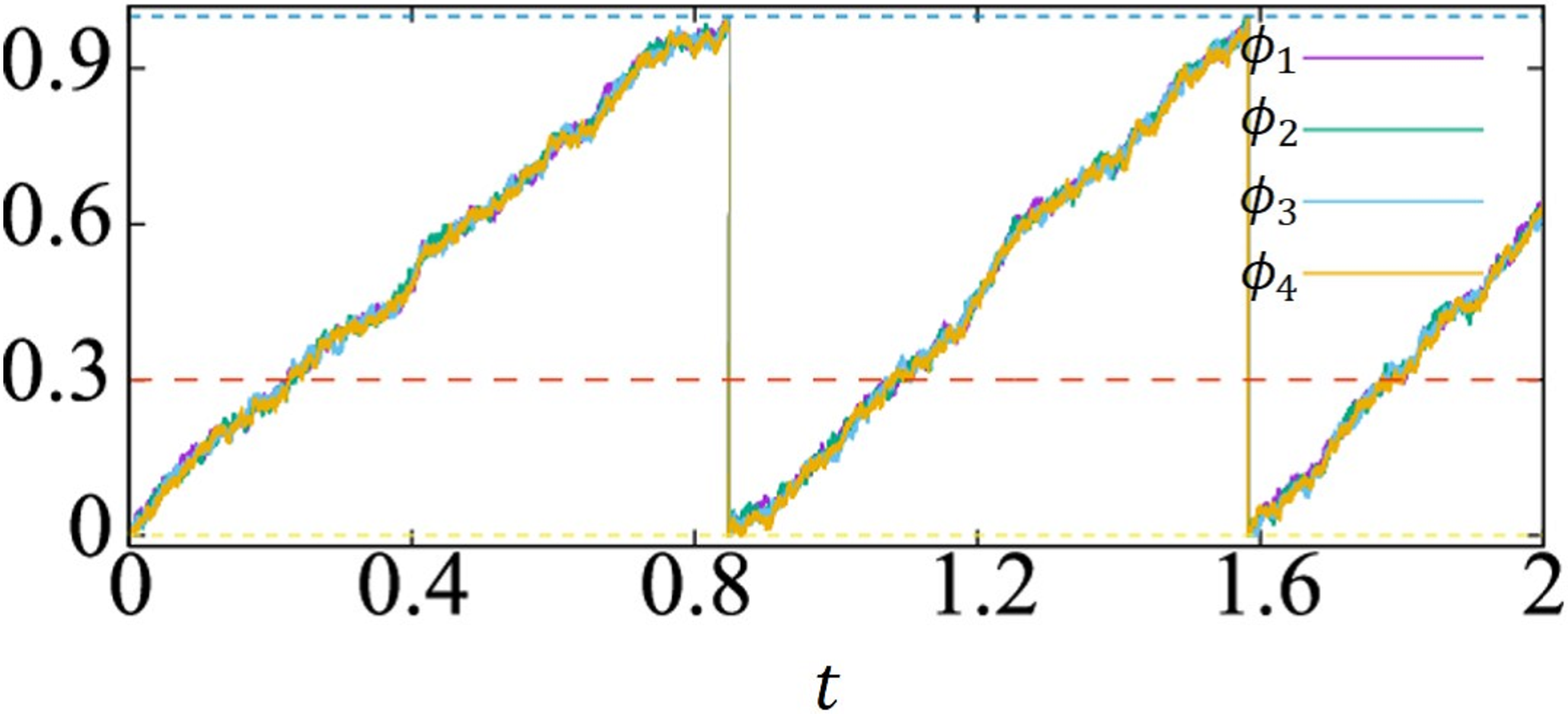}
     \put(-4,48){\small \bf (b)}
     \end{overpic}
     \begin{overpic}[width=0.43\linewidth]{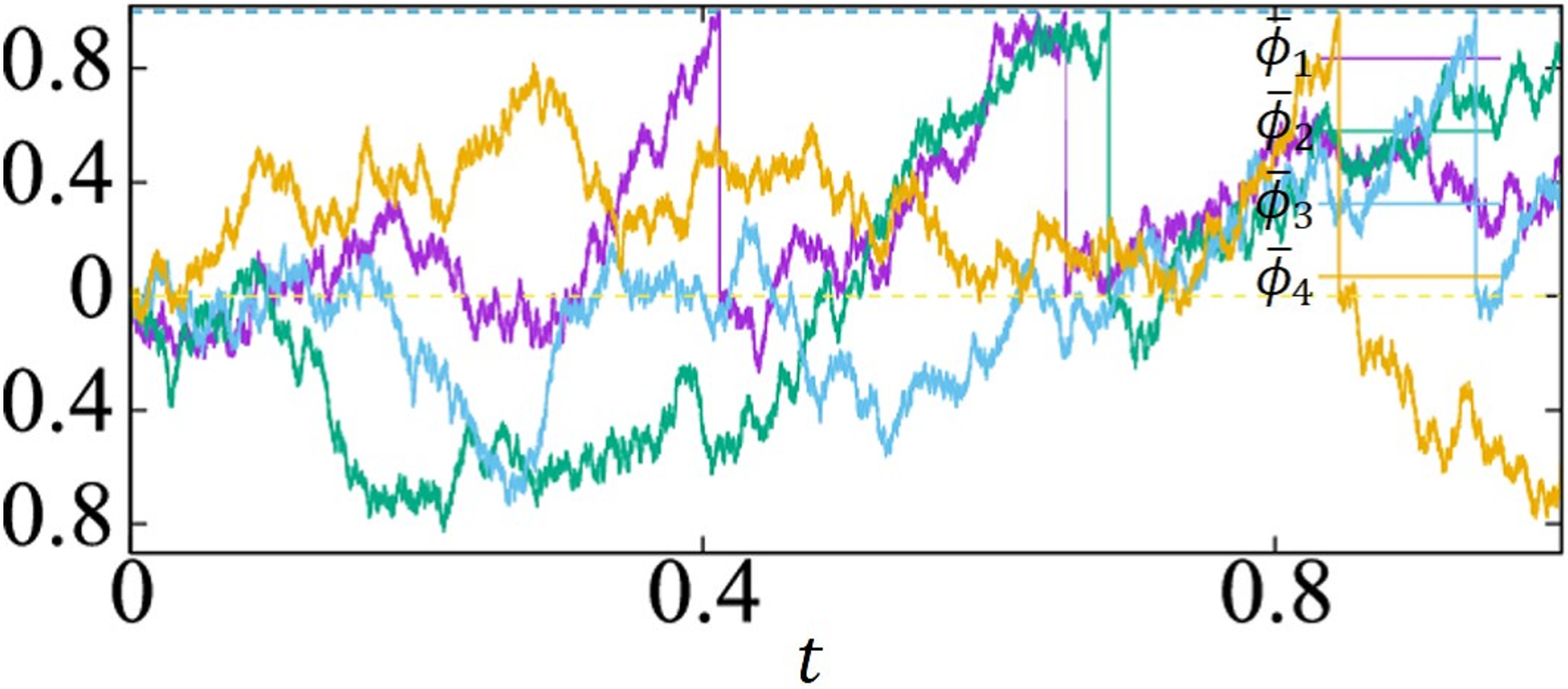}
     \put(-4,48){\small \bf (c)}
     \end{overpic}
 \hspace*{0.5cm}
%
     \begin{overpic}[width=0.43\linewidth]{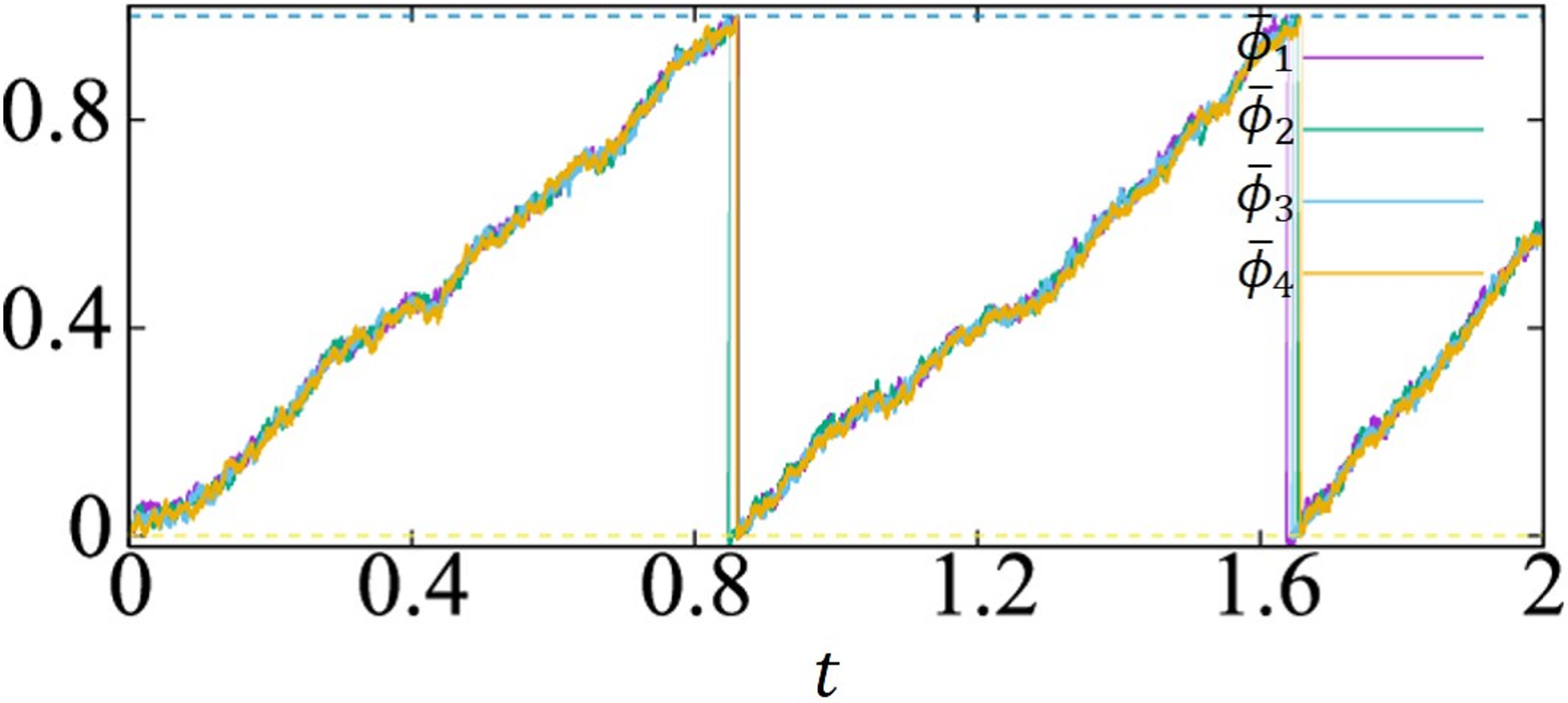}
     \put(-4,48){\small \bf (d)}
     \end{overpic}
 \end{center}
\caption{(a)(b): The realizations of $\{\phi_i\}_{i=1}^4$. (c)(d): The realizations of $\{\bar{\phi}_i\}_{i=1}^4$. The intrinsic frequency $(\mu_1, \mu_2, \mu_3, \mu_4) = (2, 1.5, 1, 0.5)$. The refractory threshold $B_i = 0.3$ $(i=1,2,3,4)$ (see the dot lines in (a)(b)).  (a)(c): $\sigma_i =1$, $A_{i,j} = 1$; (b)(d): $\sigma_i =0.2$, $A_{i,j} = 8$ ($i=1,2,3,4$, $j \in \mathcal{N}_i$). }
\label{fig:4-c}
\end{figure}

As with Section 4, we also pay attention to the conventional Kuramoto model. 
Let $\bar{\phi}_i$ denote the phase of cell $i$ $(i=1,2,\cdots,N)$ without the reflective boundary and induced beating, 
which satisfies  
\begin{subequations}\label{eq:N-c-conv}
\begin{align}
& d \bar{\phi}_i(t) = \mu_i dt + \sum_{j \in \mathcal{N}_i} A_{i,j} \sin(2\pi (\bar{\phi_j} - \bar{\phi_i})) dt + \sigma_i d W_i(t), \\
& \bar{\phi}_i(0) = 0.   
\end{align}
\end{subequations}

In Figure~\ref{fig:4-c} (a)(b) and (c)(d), we plot two trajectories of $\{\phi_i(t)\}_{i=1}^4$ and $\{\bar{\phi}_i(t)\}_{i=1}^4$ respectively for different parameters. 
For $\sigma_i =1$, $A_{i,j} = 1$, the proposed model \eqref{eq:N-c} has the synchornization due to the induced beating (see Figure~\ref{fig:4-c}(a)), 
and the noise effect at $\phi_i=0$ has been inhibited by reflective boundary (the irreversibility of beating). 
However, there is no synchronization for the conventional model \eqref{eq:N-c-conv} (see Figure~\ref{fig:4-c}(c)). 

As discussed in Section~4.2, we have to choose large enough reaction coefficients and sufficiently small noise strength to expect the ``approximated'' synchronization occurs for the conventional model. 
For $\sigma_i =0.2$, $A_{i,j} = 8$, 
our model \eqref{eq:N-c} gives a very stable synchronization (see Figure~\ref{fig:4-c}(b)), 
where the role of the reflective boundary and induced beating can be negligible, such that the solution behaviors of \eqref{eq:N-c} and the conventional model \eqref{eq:N-c-conv} (see Figure~\ref{fig:4-c}(d)) are quite similar.   

In bio-experiments, the fluctuation of the synchronized beating intervals reduces as the network size increases. 
Since both two models \eqref{eq:N-c} and \eqref{eq:N-c-conv} are quite similar when the stable synchronization happens (Figure~\ref{fig:4-c}(b)(d)), 
from now on, 
we shall pay attention to the conventional model \eqref{eq:N-c-conv}, 
and derive the CV of beating interval via a similar methodology to Section~4.2. 

We assume that all the cells beat almost simultaneously and ignore the tiny difference between the beating time of each $\bar{\phi}_i$.  
Then, analogously to \eqref{eq:syn-phi-i-pre}, 
taking the expected beating interval $\tau$ as the synchronized beating interval, 
we introduce the synchronized solution $\{\phi_i^{\text{syn}}\}_{i=1}^N$: 
\begin{equation}\label{eq:syn-phi-i-pre-N}
\phi_i^{\text{syn}} (t) = \mu^{\text{syn}} t + \psi_i^{\text{syn}}, \quad i=1,2,\cdots, N,  
\end{equation}
where $\mu^{\text{syn}} := 1/\tau$ represents the intrinsic frequency of synchronization, 
and $\psi_i^{\text{syn}}$ satisfies 
\[
\sum_{j \in \mathcal{N}_i}A_{i,j}\sin(2\pi(\psi_j^{\text{syn}} - \psi_i^{\text{syn}} )) = \mu_i - \mu^{\text{syn}}, \quad i=1,2, \cdots, N.
\]
For the conventional model, because the reaction term is $\sin(2 \pi (\bar{\phi}_j- \bar{\phi}_i))$
it makes no difference that we remove the setting ``the phase jumps to $0$ when approaching to $1$'' and define the $k$-th beating time of cell $i$ as the passage time that $\bar{\phi}_i$ reaches $k$ (see Remark~\ref{rk:no-jump}). 

We consider the case that $\sigma_i = \sigma$, $A_{i,j} = A_{j,i}$ for all $i,j \in \{1,2,\cdots, N\}$. 
Via a similar approach to \eqref{eq:CV-i-app-0}--\eqref{eq:xi}, 
\cblue{we shall calculate 
\begin{equation}\label{eq:CV-N-def}
(\mathbf{CV})^2 :=\frac{1}{N} \sum_{i=1}^N (\mathbf{CV}_i)^2 =  \frac{1}{N} \sum_{i=1}^N \lim_{t \rightarrow \infty}  \mathbf{E}[ (\xi_i(t+\tau) - \xi_i(t))^2] =  \frac{1}{N}  \lim_{t \rightarrow \infty}  \mathbf{E}\left[ \sum_{i=1}^N |\xi_i(t+\tau) - \xi_i(t)|^2 \right],
\end{equation}
}
where $\xi_i := \bar{\phi}_i - \phi_i^{\text{syn}}$ satisfies 
\begin{subequations}\label{eq:xi-N}
\begin{align}
& d \xi_i(t) = \sum_{j \in \mathcal{N}_i} b_{ij}(\xi_j(t) - \xi_i(t)) dt + \sigma_idW_i(t), \label{eq:xi-N-a} \\
& \xi_i(0) = \xi_i^0, \label{eq:xi-N-b}
\end{align}
\end{subequations}
with $b_{ij} := A_{i,j} \cos(2\pi(\psi_j^{\text{syn}} - \psi_i^{\text{syn}} ))$ and $\xi_i^0 = - \psi_i^{\text{syn}}$. 
Here, \cblue{we assume that $A_{i,j} > 0$ and $|\phi_j^{\text{syn}}-\phi_i^{\text{syn}}|=|\psi_j^{\text{syn}}-\psi_i^{\text{syn}}| \ll 1$, 
such that $b_{ij} \approx A_{i,j} > 0$.  } 
\begin{proposition}\label{prop:CV-N}
For identical noise strength $\sigma_i = \sigma$ $(i=1,2,\cdots,N)$ and symmetry reaction coefficients $A_{i,j} = A_{j,i}$, 
the fluctuation of the synchronized beating interval is given by: 
\cblue{
\begin{equation}\label{eq:CV-N}
\frac{\mathbf{CV}}{\sqrt{\tau}} =  \frac{1}{\sqrt{N}} \sigma \sqrt{ 1 + \sum_{i=2}^N \frac{1- e^{-\tau \lambda_i}}{\tau \lambda_i} }, 
\end{equation}
}
where $\mathbf{CV}$ is defined by \eqref{eq:CV-N-def} with $\xi_i$ satisfies \eqref{eq:xi-N} ($i=1,2,\cdots,N$). 
\end{proposition}
\begin{remark}
It is known that 
\[
\lim_{N \rightarrow \infty} \frac{1}{N}\sum_{i=2}^N \frac{1- e^{-\tau \lambda_i}}{\tau \lambda_i} = \lambda_\infty,
\]
where $\lambda_\infty $ is some constant. 
Therefore, the fluctuation $\mathbf{CV}/\sqrt{\tau}$ decreases with order $O(N^{-\frac{1}{2}})$ when $N$ is not so large,  
and converges to the constant $\sigma \lambda_\infty^\frac{1}{2}$ as $N \rightarrow \infty$, 
which has been confirmed by numerical simulation (see Figure~\ref{fig:cellnetwork} (b)(c)). 
Proposition~\ref{prop:CV-N} is similar to the result of \cite{Kori}.   
However, we emphasize that we establish a new analysis with more rigorous and precious mathematical argument using the stochastic calculus.  
\end{remark}
\begin{proof}[Proof of Proposition~\ref{prop:CV-N}]
\cblue{Setting the notations 
\[
\bm{\xi} = [\xi_i], \quad  \bm{\xi}^0  = [\xi_i^0], \quad \bm{W} = [W_i], \quad \bm{B} = [-b_{ij}], \quad \bm{\sigma} = \text{diag}(\sigma_1^2, \cdots, \sigma_N^2)  
\]
(here $b_{ii} = -\sum_{j \in \mathcal{N}_i} b_{ij}$ and $b_{ij} = 0$ if $j \notin \mathcal{N}_i$), 
we see that 
\begin{equation}\label{eq:xi-bm}
d \bm{\xi} =  - \bm{B} \bm{\xi} dt + \bm{\sigma} d \bm{W}(t), \quad \bm{\xi}(0) = \bm{\xi}^0, 
\end{equation}
which implies 
\[
\bm{\xi}(t) = e^{-t \bm{B}} \bm{\xi}^0 + \int_0^t e^{-(t-s) \bm{B}} \bm{\sigma} d\bm{W}(s).  
\]
Setting $|\bm{\xi}|^2 = \sum_{i=1}^N \xi_i^2$, from
\[
\bm{\xi}(t+\tau) - \bm{\xi}(t) = [e^{-(t+\tau) \bm{B}} - e^{-t \bm{B}} ]\bm{\xi}^0 + \int_0^t [e^{-(t+\tau-s) \bm{B}}  - e^{-(t-s) \bm{B}}] \bm{\sigma} d\bm{W}(s) + \int_t^{t+\tau} e^{-(t+\tau-s) \bm{B}} \bm{\sigma} d\bm{W}, 
\]
we obtain  
\begin{equation}\label{eq:E-xi-N-1}
\begin{aligned}
\mathbf{E}[ |\bm{\xi}_i(t+\tau) - \bm{\xi}_i(t)|^2] = & |e^{-t\bm{B}} (e^{-\tau \bm{B}} - \bm{I}) \bm{\xi}^0 |^2 + \int_0^t |e^{-(t-s)\bm{B}} (e^{-\tau \bm{B}} - \bm{I}) \bm{\sigma}|^2~ds \\
&  + \int_t^{t+\tau} |e^{-(t+ \tau -s)\bm{B}} \bm{\sigma}|^2 ~ds,
\end{aligned}
\end{equation}
where we have used the fact that the expectation of It\^{o}'s integral is zero and $\{W_i\}$ are independent Brownian motion. 
For $\sigma_i =\sigma$ and $A_{i,j} = A_{j,i} > 0$ ($i,j = 1,2, \cdots,N$), 
$\bm{\sigma} = \sigma \bm{I}$ and $\bm{B}=[-b_{ij}]$ is symmetry, as well as $e^{-t \bm{B}}$ and $e^{-t \bm{B}} - \bm{I}$. 
Hence, we calculate as 
\begin{equation}\label{eq:E-xi-N-2}
\begin{aligned}
|e^{-(t-s)\bm{B}} (e^{-\tau \bm{B}} - \bm{I}) \bm{\sigma}|^2 & = \sum_{i,j=1}^N | [e^{-(t-s)\bm{B}} (e^{-\tau \bm{B}} - \bm{I}) \bm{\sigma}]_{ij}|^2 = \text{tr}(e^{-2(t-s)\bm{B}} (e^{-\tau \bm{B}} - \bm{I})^2 \sigma^2 ) \\
& = \sum_{i=1}^N \sigma^2 (e^{-\tau \lambda_i} - 1 )^2 e^{-2(t-s)\lambda_i}, 
\end{aligned}
\end{equation}
\begin{equation}\label{eq:E-xi-N-2-a}
\begin{aligned}
|e^{-(t+\tau-s)\bm{B}} \bm{\sigma}|^2 & = \sum_{i,j=1}^N | [e^{-(t+\tau-s)\bm{B}}  \bm{\sigma}]_{ij}|^2 = \text{tr}(e^{-2(t+\tau-s)\bm{B}} \sigma^2 ) = \sum_{i=1}^N \sigma^2 e^{-2(t+\tau-s)\lambda_i}, 
\end{aligned}
\end{equation}
where $[\bm{C}]_{ij}$ denotes the $(i,j)$ component of matrix $\bm{C}$, 
and $\{ \lambda_i \}_{i=1}^N$ the eigenvalues of $\bm{B}$. 
}

\cblue{In view of  $b_{ii} = - \sum_{j \neq i} b_{ij}$, $b_{ij} > 0$ for $j \in \mathcal{N}_i$ and $b_{ij} = 0$ for $j \notin \mathcal{N}_i$, one can validate that the eigenvalues of $\bm{B}=[-b_{ij}]$ satisfies: 
\begin{equation}\label{eq:lambda-i}
\lambda_N \ge \lambda_{N-1} \ge \cdots \ge \lambda_2 > \lambda_1 = 0.  
\end{equation}
}

\cblue{Substituting \eqref{eq:E-xi-N-2}, \eqref{eq:E-xi-N-2-a} into \eqref{eq:E-xi-N-1}, we obtain 
\begin{equation}\label{eq:E-xi-N-3}
\begin{aligned}
& \mathbf{E}[ |\bm{\xi}_i(t+\tau) - \bm{\xi}_i(t)|^2] =  |e^{-t\bm{B}} (e^{-\tau \bm{B}} - \bm{I}) \bm{\xi}^0 |^2 \\
& \quad  \quad  \quad  \quad + \sum_{i=1}^N \sigma^2 (e^{-\tau \lambda_i} - 1 )^2 \int_0^t e^{-2(t-s)\lambda_i}~ds 
+ \sum_{i=1}^N \sigma^2 e^{-2(t+\tau-s)\lambda_i}~ds \\
= & |e^{-t\bm{B}} (e^{-\tau \bm{B}} - \bm{I}) \bm{\xi}^0 |^2 + \sum_{i=2}^N \sigma^2 (e^{-\tau \lambda_i} - 1 )^2 \frac{1-e^{-2t\lambda_i}}{2\lambda_i} + \sigma^2\left[ \tau + \sum_{i=2}^N \frac{1-e^{-2\tau \lambda_i}}{2\lambda_i} \right]. 
\end{aligned}
\end{equation}
Let $\bm{u}_i$ be the eigenvector associated with $\lambda_i$. 
It follows from \eqref{eq:lambda-i} that  
\[
e^{-t\bm{B}} \bm{u}_1 = \bm{u}_1, \quad e^{-t\bm{B}} \bm{u}_i = e^{-t\lambda_i} \bm{u}_i \rightarrow 0  \text{ as } t \rightarrow \infty, \quad i = 2, 3,\cdots,N, 
\]
together with \eqref{eq:E-xi-N-3}, which implies 
\begin{equation}\label{eq:CV-N-a}
(\mathbf{CV})^2 =  \frac{1}{N}  \lim_{t \rightarrow \infty}  \mathbf{E}[ |\bm{\xi}(t+\tau) - \bm{\xi}(t)|^2] = \frac{1}{N} \sigma^2\left[ \tau + \sum_{i=2}^N \frac{1- e^{-\tau \lambda_i}}{\lambda_i} \right]. 
\end{equation}
}
\end{proof}
%
%
%
\section{Concluding remarks} \label{sec:C-R}
To model the (synchronized) beating of cardiac muscle cells, 
we proposed and investigated the stochastic phase equations with the irreversibility of beating (reflective boundary), induced beating and refractory. 
We also develop some new analysis of the conventional Kuramoto model. 
The application of our models to reproducing the bio-experimental results had been carried out in \cite{Hayashi}. 
This paper mainly focuses on the theoretical analysis, 
where intend to reveal the relationship between the parameters of the model and the statistic properties of the (synchronized) beating intervals. 

One interesting discovery of the single-isolated cell's model is that the distribution of beating interval has the coefficient variance with an upper bound $\sqrt{2/3}\approx 81.5\%$, 
owing to the reflective boundary. 
For two-coupled cells, although we cannot obtain the closed-form expression of the statistic properties of the synchronized beating interval for the proposed model, 
from the mathematical points of view, 
it is worth to study the partial differential systems with non-standard boundary condition and singular force associated with the expectation of beating interval and the probability density of phase. 
For the conventional Kuramoto model, 
we established some new analysis to obtain the CV of the beating intervals.  
Finally, we pay attention to investigate the size-dependent fluctuation of the synchronization for $N$-cells network. 
 
\cmag{We mention some possible modifications and extensions for the proposed models, for example,   
the phase-dependent noise strength $\sigma(\phi)$ with $\sigma(0) = 0$, 
the non-interaction with other cells during refractory (i.e., $A_{i,j} = 0$ for $0 \le \phi_i \le B_i$), 
the irreversibility for both $\phi=0$ and $\phi=\phi_0$, and so on.  
Moreover, for large-size network, to model the propagation of the potential action (beating) of heart tissue, 
one can introduce a tiny time-delay $\eta $ ($\eta \ll 1$) of the induced beating, 
that is if cell $i$ beat spontaneously at time $t$ and the neighboring cells are our of refractory, 
then the neighboring cells are induced to beat at time $t + \eta$. }
%
%
%
\section*{Acknowledgments}
The authors would like to thank Kenji Yasuda for valuable comments.
A part of this work is supported by Core Research for Evolutional 
Science and Technology (CREST) of the Japan Science and Technology 
Agency (JST), Japan, and by Platform for Dynamic Approaches to Living 
System from the Ministry of Education, Culture, Sports, Science and 
Technology, Japan.

%
\appendix
\bibliographystyle{plain} 
\bibliography{20160921}





\end{document}